\documentclass{article}[11]
\usepackage{fullpage}
\usepackage{amsmath}
\usepackage{amssymb}
\usepackage{amsthm}
\usepackage{ifpdf}
\usepackage{algorithmic}
\usepackage{subfig}
\usepackage{wrapfig}
\usepackage{cite}

\ifpdf

  \usepackage[pdftex]{epsfig}
  \usepackage[pdftex]{hyperref}

\else

    \usepackage[dvips]{epsfig}
    \newcommand{\href}[2]{#2}

\fi

\theoremstyle{definition}
\newtheorem{theorem}{Theorem}[section]
\newtheorem{lemma}[theorem]{Lemma}
\newtheorem{definition}[theorem]{Definition}
\newtheorem{observation}[theorem]{Observation}

\newtheorem{corollary}[theorem]{Corollary}

\newcommand{\rank}{{\rm rank}}

\title{Self-Assembly with Geometric Tiles}
\author{
  Bin Fu%
    \thanks{Department of Computer Science, University of Texas - Pan American,
      \protect\url{binfu@cs.panam.edu}}
\and
  Matthew J. Patitz%
    \thanks{Department of Computer Science, University of Texas - Pan American,
      \protect\url{mpatitz@cs.panam.edu}}
\and
  Robert T. Schweller%
    \thanks{Department of Computer Science, University of Texas - Pan American,
      \protect\url{schwellerr@cs.panam.edu}}
\and
  Robert Sheline%
    \thanks{Department of Computer Science, University of Texas - Pan American,
      \protect\url{b.sheline@gmail.com}}
}
\date{}

\begin{document}

\maketitle

\begin{abstract}
In this work we propose a generalization of Winfree's abstract Tile Assembly Model (aTAM) in which tile types are assigned rigid shapes, or geometries, along each tile face.  We examine the number of distinct tile types needed to assemble shapes within this model, the temperature required for efficient assembly, and the problem of designing compact geometric faces to meet given compatibility specifications.  Our results show a dramatic decrease in the number of tile types needed to assemble $n \times n$ squares to $\Theta(\sqrt{\log n})$ at temperature 1 for the most simple model which meets a lower bound from Kolmogorov complexity, and $O(\log\log n)$ in a model in which tile aggregates must move together through obstacle free paths within the plane.  This stands in contrast to the $\Theta(\log n / \log\log n)$ tile types at temperature 2 needed in the basic aTAM.  We also provide a general method for simulating a large and computationally universal class of temperature 2 aTAM systems with geometric tiles at temperature 1.  Finally, we consider the problem of computing a set of compact geometric faces for a tile system to implement a given set of compatibility specifications.  We show a number of bounds on the complexity of geometry size needed for various classes of compatibility specifications, many of which we directly apply to our tile assembly results to achieve non-trivial reductions in geometry size.

%

\end{abstract}

\setcounter{page}0
\thispagestyle{empty}
\clearpage

\section{Introduction}

The stunning diversity of biological tissues and structures found in nature, including examples such as signaling axons stretching from neurons, powerfully contracting muscle tissue, and specifically tailored coats protecting viral payloads, are composed of basic molecular building blocks called proteins.  These proteins, in turn, are assembled from an amazingly small set of only around 20 amino acids.  So how is it that so much structural and functional variety can be derived from so few unique components?  The simplified answer is ``geometry''.  Essentially, a protein's function is determined by its $3$-dimensional shape, or geometry.  The exact sequence of amino acids which compose a protein (along with environmental influences such as temperature and pH levels) determine how that particular string of amino acids will fold into a protein's characteristic $3$-dimensional structure.  However, as simple as it may sound, the resulting geometries are often extremely complex, and predicting them has proven to be computationally intractable.  It is from such geometrically intricate structure that nearly all of the complexity of life as we know it arises.

Scientists and inventors have always recognized nature as providing invaluable examples and inspiration, and as for many other fields, this is also true for the study of artificial self-assembling systems.  Self-assembling systems are systems in which sets of relatively simple components begin in disconnected and disorganized initial states, and then spontaneously and autonomously combine to form more complex structures.  Self-assembling systems are pervasive in nature, and their power for creating intricate structures at even the nano-scale have inspired researchers to design artificial systems which self-assemble. One such productive line of research has followed from the introduction of the Tile Assembly Model (TAM) by Winfree in ~\cite{Winf98}.  As a basic model, the TAM has proven powerful, providing a basis for laboratory implementations \cite{SchWin04,SchWin07,ReiSahYin04,MaoLabReiSee00,WinLiuWenSee98,CheSchGoeWin07,BarSchRotWin09,LaWiRe99} as well as copious amounts of theoretical work \cite{Luhrs08,SoloveichikCW08,Winfree06,CookFuSch11,SFTSAFT,USA,ChaGopRei09,SolWin07}.  However, in this work, we have once again looked to the guidance provided by nature, this time in terms of the power and importance of the geometric complexity of the components of self-assembling systems, to extend the TAM in an attempt to harness that power.

\subsection{Overview}

    \begin{figure}[htbp]
    \centering
    \includegraphics[width=.35\linewidth]{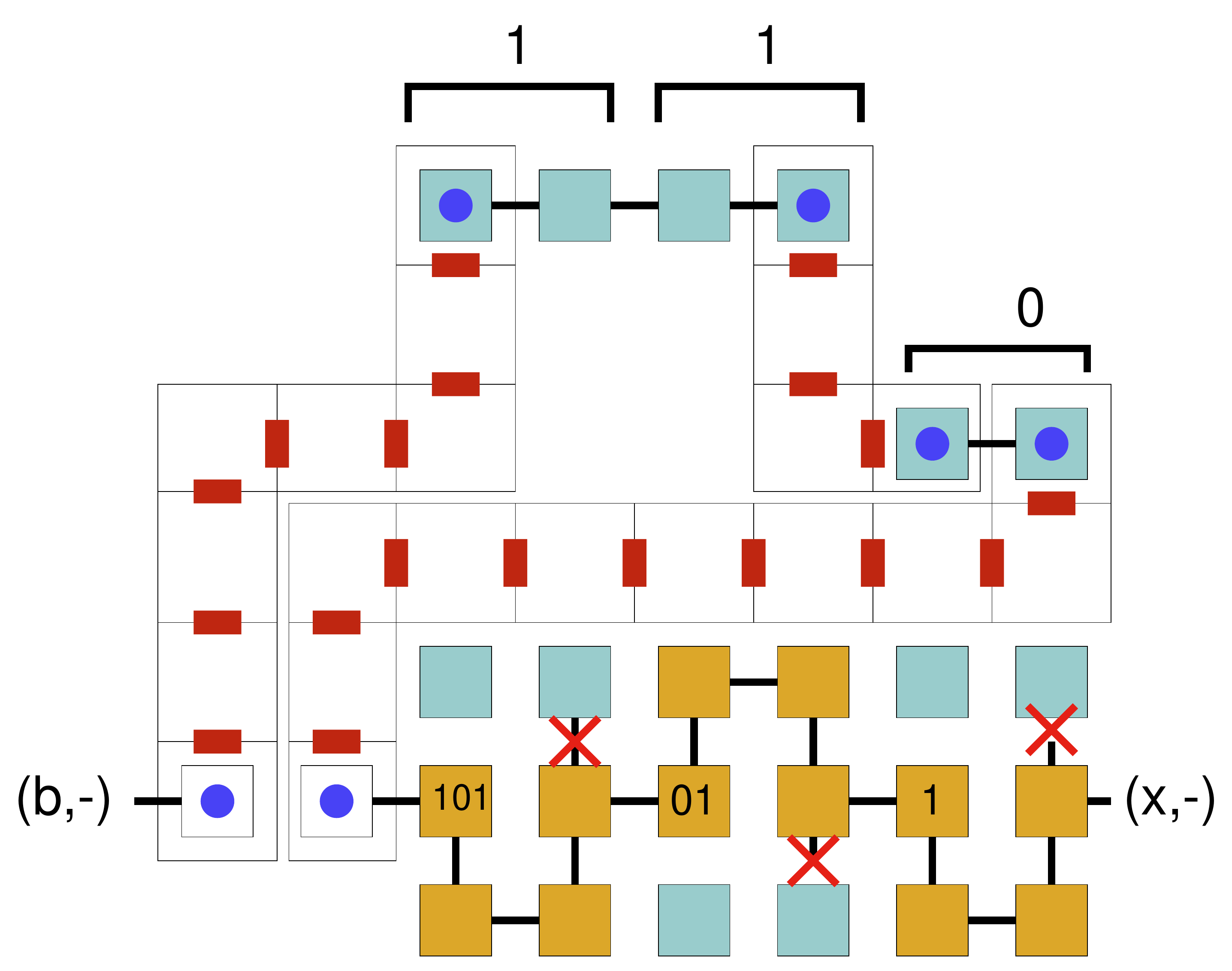}
    \includegraphics[width=.31\linewidth]{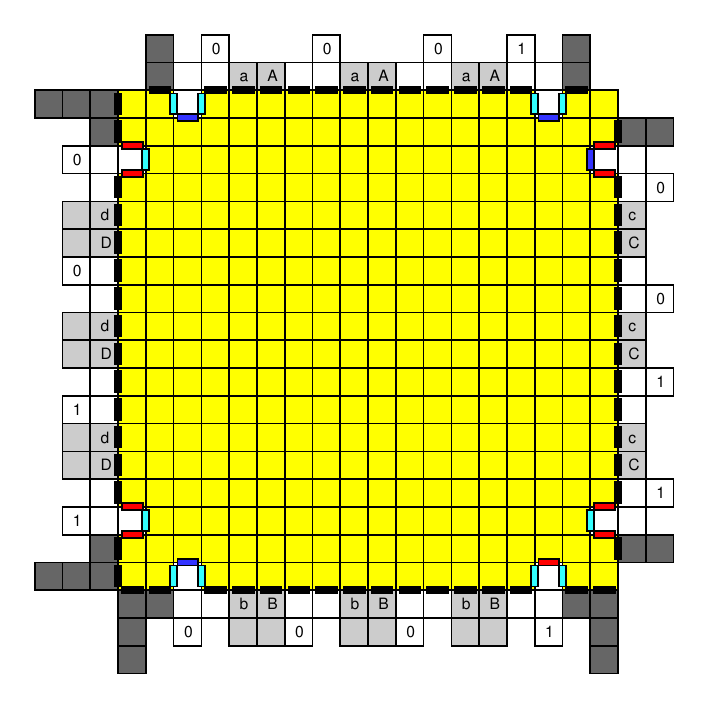}
    \includegraphics[width=.31\linewidth]{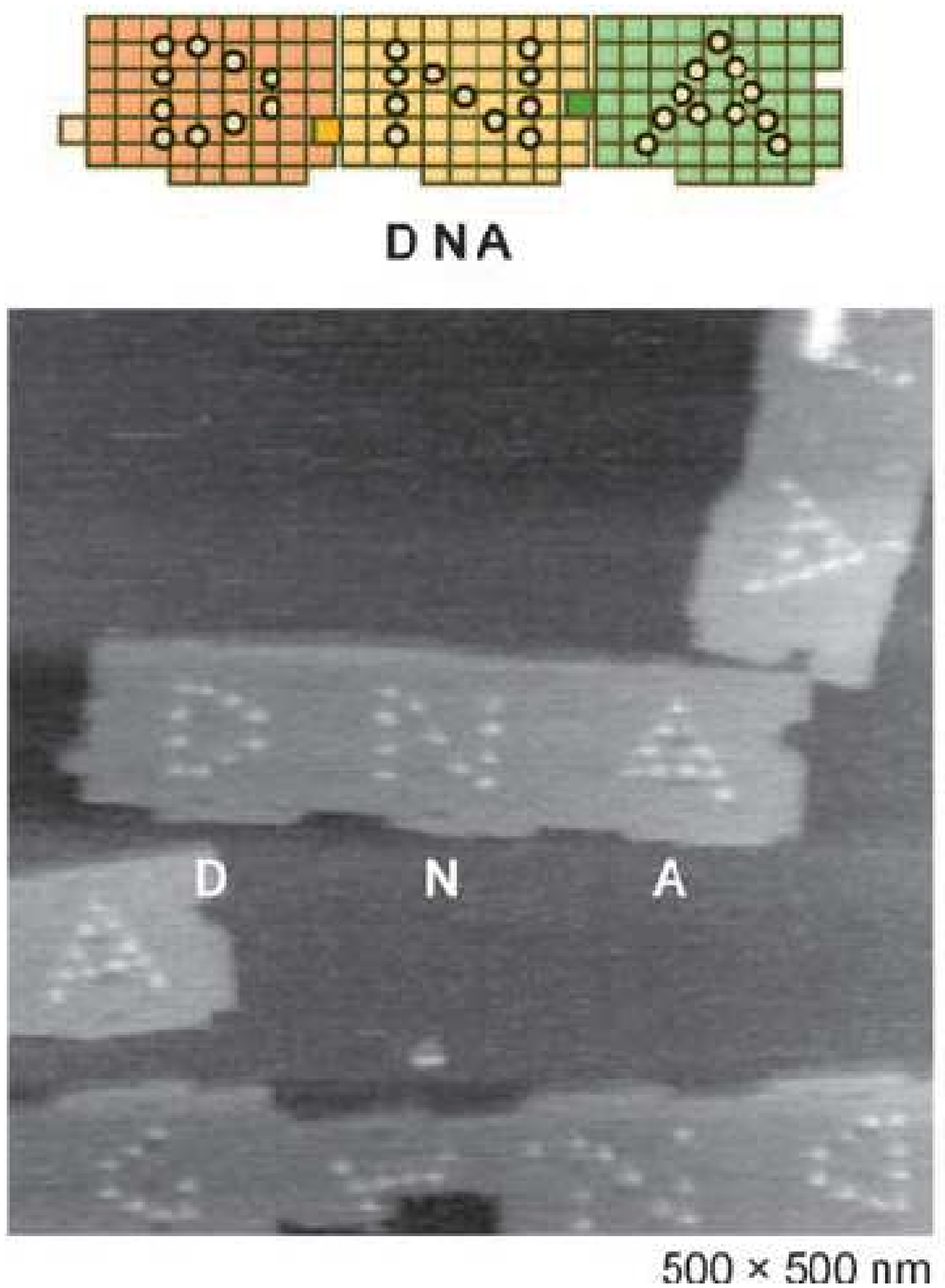}
    \caption{The use of jigsaw faced macro tiles for self-assembly is emerging in both theoretical and experimental work.  This figure contains three separate recent examples.  The first figure depicts the theoretical technique of encoding binary strings within the geometry of tile growth into the third dimension, as seen by small blue tiles in this figure~\cite{CookFuSch11}.
     The second figure depicts a macro tile assembled using staged assembly from smaller tile types~\cite{DDFIRSS07}.  Finally, the third figure depicts the experimental work of~\cite{EndoSugitaKatHidSug10} in which a
     jigsaw geometry on the face of tiles is created with the DNA origami technique. }
    \label{fig:hindrance}
    \end{figure}

  We introduce a generalization of the abstract Tile Assembly Model (aTAM) in which tile types are assigned rigid shapes, or geometries, along each tile face.  This model is motivated by the plausibility of implementing novel sophisticated nanoscale shapes with technology such as DNA origami~\cite{RothOrigami}.  We show that this model permits substantially greater efficiency in terms of tile type complexity when compared to assembling shapes in the basic temperature 2 aTAM.  Furthermore, these efficiency improvements hold even at temperature 1.

\subsection{Results}\label{sec:overview}
\begin{table*}\small
\centering
\tabcolsep=0.5\tabcolsep
\begin{tabular}{l|c|c|c|c|c|}
\multicolumn{1}{c|}{\textbf{$n\times n$ square}} & \multicolumn{2}{|c|}{\textbf{Tile Types}} & \textbf{Temperature} & \textbf{Geometry Size}
\\ \hline
ATAM (previous work) \cite{RotWin00,AdChGoHu01} & \multicolumn{2}{|c|}{$\Theta(\log n/\log\log n)$} & 2 & -
\\ \hline
GTAM (Thms. \ref{thm:gtamUpper},\ref{thm:gtamLower}) & \multicolumn{2}{|c|}{$\Theta(\sqrt{\log n})$} & 1 & $O(\sqrt{\log n})$
\\ \hline
2GAM (Thm. \ref{thm:2GAM-squares}) & \multicolumn{2}{|c|}{$O(\log\log n)$} & 2 & $O(\log n \log\log n)$
\\ \hline
\multicolumn{2}{c}{}\\

\multicolumn{1}{c|}{\textbf{Zig-zag simulation}} & \textbf{Tile Type Scale} & \textbf{Glues} & \textbf{Temperature} & \textbf{Geometry Size}
\\ \hline
Theorem \ref{thm:zigzagSimulation} & \multicolumn{1}{|c|}{$O(1)$} & $O(\sigma_w)$ & 1 & $\log\sigma_n + \log\log\sigma_n +O(1)$
\\ \hline
Theorem \ref{thm:zigzagSimulation1Glue} & \multicolumn{1}{|c|}{$O(1)$} & 1 & 1 & $\log\sigma + \log\log\sigma +O(1)$
\\ \hline
\multicolumn{2}{c}{}\\

\end{tabular}

\begin{tabular}{l|c|c|c|c|c|}
\multicolumn{1}{c|}{\textbf{Compact Geometry Design}} & \multicolumn{2}{|c|}{\textbf{Geometry Size}} & \textbf{Run Time}
\\ \hline
Random Matrix (Thms. \ref{thm:mostMatrices}) &
\multicolumn{2}{|c|}{$L(M)=\Theta(n)$} &
\\ \hline
Diagonal 1's (Cor. \ref{cor:diagOne}) & \multicolumn{2}{|c|}{$L(M)=n$} &
$O(n^2)$
\\ \hline
Diagonal 0's (Cor. \ref{cor:diagZero}) & \multicolumn{2}{|c|}{$\log n + 1\le L(M)\le \log
n + \log\log n$}   & $O(n^2)$
\\ \hline
Ind. Sub. Matrices (Thm. \ref{thm:indSubMatrices}) &
\multicolumn{2}{|c|}{$L(M)=\sum{L(M_i)}$} &
\\ \hline
Ind. Sub. Matrices (Thm. \ref{thm:indSubMatrices}) &
\multicolumn{2}{|c|}{$L(M)\le (1+\epsilon)\max(\min(m_i,n_i))$} & \\
& \multicolumn{2}{|c|}{$+ O(\log(n+m))$} & $O((m+n)^{3})$
\\ \hline
Ind. Sub. Matrices (Thm. \ref{thm:indSubMatrices}) &
\multicolumn{2}{|c|}{$\max(L(M_i))\le L(M)\le$} &\\
& \multicolumn{2}{|c|}{$(1+\epsilon)\max(L(M_i)) + O(\log(m+n))$} &
$O(2^{\max(L(M_i))\min(m,n)}(m+n)^{3})$
\\ \hline
\multicolumn{1}{c}{}\\

\multicolumn{1}{c|}{\textbf{Bar to Bump Reduction}} & \multicolumn{2}{|c|}{\textbf{Geometry Size}}
\\ \hline
Theorem \ref{thm:twoFuncUpper},\ref{thm:twoFuncLower} & \multicolumn{2}{|c|}{$n$}
\\ \hline

\end{tabular}
\caption{\footnotesize Summary of our Results. $\sigma$ denotes the number of distinct glues of a tile system to be simulated, with $\sigma_n$ and $\sigma_w$ denoting only the number of north/south and west/east glue types respectively.  $L(M)$ is the size of the smallest geometry that can satisfy a binary $n\times m$ compatibility matrix $M$.}
\label{table:summary}
\end{table*}

The abstract tile assembly model (aTAM)~\cite{Winf98}, as well as many of the nanoscale self-assembly models spawned by it, feature single stranded DNA sequences as the primary mechanism for decision making.
This commonality applies to weak systems such as deterministic temperature-1 assembly, as well as stronger ones that rely on higher temperatures or stochastic methods.
Since it is known that DNA strands are capable of hybridizing with sequences other than their exact Watson-Crick compliments, it is therefore reasonable to consider a tile assembly model
in which one glue can potentially bond with an arbitrary subset of the other glues, with possibly differing strengths. Aggarwal~et.~al.~\cite{AGKS05g} have shown that such a \emph{non-diagonal}
 glue function allows for significant efficiency gains in terms of the numbers of unique tiles used to assemble a target shape. Despite this potentially promising result,
 it is also true that designing non-specific hybridization pairs, while possible, is severely limited in a practical sense, and would likely introduce a potential for error in a much greater
 sense than is already present in laboratory experiments.


\begin{figure}[h]
\centering
  \includegraphics[width=4.0in]{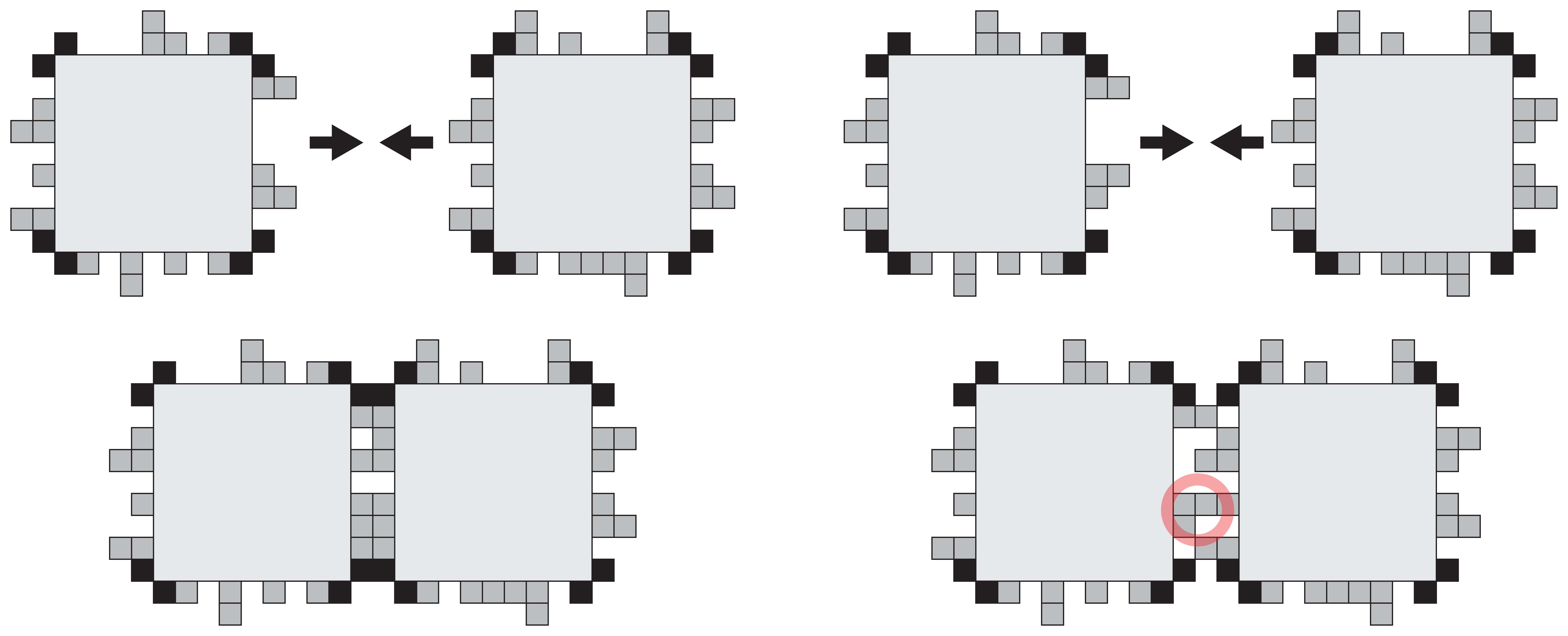}
    \caption{\label{fig:geometric-tile-transform}Examples of geometric tiles.  Note that only the black portions on the corners are binding surfaces with glues, while the ``teeth'' in between provide potential geometric hindrance.  Left: Compatible tiles.  Right: Incompatible tiles (colliding teeth, which prevent the glue pads from coming together, are circled).}
\end{figure}

If non-specific binding is impractical or impossible to implement, but powerful in theory, the question remains: are there any other mechanisms by which this power can be realized? One possible answer to this question is motivated by advances in DNA origami\cite{RothOrigami,EndoSugitaKatHidSug10} in which DNA strands can be folded into blocks with semi-rigid jig-saw faces (see the rightmost image in Figure~\ref{fig:hindrance}).  In this work we introduce a generalization of the aTAM in which tile faces are given some rigid shape (which we hereon refer to as geometry).  As suggested in Figure~\ref{fig:geometric-tile-transform}, the \emph{geometric hindrance} which can be provided by this geometry is capable of simulating non-diagonal glue functions by creating a set of compatible and non-compatible faces.  We show that this new model realizes much of the power of non-specific hybridization.  Among our results, we show that $n\times n$ squares can be assembled in $\Theta(\sqrt{\log n})$ distinct tile types, which meets an information theoretic lower bound for the model and improves what is possible without geometric tiles from $\Theta(\log n / \log\log n)$ (see \cite{RotWin00}).  In addition, this tile efficient construction requires only a temperature threshold of 1, thus showing this model can mimic both non-specific glue functions and temperature 2 self-assembly simultaneously.

Next, we show that temperature-$1$ systems utilizing geometry can efficiently simulate a powerful class of temperature-$2$ aTAM systems.  This class of systems, called \emph{zig-zag} systems, is capable of simulating arbitrary Turing machines and therefore universal computation.  Furthermore, the simulation performed using geometric tiles is efficient in that it requires no increase in tile complexity (i.e. the number of unique tile types required) or in the size of the assembly.  This is especially notable due to the fact that it is conjectured (although currently unproven) that temperature-$1$ systems in the aTAM are not computationally universal (see \cite{FuSch09,jLSAT1} for more discussion about temperature-$1$ assembly in the aTAM).

While tile geometries provide a method for greatly reducing the tile complexity required to build squares in a seeded model similar to the aTAM (i.e. one in which tiles can only combine with a growing assembly one at a time), our next result holds for geometric tiles considered within the 2-handed assembly model (sometimes referred to by other names) \cite{AGKS05g,DDFIRSS07,Winfree06,Luhrs08,AdlCheGoeHuaWas01,Adl00,RNaseSODA2010}.  We show that, in this model, the tile complexity required to build a square is reduced to only $O(\log \log n)$ tile types.  The construction presented utilizes the ability of 2-handed assembly to grow assemblies by the combination of sub-assemblies composed of groups of previously combined tiles, and, coupled with complex geometric patterns on the tile edges, forces assembling components to undergo intricate patterns of relative motion in order to combine with each other.  The tile geometries required are, however, complex ($O(\log n \log \log n)$) and in a $2$-dimensional model require disconnected components.  We then show a simple extension to $3$ dimensions which allows for connected components while retaining all other features.

Finally, we conduct a detailed analysis of problems related to computing
necessary patterns for tile geometries given specifications of the
desired compatibility matrices (i.e. the listings of which tile sides should be compatible and incompatible with each other), with the goal being to minimize the size of the necessary geometries (as well as the running time of the computations).  They deal with designing tile face geometries as a subset of
$Z_1\times Z_l$. Their solutions help show the feasibility and
limitations of geometric tile face designs. We show a number of lower and upper
bounds related to variants of the problem, some of which are
incorporated into the previously mentioned constructions.

\subsection{Organization of this paper}

The remainder of this paper is organized as follows.  In Section~\ref{sec:model} we describe and define the new models introduced here.  In Section~\ref{sec:seeded-results}, we present our constructions and proofs related to the self-assembly of $n \times n$ squares using $\Theta(\sqrt{\log n})$ tile types, as well as the simulation of zig-zag, temperature-$2$ aTAM systems by temperature-$1$ systems with geometric tiles.  Section~\ref{sec:2GAM-results} describes our construction which utilizes geometric tiles as well as $2$-handed assembly to self-assemble $n \times n$ squares using $O(\log \log n)$ tile types.  Additionally, there is a technical appendix which contains the majority of the proofs and construction details for the results presented in the previous sections, as well as the results related to computing compatibility matrices.

\section{Model}
\label{sec:model}
In this section we define the basic \emph{geometric tile assembly model} (GTAM) and the \emph{two-handed planar geometric tile assembly model} (2GAM). We begin with an informal description of the aTAM.  We then define the Geometric Tile Assembly Model (GTAM).  The GTAM generalizes the aTAM~\cite{Winf98} by adding a geometry to each tile face that may prevent two tiles from attaching.

\subsection{Basics}

A tile type is a unit square with four sides, each having a glue consisting of a label (a finite string) and strength (0, 1, or 2). We assume a finite set T of tile types, but an infinite number of copies of each tile type, each copy referred to as a tile. A supertile (a.k.a., assembly) is a positioning of tiles on the integer lattice $\mathbb{Z}^2$. Two adjacent tiles in a supertile interact if the glues on their abutting sides are equal. Each supertile induces a binding graph, a grid graph whose vertices are tiles, with an edge between two tiles if they interact. The supertile is $\tau$-stable if every cut of its binding graph has strength at least $\tau$, where the weight of an edge is
the strength of the glue it represents. That is, the supertile is stable if at least energy $\tau$ is required to separate
the supertile into two parts. A seeded tile assembly system (TAS) is a triple $T = (T, \tau, s)$, where T is a finite tile set, $\tau$ is the temperature, usually 1 or 2, and $s\in T$ is a special tile type denoted as the $\emph{seed}$. Given a TAS $T = (T, \tau, s )$, a supertile is producible if either it is
the seed tile, or it is the $\tau$-stable result of attaching a single tile $r\in T$ to a producible supertile. A supertile $\alpha$ is
terminal if for every tile type $r\in T$, $r$ cannot be $\tau$-stably attached to $\alpha$. A TAS is directed (a.k.a.,
deterministic or confluent) if it has only one terminal, producible supertile. Given a connected shape $X \subset \mathbb{Z}^2$,
a TAS T produces X uniquely if every producible, terminal supertile places tiles only on positions in $X$
(appropriately translated if necessary).

\subsection{Geometric Tiles and the Basic Geometric Tile Assembly Model (GTAM)}
In this paper we generalize the basic aTAM by assigning a geometric pattern to each side of a tile type along with its glue.  For each tile set in the GTAM, fix two values $w,\ell \in \mathbb{N}$.  While at a high-level we still consider tiles as occupying unit squares within the plane, in order to determine whether or not adjacent tiles are \emph{geometrically compatible} with each other, we define a \emph{tile body} to be an $\ell \times \ell$ square (see Figure~\ref{fig:geometric-tile-definition}), and we define a \emph{(tile face) geometry} to be a subset of $\mathbb{Z}_w \times \mathbb{Z}_\ell$.  A \emph{geometric} tile type consists of a tile body which has both a glue and a geometry assigned to each side.  For a tile type $t$, let $northGeometry(t)$ denote the geometry assigned to the north side of $t$.  Define $eastGeometry(t)$, $southGeometry(t)$, and $westGeometry(t)$ analogously.  Intuitively, the geometry of a tile type face represents the positions of inflexible bumps, or ``filled-in'' locations of the $w \times \ell$ rectangle, that can prevent two tiles from lining up adjacently to one another so that the rectangles of their adjacent geometries completely overlap.  Only if the $w \times \ell$ geometries on adjacent sides of two combining tiles can completely overlap so that no location contains a filled-in portion of both, can any glues on those adjacent sides interact.  Formally, we say a tile type $t$ is \emph{east incompatible} with tile type $r$ if $eastGeometry(t) \bigcap westGeometry(r) \neq \emptyset$. We define $\emph{north}$, $\emph{south}$, and $\emph{west}$ incompatibility analogously.  Seeded Geometric Tile Assembly takes place in the same manner as in the aTAM, with the added requirement that a tile type cannot be attached to a supertile at a position in which the tile type is either east, west, north, or south incompatible with another adjacent tile type in the supertile at a position west, east, south, or north, respectively, of the attachment position.  As in the original aTAM, tiles are not allowed to rotate and must always maintain their pre-specified orientation, even while moving into position to attach to an assembly.

\begin{wrapfigure}{r}{2.2in}
\vspace{-20pt}
\begin{center}
    \includegraphics[width=2.0in]{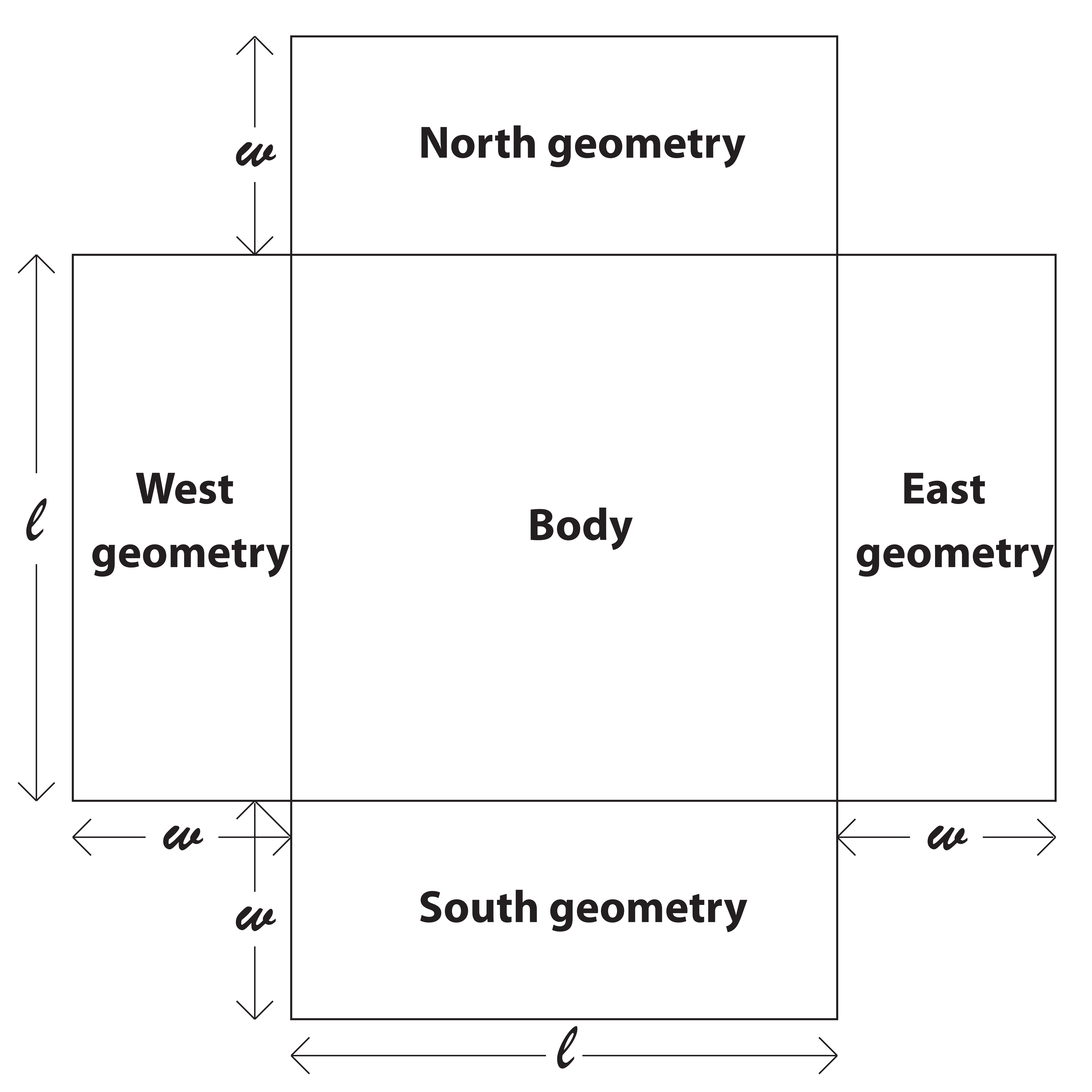} \caption{\label{fig:geometric-tile-definition} \footnotesize Definition of a geometric tile.}
\end{center}
\vspace{-20pt}
\end{wrapfigure}

\subsection{Two-Handed Geometric Tile Assembly Model}
The Two-Handed Geometric Tile Assembly Model (2GAM) extends the GTAM by allowing large assembled supertiles to attach to one another.  We further restrict the model to planar assembly in which two supertiles may only attach if there exists a collision free path for the supertiles to traverse to reach their point of connection.  In two dimensional assembly this enforces that supertiles must be able to slide into position while staying in the 2D plane.  With standard aTAM tiles, this requirement enforces that individual tiles of a supertile do not collide with individual tiles from another supertile while the supertiles shift into position.  With geometric tiles, we must also enforce that the geometries of individual tiles do not overlap with other tile geometries.

\paragraph{Informal Definition of the 2GAM}  As in the GTAM, tiles are composed of tile bodies and tile face geometries as shown in Figure~\ref{fig:geometric-tile-definition}. Within the 2GAM, two tiles may attach if 1) there exists a collision free path within the 2D plane to shift the tiles into an adjacent position in which the east (or south) geometry box of one tile exactly overlaps the west (or north) geometry box of the second tile, and 2) the east (north) and west (south) glues of each tile are equal and have strength at least $\tau$.  More generally, preassembled multiple tile supertiles may come together if there is a collision free path in which the supertiles line up to create a $\tau$-stable assembly.  The set of \emph{producible} supertiles within the 2GAM is defined recursively:  As a base case, all singleton supertiles consisting of a single tile are producible.  Recursively, for any two producible supertiles $\alpha$ and $\beta$ such that there exists a collision free path within the plane to shift $\alpha$ and $\beta$ into a $\tau$-stable configuration $\gamma$, then the supertile $\gamma$ is also producible.  The subset of producible assemblies of a 2GAM system to which no producible assembly can attach defines the \emph{terminally produced} supertiles.  Intuitively, this set represents the set of assemblies we expect to see from a system if it is given enough time to assemble, and we refer to this as the output of the system.  A 2GAM is directed (e.g., deterministic, confluent) if it has only one terminal, producible supertile. Given a connected shape $X\subseteq \mathbb{Z}^2$, a 2GAM $\Gamma$ produces $X$ uniquely if every producible, terminal supertile places tiles only on positions in $X$ (appropriately translated if necessary).

Please refer to Section~\ref{sec:formal-2GAM} for a more formal definition of the 2GAM model.  Additionally, for a discussion of the different types of tile face geometries that are possible and the classes into which they can be categorized, please see Section~\ref{sec:geometry-classes}.

%

\section{Complexities for the GTAM: Squares and $\tau=1$ Assembly}\label{sec:seeded-results}
In this section we examine the power of the GTAM in the context of efficiently building squares and simulating temperature $\tau=2$ ATAM systems at $\tau=1$.  We first show in Secton~\ref{subsec:TheTileComplexity} that the tile complexity of $n\times n$ squares in the GTAM is $\Theta(\sqrt{\log n})$ for almost all $n$ by providing an order $\sqrt{\log n}$ tile complexity upper bound construction for all $n$, and a matching information theoretic lower bound for almost all $n$.  In addition, our upper bound construction utilizes only temperature $\tau=1$.  This stands in contrast to the temperature $\tau=2$, $\Theta(\log n / \log\log n)$ tile complexity result that can be achieved in the ATAM~\cite{AdChGoHu01}.

As the square construction shows, the GTAM seems to be powerful at $\tau=1$.  In Section~\ref{subsec:SimulatingTemperature} we consider the problem of simulating $\tau=2$ ATAM systems within the GTAM, but at $\tau=1$.  We show that for a large class of temperature $\tau=2$ ATAM systems called \emph{zig-zag} systems, such a simulation is possible with no scale up in tile complexity or assembly size.  Of particular note is the fact that zig-zag systems are capable of simulating universial Turing machines, something that is conjectured to not be possible in the ATAM at $\tau=1$.

\subsection{The Tile Complexity of GTAM squares: $\Theta(\sqrt{\log n})$}\label{subsec:TheTileComplexity}
In this section we analyze the size of the smallest tile type GTAM system that uniquely assembles an $n\times n$ square.  Our first result is a construction that will assemble an $n\times n$ square using $O(\sqrt{\log n})$ tile types.  We then show that this is tight for almost all $n$ by applying an information theoretic argument to show that for almost all $n$, at least $\Omega(\sqrt{\log n})$ distinct tile types are required to uniquely assemble an $n\times n$ square.  This result stands in contrast to the $\Theta(\log n / \log\log n)$ tile complexity for building squares in the standard ATAM model, showing that the GTAM is strictly more powerful than the ATAM.  Further, our upper bound construction uses only temperature 1, while the ATAM construction requires temperature 2.

In the remainder of this section we prove the following theorems:

\begin{theorem}\label{thm:gtamUpper}
The minimum tile complexity required to assemble an $n\times n$ square in the GTAM is $O(\sqrt{\log n})$.  Further, this complexity can be achieved by a temperature $\tau=1$ system with $O(\sqrt{\log n})$ size geometry.
\end{theorem}

\begin{theorem}\label{thm:gtamLower}
For almost all integers $n$, the minimum tile complexity required to assemble an $n\times n$ square in the GTAM is $\Omega(\sqrt{\log n})$.
\end{theorem}

For the sake of brevity, we only give a high level overview of the upper bound construction and place the details in referenced appendix sections.

\subsubsection{$O(\sqrt{\log n})$ Construction Overview}
The tile system for the assembly of $n\times n$ squares in the GTAM at temperature $\tau=1$ and tile complexity $O(\sqrt{\log n})$ starts with the assembly of a roughly $2\times\log n$ rectangle ($2\times \lceil \log(\frac{n+1}{2})\rceil +2$ to be precise) that is used as a base to encode a roughly $\log n$ digit binary number ($\lceil \log(\frac{n+1}{2})\rceil$ digits to be precise).  The rectangle is efficiently built with $O(\sqrt{\log n})$ tile types by using a tile set for simulating a 2-digit, base-$\sqrt{\log n}$ counter.  Such counters are known to exist in the ATAM at $\tau=2$.  To achieve $\tau=1$ in the GTAM, we apply the transformation described in Theorem~\ref{thm:zigzagSimulation1Glue} to convert a zig-zag (see Definition~\ref{def:zigzag}) version of the $\tau=2$ counter into a $\tau=1$ GTAM system.  The tileset for the basic $\tau=2$ ATAM counter is given in Figure~\ref{fig:temp2BaseCounterWithExample}, and the $\tau=1$ GTAM version is given in Figure~\ref{fig:Temp2ToTemp1SingleTileExample}.  The details of this portion of the construction are described in Section~\ref{subsubsec:baseCounter}.

\begin{figure}[t!]
\begin{center}
    \includegraphics[width=6.0in]{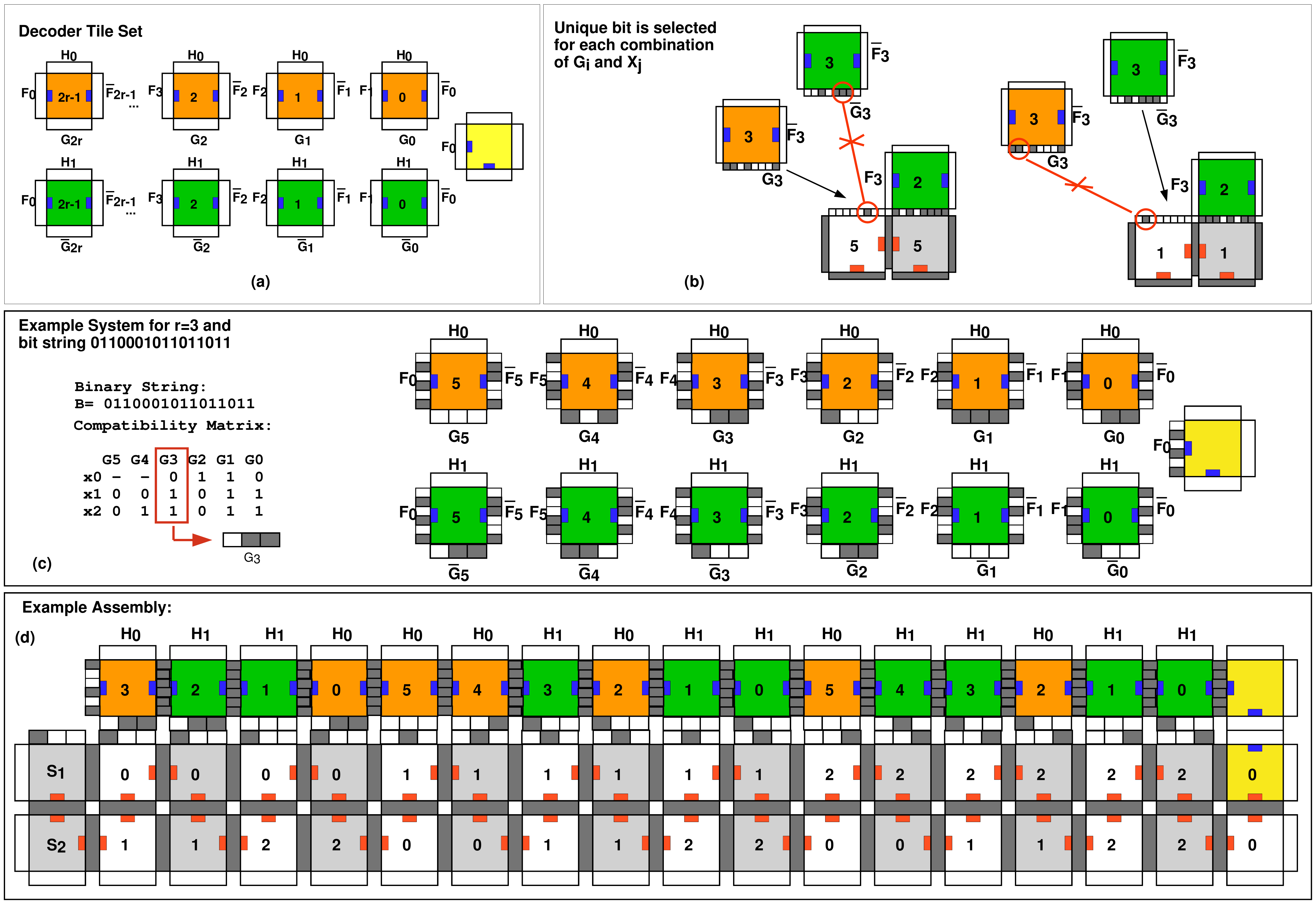} \caption{\label{fig:decoderTilesCombo} \footnotesize (a) The decoder set of tiles grows across the north face of the assembled counter from Figure~\ref{fig:Temp2ToTemp1SingleTileExample} to build a binary string of tiles by placing at each position a tile type representing either a 0-bit or a 1-bit.  (b)  For a given binary string $B=b_{2r-1}\ldots b_2 b_1 b_0$ to be assembled, the geometries $G_j$, $\bar{G_j}$, and $X_i$ are assigned such that  $G_j$ and $X_i$ are compatible if and only if $b_{2r(r-1) - 2ri +j}=0$.  (c)  For a given binary string $B=b_{2r-1}\ldots b_2 b_1 b_0$ to be assembled, the listed compatibility matrix is obtained, along with the described assignment of geometry to each tile face.  Such compatibility constraints can be achieved with geometry of length $r$ for a length $2r^2-2$ string $B$.  For more compact geometries in the case of less complex strings $B$, see Section~\ref{sec:matrix}.  (d) An example assembly of decoder tiles. }
\end{center}
\vspace{-20pt}
\end{figure}

The next step of the construction grows a third row of tiles on top of the surface of the roughly $2\times \log n$ rectangle.  Within this assembled row one tile type representing a binary 0 or 1 bit is placed at each position, thus assembling a length roughly $\log n$ binary string upon completion.  To generate $O(\log n)$ bits from only $O(\sqrt{\log n})$ distinct tile types is impossible in general within the ATAM.  Within the GTAM, at this stage in the assembly we make use of the \emph{non-specific} hindrance property of geometric tile faces to select the correct bit (and reject the wrong bit) at each of the $\log n$ bit positions.  The key idea is that a collection of $m$ geometric tile faces can be designed such that each face is compatible with a specified subset of the other faces, while incompatible with all others.  Thus $m$ tile faces can encode a compressed $m^2$ binary pieces of information (compatible or not compatible), thereby providing the possibility for a more tile type succinct assembly of an $n\times n$ square.  A description of the decoder tile set is given in Figure~\ref{fig:decoderTilesCombo} along with an example of how the assignment of geometry to tile faces permits the decoder tiles to efficiently select the correct bits.  The details of this portion of the construction are described in Section~\ref{subsubsec:bitDecoder}.

Once the binary string is assembled on the surface of the $3\times \log n$ rectangle, we utilize a well known constant sized set of tiles that implement a binary counter in the ATAM at $\tau=2$~\cite{RotWin00}.  This system reads a given surface of glues that denote an initial binary value for the counter, and then assembles upwards, incrementing a binary value encoded in tile types at every other row of the assembly.  Once the counter is maxed out the construction stops, thus growing a rectangle of height roughly $2^{\log n}$ minus the initial value of the counter.  This $\tau=2$ ATAM counter construction is a zig-zag construction (see Definition~\ref{def:zigzag}).  Thus, our construction applies Theorem~\ref{thm:zigzagSimulation1Glue} to convert to a $\tau=1$ GTAM version.

Finally, with the ability to generate large length $O(n)$ rectangles with $O(\sqrt{\log n})$ tile types at $\tau=1$, we combine 3 of these constructions to assemble the border of an $n\times n$ square using a factor of 3 times more tile types.  A high level schematic of the approach is given in Figure~\ref{fig:seededSquareOverview}.  With the shell constructed, the completion of the final rectangle can seed a growth of filler tiles to fill in the body of the square, finishing the construction.  The final details of the construction are described in Section~\ref{subsubsec:binaryCounter}.

\subsubsection{Tight Kolmogorov Lower Bound for GTAM Squares: $\Omega(\sqrt{\log n})$}

\begin{theorem}\label{thm:gtamLower}
For almost all integers $n$, the minimum tile complexity required to assemble an $n\times n$ square in the GTAM is $\Omega(\sqrt{\log n})$.
\end{theorem}
\begin{proof}
	The Kolmogorov complexity of an integer $n$ with respect to a universal Turing machine $U$ is $K_U(n) = \min|p|$ s.t. $U(p) = b_n$ where $b_n$ is the binary representation of n.  It is known that $K_U(n) \geq \lceil\log n \rceil - \Delta$ for at least $1-(\frac{1}{2})^\Delta$ of all $n$ (see \cite{Li:1997:IKC} for results on Kolmogorov complexity).  Thus, for any $\epsilon > 0$, $K_U(n) \geq (1-\epsilon)\log n = \Omega(\log n)$ for almost all $n$.

Consider a tile simulator program (of constant size in bits) that reads as input a GTAM tile system (encoded as a bit string).  Suppose the simulator is modified so that it outputs the maximum extent (i.e. width or length) of a shape that is terminally produced by the input system.  When such a simulator is paired with a GTAM system that uniquely assembles an $n\times n$ square, the combined program constitutes a program that outputs the integer $n$, implying that the total number of bits of the simulator (constant) plus the encoding of the tile set must be at least $K_U(n)$.  As the simulator has a constant size, this implies that the number of bits to to encode the GTAM system must be at least $K_U(n)$, which is $\Omega(\log n)$ for almost all $n$.  To achieve our bound we now show that any GTAM system can be encoded using $O(|T|^2)$ bits (independent of the size/area of the tile face geometries) assuming a constant bounded temperature.  To achieve this, we do not explicitly encode the geometry for each tile face, but instead utilize a \emph{compatibility matrix}.

\paragraph{Encoding a GTAM system.}
For a GTAM system $\Gamma = (T,\tau, s)$, arbitrarily index each distinct face of each distinct tile in $T$ from $1$ to $4|T|$.  Define the \emph{compatibility matrix} $M$ for $\Gamma$ to be the $4|T| \times 4|T|$ matrix such that $M_{i,j} = 1$ $\iff$ $i$ is the index of an east (or north respectively) edge and $j$ is the index of a west (south respectively) edge and $i$ and $j$ have incompatible edge geometries.  $M$ can be encoded using $O(T^2)$ bits, and the remaining portions of $\Gamma$ can easily be encoded in asymptotically fewer bits, yielding an $O(T^2)$ bit encoding for any GTAM system.  Note that even without the explicit representation of the GTAM's geometries, a simulator can derive what the system will build from the compatibility matrix $M$.

Now consider the smallest tile type GTAM system $\Gamma=(T,\tau,s)$ that uniquely assembles an $n\times n$ square.  As $\Gamma$ can be encoded in $O(|T|^2)$ bits, we know that for almost all $n$, $c_1 |T|^2 \geq c_2 \log n$ for constants $c_1, c_2$. Therefore, $|T| = \Omega(\sqrt{\log n})$ for almost all $n$.
\end{proof}

\subsection{Simulating Temperature $\tau=2$ ATAM Systems with $\tau=1$ GTAM Systems}\label{subsec:SimulatingTemperature}

\begin{definition}{\textbf{Zig-Zag System.}}\label{def:zigzag} A tile system $\Gamma = (T,\tau,s)$ is called a zig-zag system if:
\begin{enumerate}
\item The location and type of the $i^{th}$ tile to attach is the same for all assembly sequences.
\item The $i^{th}$ tile attachment occurs to the north, west, or east (not south) of the previously placed tile attachment in all assembly sequences.
\end{enumerate}
\end{definition}

For the proofs of the following two theorems and technical details about the notion of tile system ``simulation'', please see Section~\ref{sec:zig-zag-details}.

\begin{theorem}\label{thm:zigzagSimulation}
Any temperature $\tau=2$ zig-zag ATAM tile system $\Gamma = (T,2,s)$ can be simulated by a $\tau=1$ GTAM tile system $\Upsilon=(R,1,q)$ with tile type scale $|R|/|T| = O(1)$.  The simulation utilizes geometry size at most $\log\sigma_n + \log\log\sigma_n + O(1)$ where $\sigma_n$ is the number of distinct north/south glue types represented in $T$.
\end{theorem}

\begin{theorem}\label{thm:zigzagSimulation1Glue}
Any temperature $\tau=2$ zig-zag ATAM tile system $\Gamma = (T,2,s)$ can be simulated by a $\tau=1$ GTAM tile system $\Upsilon=(R,1,q)$ using only 1 non-null glue type and tile type scale $|R|/|T| = O(1)$.  The geometry size of the simulation system is at most $\log\sigma + \log\log\sigma + O(1)$ where $\sigma$ is the number of distinct glue types represented in $T$.
\end{theorem}

\section{2GAM Results}
\label{sec:2GAM-results}

In this section, we explore the theoretical limits achievable when utilizing geometric tiles by designing tiles whose edges contain highly complex geometries.  Furthermore, we move to the 2-handed variant of the GTAM, the $2$GAM, to allow for the geometric hindrances experienced by individual tiles to be grouped and combined to provide more complex interactions between larger supertiles.  The goal, rather than providing a realistic and potentially experimentally realizable set of constructions, is to gain further understanding into the interplay between geometry and the types of computations which can be carried out via algorithmic self-assembly.

We now present the details of our construction, which reduces the tile complexity required to self-assemble an $n \times n$ square to a mere $O(\log \log n)$ tile types, while requiring a geometry size of $O(\log n \log \log n)$.  Our construction requires the constraint of planarity, in which components are not allowed to float into position from above or below the assembly, but must always be able to slide into position with a series of translations along only the $x$ and $y$ axes.  However, the intricate geometric designs and complex series of movements require that individual tile geometries are composed of disconnected components.  (Note that in Section~\ref{sec:2d-conversion-to-3d} we show how to extend the tiles into the third dimension, utilizing a total of $4$ planes, in a manner which results in connected tiles and also implicitly enforces the restriction that only tile translations along the $x$ and $y$ axes must be sufficient to allow for tile attachments.)

\subsection{Self-assembly of an $n \times n$ square with $O(\log \log n)$ tile types}

\begin{theorem}\label{thm:2GAM-squares}
For every $n \in \mathbb{N}$, there exists a 2GAM tile system $\Gamma = (T,2)$ which uniquely produces an $n \times n$ square, where $|T| = O(\log \log n)$, and with $O(\log n \log \log n)$ size geometry.
\end{theorem}

To prove Theorem~\ref{thm:2GAM-squares}, we present the following construction.

\begin{wrapfigure}{l}{3.5in}
\vspace{-15pt}
\begin{center}
    \includegraphics[width=3.0in]{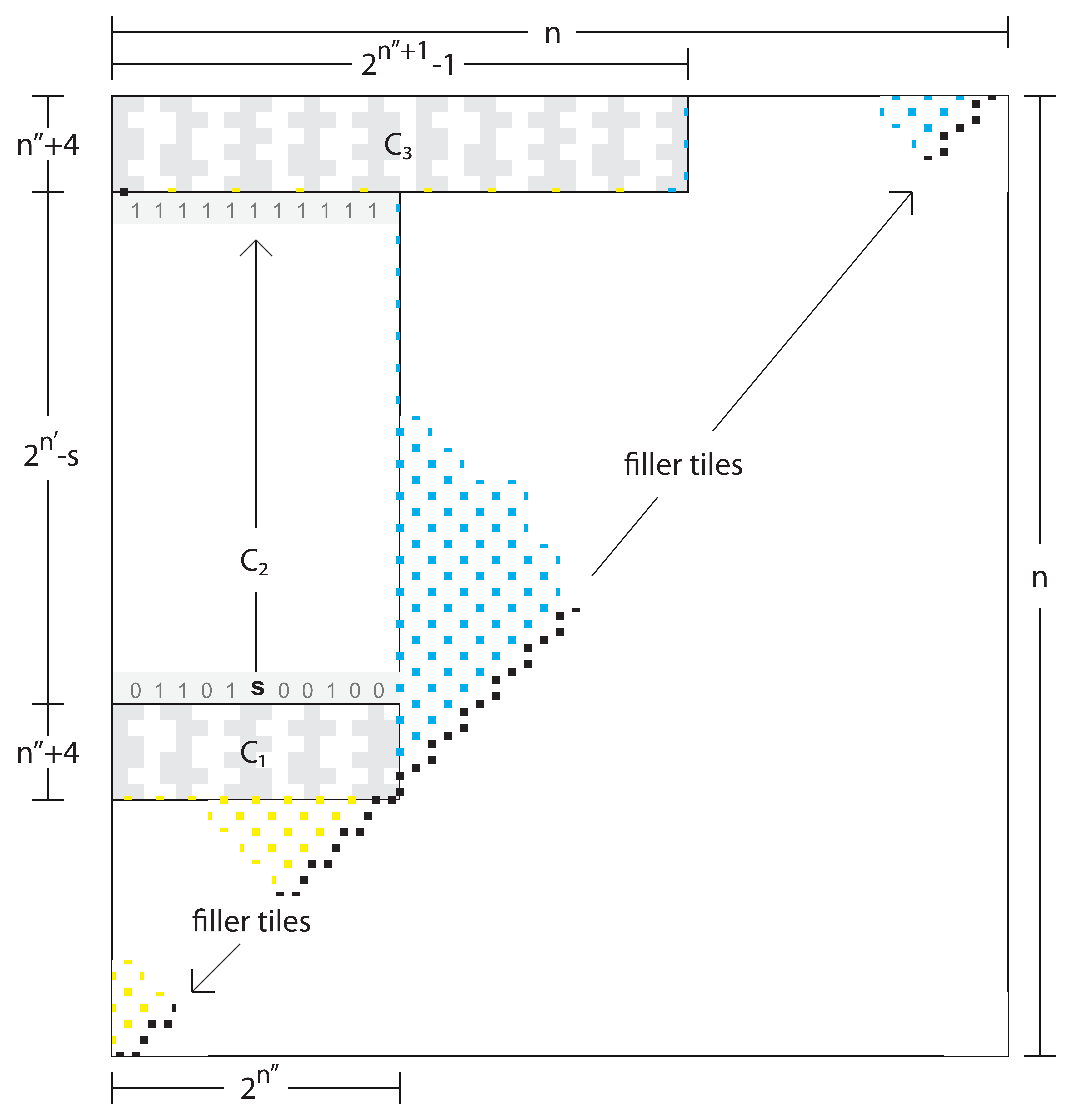} \caption{\label{fig:planar-high-level} \footnotesize A high-level sketch of the construction for building a square in the 2GAM.}
\end{center}
\vspace{-55pt}
\end{wrapfigure}

\subsubsection{High-level sketch of the construction}

Following is a list of values based on the particular dimensions of the square to be formed and which are used throughout the following discussion:

\begin{itemize}
  \item $n$: dimensions of the square to self-assemble
  \item $n'$: $\lceil \log n \rceil$
  \item $n''$: $\lceil \log n' \rceil$
  \item $s$: $2^{n'} + 2^{n''} + 2n'' + 8 - n$
  \item $h$: $2^{n''-1} - 1$
  \item $C_1$: $2$-handed counter which counts from $0$ through $2^{n''}-1$ for a total of $2^{n''}$ columns
  \item $C_2$: standard counter which counts from $s$ through $2^{n'}-1$ for a total of $2^{n'} - s$ columns
  \item $C_3$: $2$-handed counter with ``buffer'' columns which counts from $0$ through $2^{n''}-1$ for a total of $2^{n''+1}-1$ columns
\end{itemize}

Figure~\ref{fig:planar-high-level} shows a high level view of the main components of this construction. Without loss of generality, we can consider the construction to be composed of a series of sub-assemblies, or modules, which assemble in sequence, with each module completely assembling before the next begins.  The careful design of all modules ensures that none can grow so that they occupy space required by another, and that each will be able to terminally grow to precisely defined dimensions that result in the final combination forming exactly an $n \times n$ square.  For the rest of this discussion, we will describe the formation of the modules in such a sequence. (See Figure~\ref{fig:full-square-assembly-sequence} for a series of high-level images which exemplify the ordering of the formation of the square from these modules.)

Similar to the construction in Section~\ref{subsec:TheTileComplexity}, this construction makes use of one counter, $C_1$, to assemble an encoding of a number which in turn seeds another counter, $C_2$.  $C_1$ assembles in a $2$-handed manner, meaning that each number which is counted is represented by exactly one one-tile-wide column of tiles, and individual columns form separately and then combine to form the full counter of length $2^{n''}$ (similar in design to counters found in \cite{SFTSAFT}).  Each column of the counter, besides representing a counter value, is used to represent (on the north face of the northernmost tile) one bit of the seed value $s$ for $C_2$.  Each column can form in one of two versions: one that represents a $0$, and one that represents a $1$, for the corresponding bit of $s$.  The east and west sides of the tiles forming the columns of this counter contain geometries which force the columns, in order to combine, to ``wiggle'' up and down in patterns based on the counter values of those columns.  See Figure~\ref{fig:2gam-column-for-poster} for an example pair of compatible columns.  The columns also contain tiles with geometries which ``read'' those patterns of wiggling and allow columns to combine with each other if and only if they are the correct versions of the counter columns, namely those with the seed bit values which correctly correspond to their location in the counter.  It is the tiles of this component as well as those of the counter $C_3$ to which the intricate geometries are applied, and thus they receive a much more detailed explanation in Section~\ref{sec:counters-c1-and-c3}.

\begin{wrapfigure}{r}{2.5in}
    \vspace{-25pt}
    \begin{center}
    \includegraphics[height=5.5in]{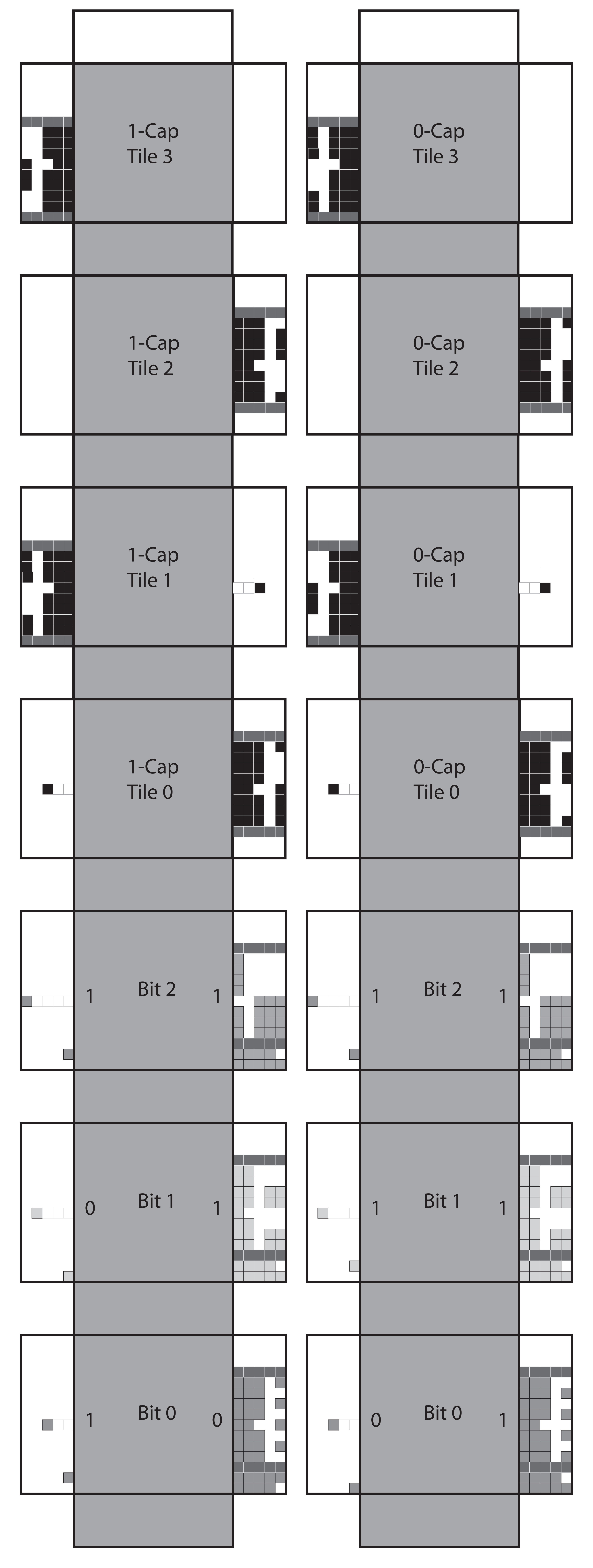} \caption{\label{fig:2gam-column-for-poster} \footnotesize Example columns for the counter $C_1$.  Note that all colored areas are filled-in, and areas colored white are empty, although they may be outlined for reference.}
    \end{center}
\end{wrapfigure}

$C_2$ is a standard binary counter (i.e. one that would also assemble correctly in the aTAM) which utilizes $16$ tile types (see Figure~\ref{fig:c2-counter-tiles}) and grows to complete the majority of the western side of the square.  Next, a small set of $7$ ``filler'' tile types (see Figure~\ref{fig:filler-tiles}) fill in the majority of the square, and once they have filled in a sufficient portion of the northern portion they provide a platform to which $C_3$ can attach (as long as $C_3$ is fully formed). In order to provide a directed system with only one terminal assembly, the ``incorrect'' columns (those which couldn't become part of $C_1$) are able to combine into the $2$-handed counter structure $C_3$ via some extra buffer columns (see Figures \ref{fig:cap-tile-geometry} and \ref{fig:column-combination-example-buffer}).  Finally, the filler tiles are able to complete the formation of the square.  Note that the tile types which make up $C_2$ and the filler tiles require no geometries but only standard glues.

By utilizing the assembly of supertiles (i.e. sub-assemblies of grouped tiles) and carefully designing geometries which force the supertiles forming $C_1$ to move in well-defined patterns as they attach, we are able to essentially ``transmit'' information about tiles in one location of a supertile to the interfaces where potential binding is occurring with other tiles in the same supertile.  By concatenating this information from such a group of distant tiles, the binding ``decision'' can be made based on an arbitrarily large amount of information (as long as the geometry sizes scale appropriately).  This results in a dramatic lowering of the tile complexity required to assemble an $n \times n$ square, with the tradeoff being an increase in the complexity of the tiles themselves.

Please see Section~\ref{sec:2GAM-results-details} for much more detail and several supplementary images describing this construction.  Additionally, a comprehensive example has been provided in Section~\ref{sec:2HAM-square-example} to which the reader can refer for additional clarity.

\subsection{Analysis of tile complexity and geometry size}

First, we analyze the tile complexity of this construction, module by module, in order to determine the overall complexity.

The tile complexity of each component is as follows:
\begin{itemize}
    \item $C_1$:
        \subitem Counter tiles: There are $n''$ bit positions which each require a constant number of tile types (as can be seen from the depiction in Figure~\ref{fig:increment-and-edge-detect-tiles}), plus the requirement for a hard-coded column on each of the west and east sides, for a total of $O(\log \log n)$ tile types.
        \subitem Cap tiles:  There are $4$ cap tile positions which each need to be able to represent a $0$ or a $1$ cap, for a total of $8$ tile types.
    \item $C_2$: 16 tile types.
    \item $C_3$:
        \subitem Counter tiles: There are $n''$ bit positions in the buffer columns which each require a constant number of tile types, for a total of $O(\log \log n)$ tile types.
        \subitem Buffer cap tiles: There are $4$ cap tile positions which each require a single tile type, for a total of $8$ tile types.
    \item Filler tiles:  7 tile types.
\end{itemize}

Thus, the total tile type complexity is $O(\log \log n) + O(1) + O(1) + O(\log \log n) + O(1) + O(1) = O(\log \log n)$.

Next, we simply note that the geometries defined for all tiles in this construction consist of rectangles of dimensions $(2^{n''} + h + 4) \times (n'' + 2) = (2^{n''} + \lceil 2^{n''} / 2 \rceil + 4) \times (n'' + 2) = O(\log n \times \log \log n)$, and therefore the geometry size is $O(\log n \log \log n)$.

\bibliographystyle{amsplain}
\bibliography{tam}

\clearpage
\appendix

\section{Formal 2GAM Definition}
\label{sec:formal-2GAM}

\paragraph{Hindrance Map}  To permit a supertile to interweave itself into an attachable position requires modeling translations of a supertile at the resolution of the size of individual units of the tile face geometries.  To this end we define a \emph{Hindrance Map} for a supertile $\alpha$ that represents the set of positions that a supertile takes up, including the bodies of each tile in the supertile, along with the positions blocked by each geometry for each tile face.

Formally, consider a supertile $\alpha$.  The Hindrance Map $H_\alpha$ is the following set of positions:
For each tiled position $(x,y)$ in supertile $\alpha$, the following points are defined to be in $H_\alpha$. (Please refer to Figure~\ref{fig:geometry-coordinates} for depictions of tile components in terms of $w$ and $\ell$, and note that the entire area occupied by a

BODY:  $\{ (i,j) | x \cdot (w + \ell) + w \leq i < x \cdot (w + \ell) + w + l,\\
                   y \cdot (w + \ell) + w \leq j < y \cdot (w + \ell) + w + l\}$\\
                   
WEST GEOMETRY: $\{(i,j) | x \cdot (w + \ell) \leq i < x \cdot (w + \ell) + w, \\
                          y \cdot (w + \ell) + w \leq j < y \cdot (w + \ell) + w + l,\\
                            (i - (x \cdot (w + \ell)), j - (y \cdot (w + \ell) + w)) \in WestGeometry(\alpha(x,y))\}$\\
                            
SOUTH GEOMETRY: $\{(i,j) | x \cdot (w + \ell) + w \leq i < x \cdot (w + \ell) + w + l, \\
                          y \cdot (w + \ell) \leq j < y \cdot (w + \ell) + w,\\
                            (i - (x \cdot (w + \ell) + w), j - (y \cdot (w + \ell))) \in SouthGeometry(\alpha(x,y))\}$\\
                            
NORTH GEOMETRY: $\{(i,j) | x \cdot (w + \ell) + w \leq i < x \cdot (w + \ell) + w + l, \\
                          (y + 1) \cdot (w + \ell) \leq j < (y + 1) \cdot (w + \ell) + w,\\
                            (i - (x \cdot (w + \ell) + w), j - ((y + 1) \cdot (w + \ell))) \in SouthGeometry(\alpha(x,y))\}$\\

EAST GEOMETRY: $\{(i,j) | (x + 1) \cdot (w + \ell) \leq i < (x + 1) \cdot (w + \ell) + w, \\
                          y \cdot (w + \ell) + w \leq j < y \cdot (w + \ell) + w + l,\\
                            (i - ((x + 1) \cdot (w + \ell)), j - (y \cdot (w + \ell) + w)) \in WestGeometry(\alpha(x,y))\}$\\

%
%

The final hindrance map $H_\alpha$ is the union of the sets BODY, SOUTH, NORTH, EAST, and WEST for each tiled position of $\alpha$.

\begin{figure}
\begin{center}
    \includegraphics[width=3.0in]{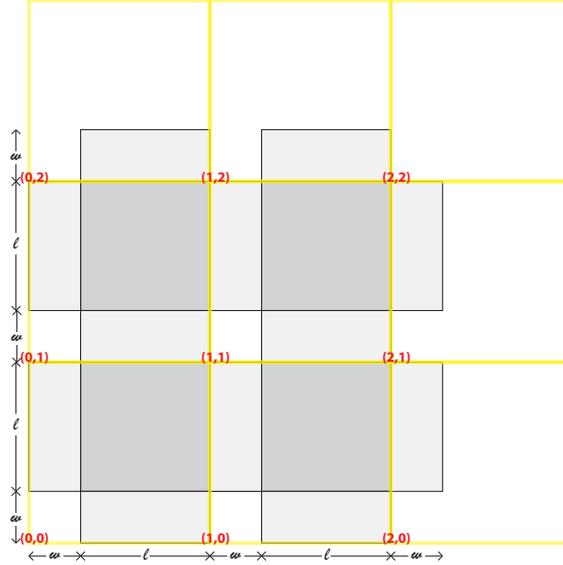}\caption{\label{fig:geometry-coordinates} Mapping of tile coordinates (red) to dimensions and units of tile bodies and geometries.}
\end{center}
\end{figure}

\paragraph{Planar Translation of Supertiles}
Given two supertiles $\alpha$ and $\beta$, a \emph{collision free} translation of $\beta$ with respect to $\alpha$ is any translation of $H_\alpha$ that can be obtained by a sequence of unit translations $\{v_0, u_1, u_2, u_3, \ldots u_r\}$ where $v_0$ is an initial translation that shifts $H_\alpha$ such that all positions of $H_\alpha$ are northwest of all positions of $H_\beta$, and each $u_i$ is one of the translations $\{ u_n = (0,1), u_e = (-1,0), u_s = (0,-1), u_w = (-1,0) \}$.  Further, after each translation $u_i$, it must be the case that $H_\alpha$ does not overlap $H_\beta$.  A \emph{grid locked collision free} translation is a collision free translation in which $H_\alpha$ has been shifted by multiples of $w +\ell$ in both the $x$ and $y$ direction.  We are interested in grid locked translations as they correspond to direct translations of $\alpha$ at the resolution of tiles, rather than the higher resolution translations of $H_\alpha$.  We require grid locked translations for a supertile to attach to another supertile as such a translation is needed for the tiles to \emph{line up}.

\paragraph{2GAM Model}
A Two-Handed Planar Geometric Tile Assembly System consists of a duple $(T,\tau)$ where $T$ is a set of geometric tile types and $\tau$ is the positive integer temperature of the system.  Given a 2GAM system $\Gamma = (T, \tau )$, a supertile is producible if either it is a single tile from T, or it is the $\tau$-stable result of a grid locked collision free translation of two producible assemblies. A supertile $\alpha$ is terminal if for every producible supertile $\beta$, $\alpha$ and $\beta$ cannot be $\tau$-stably attached. A 2GAM is directed (e.g., deterministic, confluent) if it has only one terminal, producible supertile. Given
a connected shape $X\subseteq \mathbb{Z}^2$, a 2GAM $\Gamma$ produces $X$ uniquely if every producible, terminal supertile places tiles only on positions in $X$ (appropriately translated if necessary).

\section{Classes of Tile Face Geometries}
\label{sec:geometry-classes}

Here we discuss different classes into which tile face geometries can be classified.  Note that all GTAM results presented in this paper have ``bump'' geometries, while the 2GAM result (in its $2$-dimensional form) has ``unrestricted'' geometries.  See Figure~\ref{fig:geometry-classes} for an example of each class.

\begin{wrapfigure}{l}{1.0in}
\vspace{-20pt}
\begin{center}
    \includegraphics[height=3.0in]{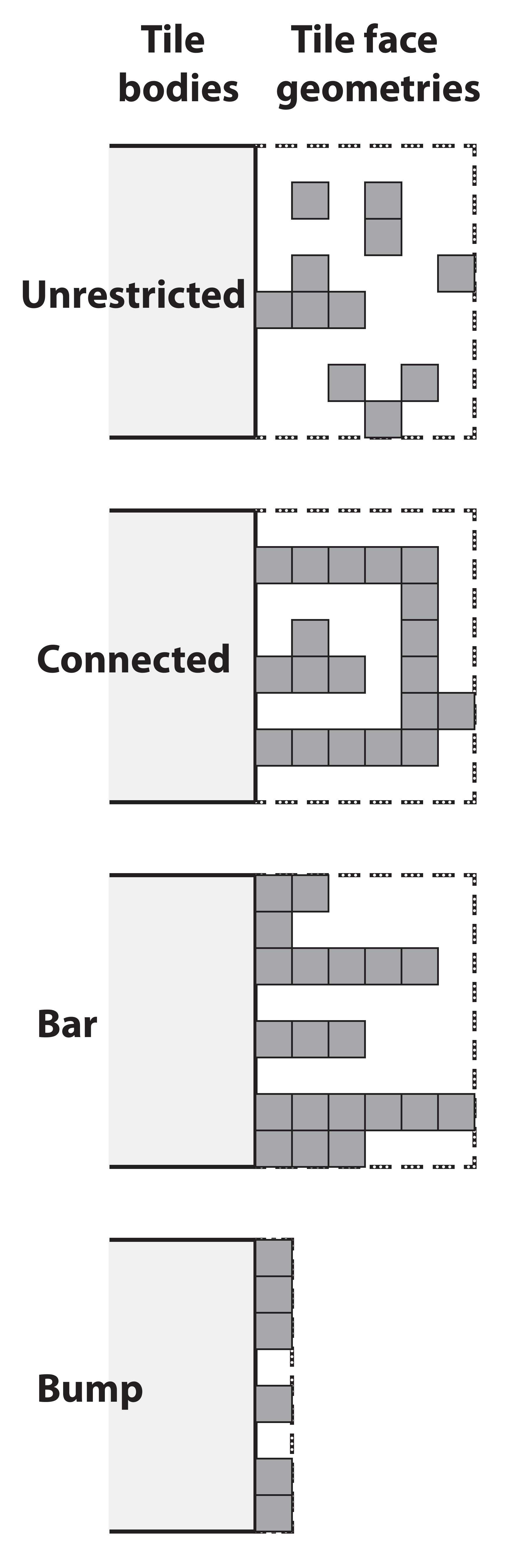} \caption{\label{fig:geometry-classes} \footnotesize Classes of tile geometries.}
\end{center}
\end{wrapfigure}

\paragraph{Unrestricted}
This is the most general class of geometries and places no restrictions on the portions of a tile face geometry region (i.e. the $w \times \ell$ rectangle) which are filled-in and which are empty.  Such geometries may be infeasibl to implement as the pieces of a geometric face may not be connected to the tile body.  However, given a third dimension it is plausible that such a scheme might be implemented by attaching particles to the face of a substrate which also attaches to the tile body.

\paragraph{Connected}
A slightly more restricted class, the connected class allows arbitrary patterns to be filled-in in within the tile face geometries as long as all such portions retain a connected path to the tile body.

\paragraph{Bar}
A bar geometry is restricted to lines of filled-in points which are connected to the tile body and extend directly away from it.  Each bar can be of length $x$ where $0 \leq x < w$.

\paragraph{Bump}
Bump geometries are the simplest possible types of geometry and consist of a set of points which are directly connected to the tile body.  This class can be thought of as simplified bar geometry with $w = 1$.

\section{Additional Details for the $O\sqrt{\log n}$ GTAM Square Construction}
\label{sec:seeded-details}

\subsubsection{Construction Notation}
Consider a positive integer $n$ (the width of the square we wish to assemble).  For the sake of clarity, assume $n$ is even.
\begin{itemize}
\item Let $n' = \lceil \log{\frac{n+1}{2}} \rceil$.
\item Let $r = \lceil \sqrt{(1/2) n' + 1/2}\rceil$.
\item Let $B = 2^{n'} - n/2 - 1$.
\end{itemize}

\begin{figure}[htp]
\begin{center}
    \includegraphics[width=6.5in]{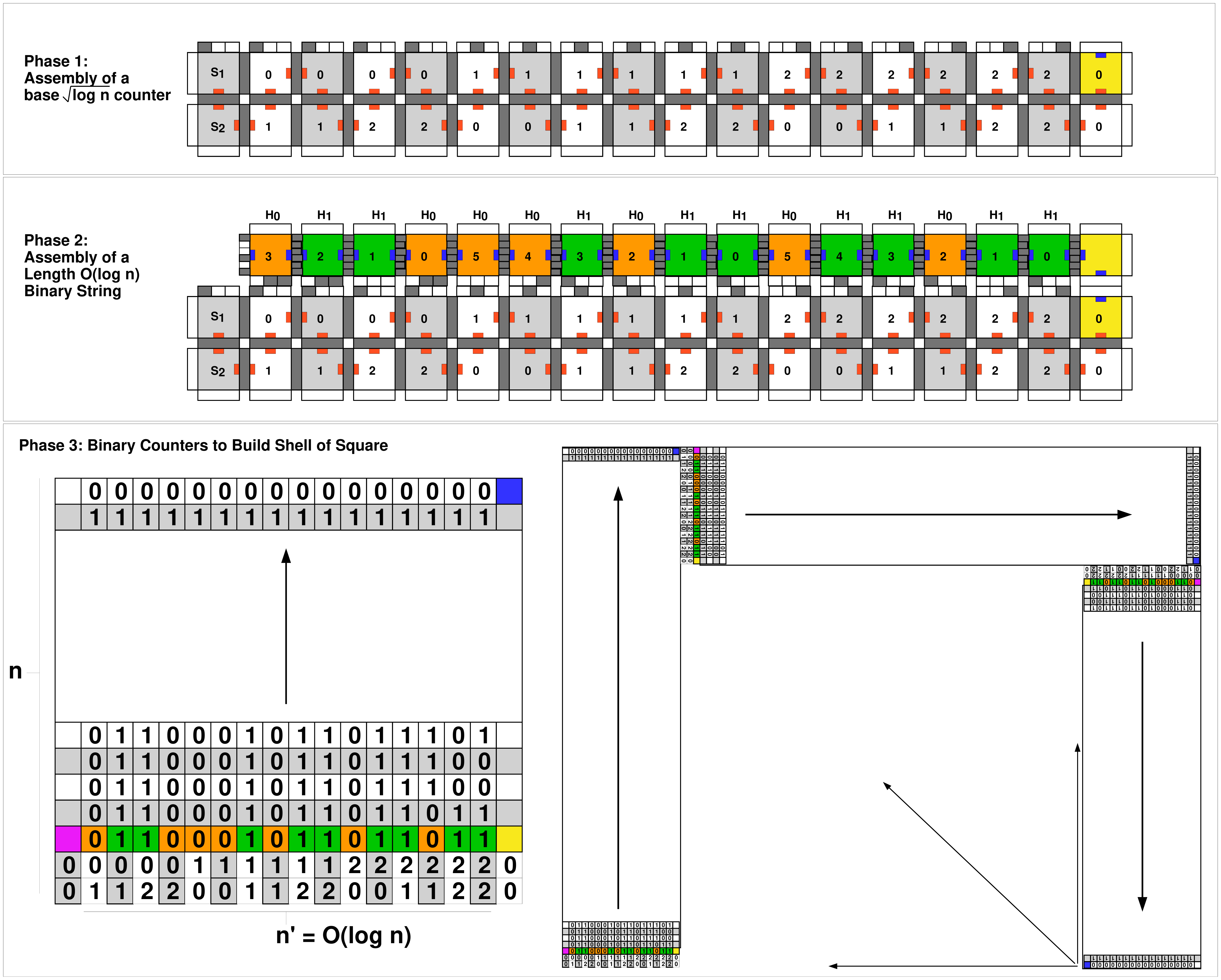} \caption{\label{fig:seededSquareOverview} \footnotesize This figure shows a high level overview of the different steps in the assembly of $n\times n$ squares in the GTAM with $O(\sqrt{\log n})$ tile types at temperature $\tau=1$.  In phase 1, a tile set that implements a 2-digit, base $O(\sqrt{\log n})$ counter is used to form a length $O(\log n)$ bed upon which a binary number will be assembled.  Phase 2 places green and orange decoder tiles which denote either a 0 or 1 bit at each position of the third row of the assembly.  In phase 3, the assembled binary string is utilized as the seed for a binary counter set of tile types which grow a length $n$ rectangle.  Phases 1-3 are repeated two more times to create the outer shell of an $n\times n$ square.  Finally, a collection of filler tiles are seeded to fill in the body of the square.   }
\end{center}
\end{figure}

\subsubsection{Base $O(\sqrt{\log n})$ Counter at Temperature $\tau=1$}\label{subsubsec:baseCounter}
\begin{figure}[htp]
\begin{center}
    \includegraphics[width=6.5in]{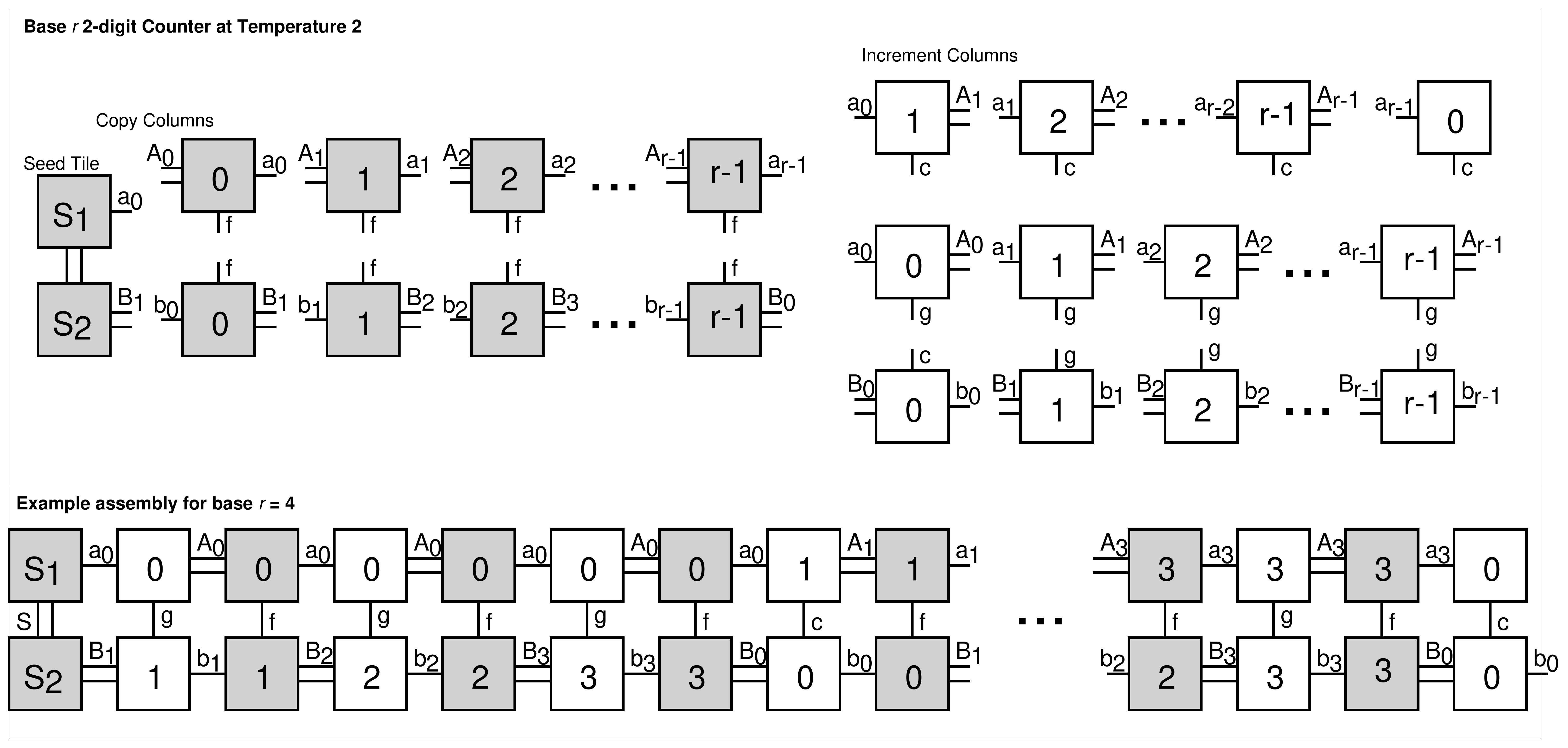} \caption{\label{fig:temp2BaseCounterWithExample} \footnotesize This figure contains a tile system that assembles a $2\times r^2$ rectangle for a given integer $r$.  The construction is an implementation of a 2-digit, base-$r$ counter with tile complexity $5r + 2$ that assembles at temperature $\tau=2$ in the standard ATAM. }
\end{center}
\end{figure}

\begin{figure}[htp]
\begin{center}
    \includegraphics[width=6.5in]{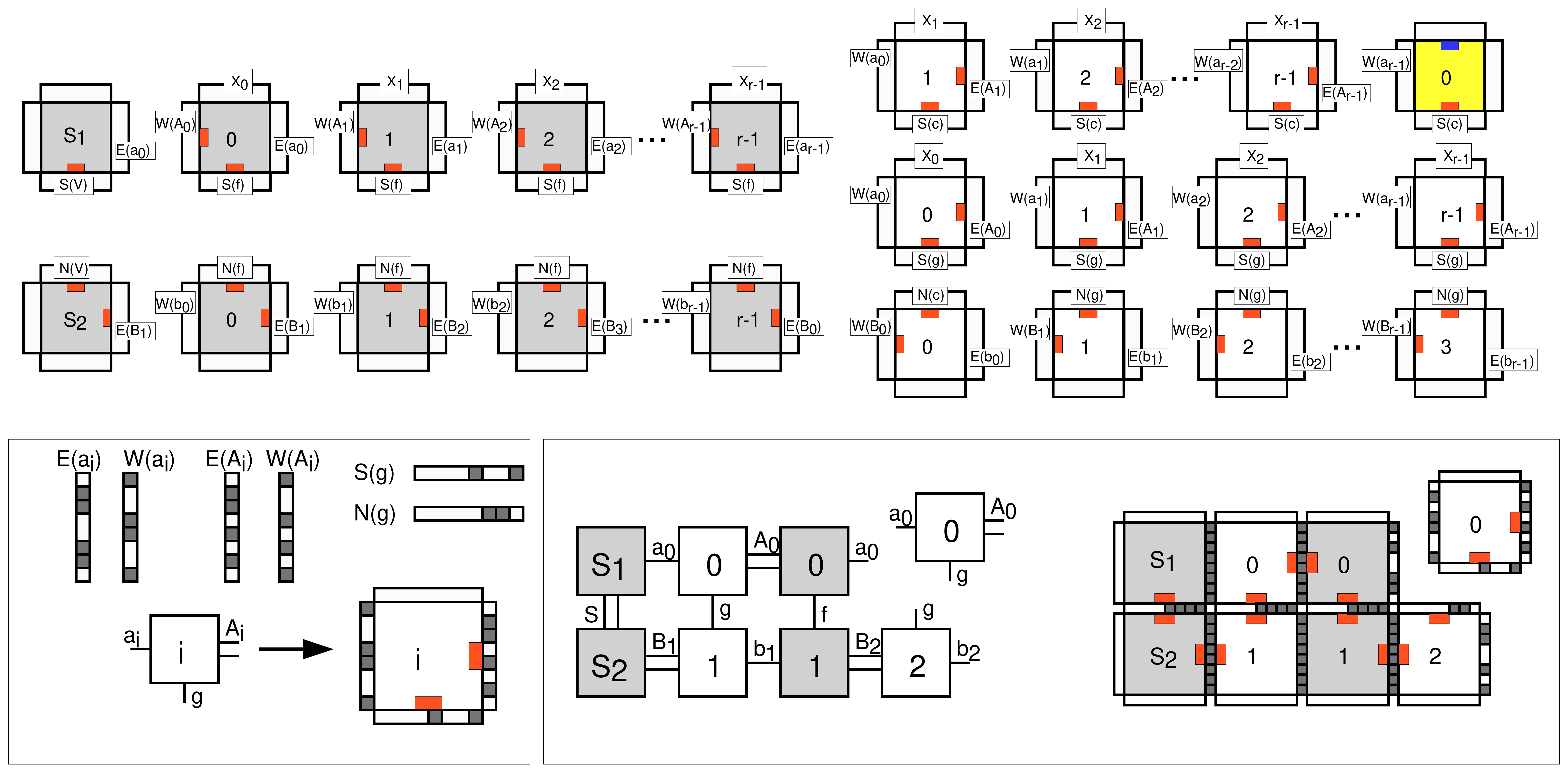} \caption{\label{fig:Temp2ToTemp1SingleTileExample} \footnotesize This tile system is the result of applying the transformation from Theorem~\ref{thm:zigzagSimulation1Glue} to the zig-zag counter described in Figure~\ref{fig:temp2BaseCounterWithExample}.  The systems utilizes a single (red) glue at temperature $\tau=1$.  For each east/west (or north/south) glue type $x$ in the initial ATAM system, there exist two corresponding geometries $W(x)$ and $E(x)$ (or $N(x)$ and $S(x)$) in the new GTAM system such that $W(x)$ is compatible with $E(x)$, but incompatible with all other geometries in the system.}
\end{center}
\end{figure}

The first step of the construction is the assembly of a $2\times \log n$ rectangle that will serve as a bed for a third layer that will place a row of tiles that represent a $\log n$ bit binary number.  The $\tau=1$ GTAM system used for this portion of the assembly is obtained by applying the transformation from Theorem~\ref{thm:zigzagSimulation1Glue} to convert an efficient $\tau=2$ ATAM system into a equivalent $\tau=1$ GTAM system.

The temperature $\tau=2$ ATAM system that will be converted is described in Figure~\ref{fig:temp2BaseCounterWithExample} and constitutes a 2-digit, base $r$ counter.  The counter works within the ATAM at temperature $\tau=2$ and is a generalization of the base-2 version first described in~\cite{RotWin00}.  By specifying the east glues of tiles $S_1$ and $S_2$, the counter can be \emph{seeded} to any specified starting value.  In the tile set given in the figure, the counter is initialized to value 0.  The value of the counter is incremented at every other column as the assembly grows from west to east, finally halting when the counter rolls over to 0.  Thus, for a given choice of $r$ and an initial seed value $b$, the final assembly will be a $2\times 2r^2 - 2b$ rectangle.  For our construction, we utilize $r = \lceil \sqrt{(1/2) n' + 1/2}\rceil$, which guarantees enough room to place a length $n'$ binary string in the next step of the construction.  As we can initialize the counter to be shorter if needed, we assume the counter has been initialized to grow to length exactly $n' + 2$ (an extra 2 positions are not used to encode bits in our constructions, thus the extra 2 length).

To modify the ATAM system of Figure~\ref{fig:temp2BaseCounterWithExample} to a $\tau=1$ system, we observe that it is a zig-zag system according to Definition~\ref{def:zigzag} (rotated 90 degrees).  Therefore, we can apply Theorem~\ref{thm:zigzagSimulation1Glue} to obtain an equivalent $\tau=1$ GTAM system shown in Figure~\ref{fig:Temp2ToTemp1SingleTileExample}.  The general case details of the conversion are detailed in Section~\ref{subsec:SimulatingTemperature}, but the basic idea of the transformation is to replace east strength 1 glues with the null glue type, and assign a unique geometry to each edge for each glue type.  In particular, for an east/west glue type x, a corresponding pair of geometries are computed, $E(x)$ and $W(x)$, such that $E(x)$ is incompatible with all other geometries within the system with the exception of $W(x)$, and vice versa.  Each occurrence of the glue type x on the east face of a tile type is replaced by a GTAM tile type with geometry $E(x)$ for the east face geometry.  The same replacement by $W(x)$ is done for west occurrences of $x$.  This geometry assignment is also applied to all north/south glue types as well.  The result is a system that assembles in the same fashion as the original temperature $\tau=2$ system and with the same tile complexity, but does so at temperature $\tau=1$.

\subsubsection{Bit Decoder Tiles}\label{subsubsec:bitDecoder}


The next step in the construction consists of a collection of decoder tiles which grow across the surface of the $2\times 2r^2$ rectangle assembled from the previous step.  The general tile set for these decoder tiles is given in Figure~\ref{fig:decoderTilesCombo}.  The growth of these tiles is initialized by the blue glue displayed by the final tile placement of the base $r$ counter from the previous section.  The decoder tiles consist of a repeating chain of $2r$ tiles with labels $0$ to $2r-1$.  The west geometry of tiles with label $i$ are compatible with the east geometries of tiles with label $(i+1)\mod 2r$, and incompatible with all other geometries.  Thus, the chain of tiles must assemble in the proper order.  Further, there exactly 2 tile types with each integer label from $0$ to $2r-1$, a \emph{green} type and an \emph{orange} type.  Our goal is to assemble a supertile that encodes a given target binary string along its surface.  The encoding is the pattern of green and orange tiles with orange tiles representing binary 0 bits, and green representing binary 1 bits.  The key to get the goal binary string assembled is to enforce that at each bit position the correct bit is chosen.  We do this by appropriately assigning geometry to the south face of the decoder tiles and the north face of the base $r$ counter tiles.

In more detail, suppose we are given a target binary string $B=b_{2r-1}\ldots b_2 b_1 b_0$ to be assembled, the geometries $G_j$, $\bar{G_j}$, and $X_i$ are assigned such that  $G_j$ and $X_i$ are compatible if and only if $b_{2r(r-1) - 2ri +j}=0$.  A detailed example of such a compatibility matrix is given in Figure~\ref{fig:decoderTilesCombo}, along with a sample set of geometry assignments to tile faces that satisfies such constraints.  More generally, such compatibility constraints can be achieved with geometry of length $r$ for a length $2r^2-2$ string $B$ by way of Theorem~\ref{thm:mostMatrices}, which is asymptotically the best achievable in most cases.  In the case of less complex binary strings, more compact geometries can be obtained as discussed in Section~\ref{sec:matrix}.  From these compatibility constraints, the desired target binary string is guaranteed to assemble.  Further, the final placed tile can be specified as a special type with a north \emph{purple} glue which seeds the next portion of the construction.

\subsubsection{Binary Counter Tiles}\label{subsubsec:binaryCounter}
The next portion of the construction, seeded by the final tile placed during the decoder tile portion, consists of a binary counter tile system which grows north, start from the binary string decoded in the previous section, up until the counter rolls over.  The construction is a $\tau=1$ GTAM version of a well known temperature $\tau=2$ ATAM construction of $O(1)$ tile types~\cite{RotWin00}.  The counter increments every other row, and thus will grow to a height of exactly $2^{n'+1} - 2B -2$, where $B$ is the initial value of the counter which is the binary string encoded in the previous phase of the construction. Our goal is for the counter to build to a height equal to $n$, the dimension of the goal $n\times n$ square.  We thus choose $B$ in the previous section to be $B=2^{n'}-n/2 -1$.

Finally, the construction is finished by observing that the binary counter construction can be implemented such that the final tile placed at the northeast corner of the assembly exposes a glue type which seeds the assembly of a rectangle that grows east for exactly $n- n' - 2$ units.  This can be accomplished in the exact same way we achieved a north growing rectangle of length $n$, but with an alternate choice of initial binary string assignment. This construction can in turn seed a south growing rectanle of the same length.  This final rectangle can seed the growth of a final $O(1)$ size collection of tile types which fills in the casing of the hollow $n\times n$ square, yielding the final full $n\times n$ square.  A high level figure depicting the construction is given in Figure~\ref{fig:seededSquareOverview}. 
\section{Details of Zig-Zag Simulation}
\label{sec:zig-zag-details}

\begin{definition}{\textbf{Nice Zig-Zag System.}}\label{def:nicezigzag} A zig-zag system $\Gamma = (T,\tau,s)$ is called a \emph{nice} zig-zag system if:
 \begin{enumerate}
    \item The terminal assembly of $\Gamma$ contains no exposed east-west non-$\null$ glue types.
    \item All producible assemblies of $\Gamma$ contain no mismatched glues.
 \end{enumerate}

\end{definition}


\begin{definition}{\textbf{Tile System Simulation.}} An ATAM or GTAM system $\Gamma_2=(T_2,\tau_2,s_2)$ is said to simulate a second ATAM or GTAM system $\Gamma_1=(T_1, \tau_1, s_1)$ if:
\begin{enumerate}
\item There exists a function $F: T_1 \rightarrow P(T_2)$ such that an assembly sequence $\langle (t_1, x_1,y_1), \ldots (t_i,x_i,y_i) \rangle$ is valid for system $\Gamma_1$ if and only if there exists tile types $r_i \in F(t_i)$  such that the assembly sequence $\langle (r_1\in F(t_1),x_1,y_1), \ldots (r_i\in F(t_i),x_i,y_i) \rangle$ is valid for $\Gamma_2$.
\end{enumerate}
The tile complexity scale factor of the simulation is defined to be $|T_2|/|T_1|$.  This definition only considers the case of simulating a system without scaling the size of the assembly (scale factor 1). See~\cite{CookFuSch11} for a more general definition of simulation that permits scaled assembly size factors.
\end{definition}

\begin{observation} The ``simulate'' relation between tile systems is transitive.  Further, if system $A$ simulates $B$ with tile type scale factor $x$, and $B$ simulates $C$ with tile type scale factor $y$, then $A$ simulates $C$ with tile type scale factor $xy$.
\end{observation}

\begin{lemma}\label{lemma:nicezigzag}  Any zig-zag system $\Gamma_1=(T_1, \tau_1, s_1)$ can be simulated by a nice zig-zag system $\Gamma_2=(T_2,\tau_2,s_2)$ with tile type scale factor $|T_2|/|T_1| = O(1)$.
\end{lemma}

\begin{proof} of Theorem~\ref{thm:zigzagSimulation}:
Consider a zig-zag system $\Gamma = (T,2,s)$.  By Lemma~\ref{lemma:nicezigzag}, there exists a zig-zag system $\Gamma' = (T',2,s')$ that simulates $\Gamma$ with $O(1)$ tile type scale such that the assembly of $\Gamma'$ has the ``nice'' properties described in Definition~\ref{def:nicezigzag}.

We now define a system $\Upsilon=(R,1,q)$ that simulates $\Gamma'$ with $O(1)$ tile type scale, and thus simulates $\Gamma$ with $O(1)$ tile type scale by the transitivity of simulation.

Let $\sigma'_w$ denote the set of all west-east glues that are represented in the tile set $T'$.  Let $\sigma'_n$ denote the set of all north-south glues that are represented in the tile set $T'$.  Let $H$ denote an injective mapping from the glue types represented in $T'$ to some new set of strength-1 glues $\rho$.

For each $t\in T'$, define the the geometric tile type $r_t$ as follows.  Denote the north, south, east, and west glues of $t$ as $north(t)$, $south(t)$, $east(t)$, and $west(t)$ respectively.  Let the west glue of $r_t$ be $west(r_t)= H(west(t))$ and $east(r_t) = H(east(t))$.  If $north(t)$ is a strength-2 glue, then $north(r_t) = H(north(t))$, otherwise $north(r_t) = null$.  If $south(t)$ is a strength-2 glue, then $south(r_t) = H(south(t))$, otherwise $south(r_t) = null$.

We assign the empty set geometry to all east/west edges of tile types in $R$. We assign non-empty geometries to north/south edges of each tile type $r_t$ such that the south edge of a given $r_t\in R$ is compatible with a north edge of a tile type $r_v\in R$ if and only if $south(t) = north(v)$.  By Theorem~\ref{cor:diagZero}, we can achieve such compatibility requirements with a geometry of size at most $\log \sigma'_n + \log\log\sigma'_n$.

We now show that the system $\Upsilon=(R,1,r_{s'})$ simulates $\Gamma'=(T',2,s')$ (with tile type scale factor 1).  Consider the function $F(t) = \{r_t\}$ to map tile types from $T'$ to $R$.  Suppose $\Upsilon$ correctly simulates $\Gamma'$ up to the first $i-1$ steps.  That is, for the first $i-1$ assembly sequence steps $\langle (t_1, x_1,y_1), \ldots (t_{i-1},x_{i-1},y_{i-1}) \rangle$ of the unique assembly sequence of  $\Gamma'$,  the first $i-1$ steps of the assembly sequence of $\Upsilon$ are $\langle (r_{t_1}, x_1,y_1), \ldots (r_{t_{i-1}},x_{i-1},y_{i-1}) \rangle$.  Consider the $i^{th}$ step in the assembly sequence of $\Gamma'$, $(t_i, x_i, y_i)$.  We know that the placement position $(x_i,y_i)$ must occur north, west, or east of the previously placed tile position $(x_{i-1}, y_{i-1})$.  If the placement occurs to the north, then $(r_{t_i},x_i,y_i)$ is a valid step in the assembly sequence for $\Upsilon$.  This is because the tile transformation explicitly assigns tile type $r_{t_{i-1}}$ a strength-1 north glue that matches the south glue of $r_{t_{i}}$, as $t_{i-1}$ must have a north strength-2 glue that matches the south glue of $t_i$.  Further, by the the fact that $\Gamma'$ is nice, there are no exposed east/west glue faces of the $\Gamma'$ assembly after $i-1$ steps, and thus the $\Upsilon$ assembly will have no other exposed glues (of strength-1) beyond the single north glue of $r_{t_{i-1}}$.  This implies that the placement of $r_{t_i}$ at position $(x_i,y_i)$ is the only possible next tile placed for system $\Upsilon$.

 Now suppose the $i^{th}$ tile attachment occurs to the east of the previously placed tile in the $\Gamma'$ assembly.  The required strength-2 attachment threshold for the $i^{th}$ tile can be achieved by two strength-1 glues, or by a single strength-2 glue.  In the cooperative strength-1 glue case, we know that $r_{t_i}$ is the only tile in $R$ that both matches the east glue of $r_{t_{i-1}}$ and has a compatible geometry with the north face of the tile type at position $(x_{i}, y_{i}-1)$.  Thus, $(r_{t_{i}},  x_i, y_i)$ is a valid next element of $\Upsilon$'s assembly sequence, and is the only valid next element by the nice properties of $\Gamma'$.  For the case of a strength-2 east attachment, we know that the $(r_{t_{i}},  x_i, y_i)$ attachment is valid for $\Upsilon$ because of the the matching strength-1 east glue of $r_{t_{i-1}}$ and west glue of $r_{t_{i-1}}$, and because the tile type south of $(x_i,y_i)$, if there is one, is guaranteed to have a compatible north geometry with $r_{t_{i}}$ by the ``no mismatched glue'' property of nice zig-zag systems.  The attachment is also unique because of the ``no exposed glues'' property of nice zig-zag systems.  The remaining west attachment case is analogous to the east attachment case.

 Therefore, $\Upsilon=(R,1,r_{s'})$ simulates $\Gamma'=(T',2,s')$ with tile type scale $|R|/|T'|=1$, and therefore also simulates $\Gamma=(T,2,s)$ with tile type scale $O(1)$.  The size of the geometry of the simulation is at most $\log \sigma'_n + \log\log\sigma'_n$ by Theorem~\ref{cor:diagZero}, and is thus $\log {2\sigma_n} + \log\log{2\sigma_n} = \log\sigma_n + \log\log\sigma_n + O(1)$ where $\sigma_n$ is the number of north/south glue types in $T$.
\end{proof}

\begin{proof} of Theorem~\ref{thm:zigzagSimulation1Glue}:
We use the same approach for simulating nice zig-zag systems as given in the proof for Theorem~\ref{thm:zigzagSimulation}, with the modification that all non-null glue types from system $T'$ are mapped to a single strength-1 glue type $x$.  Further, rather than empty set geometry assigned to all east west glues in the simulation set, we assign geometries that satisfy a compatibility matrix which assigns a 0 to entries which correspond to identical glues, and 1 to entries corresponding to non-identical glues.  Thus, while there is a single glue, only the appropriate tile with geometry representing the appropriate east/west glue will attach.  By the same analysis given for Theorem~\ref{thm:zigzagSimulation1Glue}, we achieve a simulation set for any zig-zag system with $O(1)$ tile type scale and $\log\sigma +\log\log\sigma + O(1)$ geometry size.
\end{proof}

\section{Additional details for 2GAM Results}
\label{sec:2GAM-results-details}

This section includes additional details and images describing the construction in Section~\ref{sec:2GAM-results}.

\begin{figure}[htp]
\begin{center}
    \includegraphics[width=7.0in]{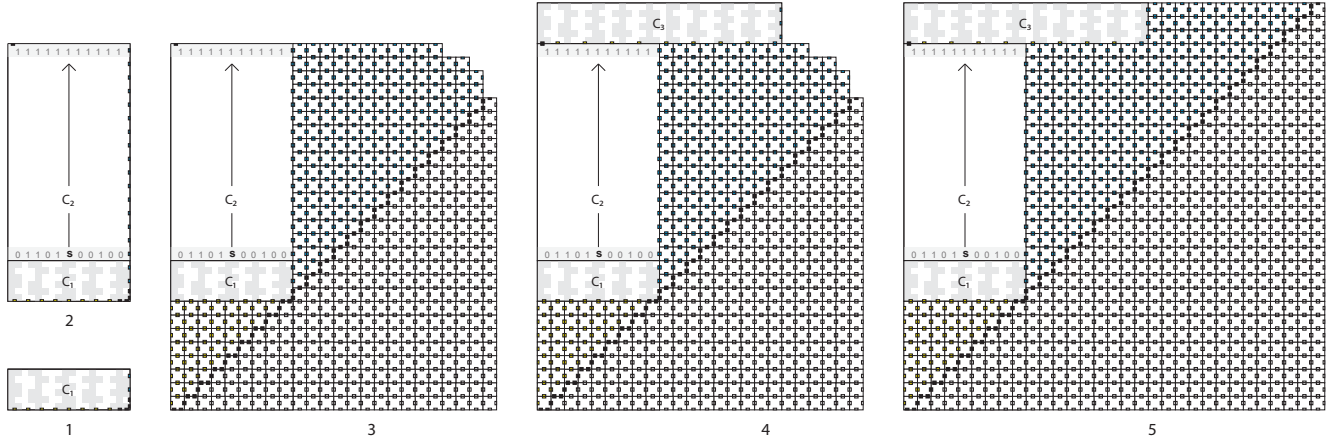} \caption{\label{fig:full-square-assembly-sequence} \footnotesize A sketch of the assembly sequence of the square and its components:  1. first the counter module $C_1$ assembles; 2. the initial value $s$ is seeded by $C_1$, which allows $C_2$ to assemble; 3. filler tiles complete a sufficient portion of the square; 4. A fully formed version of $C_3$ attaches; and 5. filler tiles complete the square.}
\end{center}
\end{figure}

\subsection{The $2$-handed counters $C_1$ and $C_3$}
\label{sec:counters-c1-and-c3}

For each tile in this construction, the bodies are squares of $(2^{n''} + h + 4) \times (2^{n''} + h + 4)$ geometric units.  (See Figure~\ref{fig:complex-geometry-tile-dimensions} for an example.) The north and south geometries consist of $(2^{n''} + h + 4) \times (n'' + 2)$ rectangles of geometric units (i.e. rectangles of length $\ell = 2^{n''} + h + 4$ and width $w = n'' + 2$).  All north geometries are completely empty, while all south geometries are completely filled in.  However, the east and west geometries are composed of $(n'' + 2) \times (2^{n''} + h + 4)$ rectangles of geometric units (i.e. rectangles of length $\ell = 2^{n''} + h + 4$ and width $w = n'' + 2$) which contain intricate collections of gaps, which we call \emph{sockets}, and projections, which we call \emph{prongs}.

Note that while the techniques utilized in this construction can be generalized to form counters of arbitrary bases, the tile complexities are asymptotically identical regardless of the base, so for simplicity of explanation we utilize only base $2$ counters.

\begin{figure}[htp]
\begin{center}
    {\subfloat[{\scriptsize Dimensions of the geometric units of a counter tile for bit position $p$, representing bit $b_1$ on its west side and $b_2$ on its east side.  White portions indicate areas which are not filled in, while grey areas are filled in (although areas $b_1$, $b_2$, and $B_2$ contain a mixture).}]
    {\label{fig:complex-geometry-tile-dimensions}\includegraphics[width=3.0in]{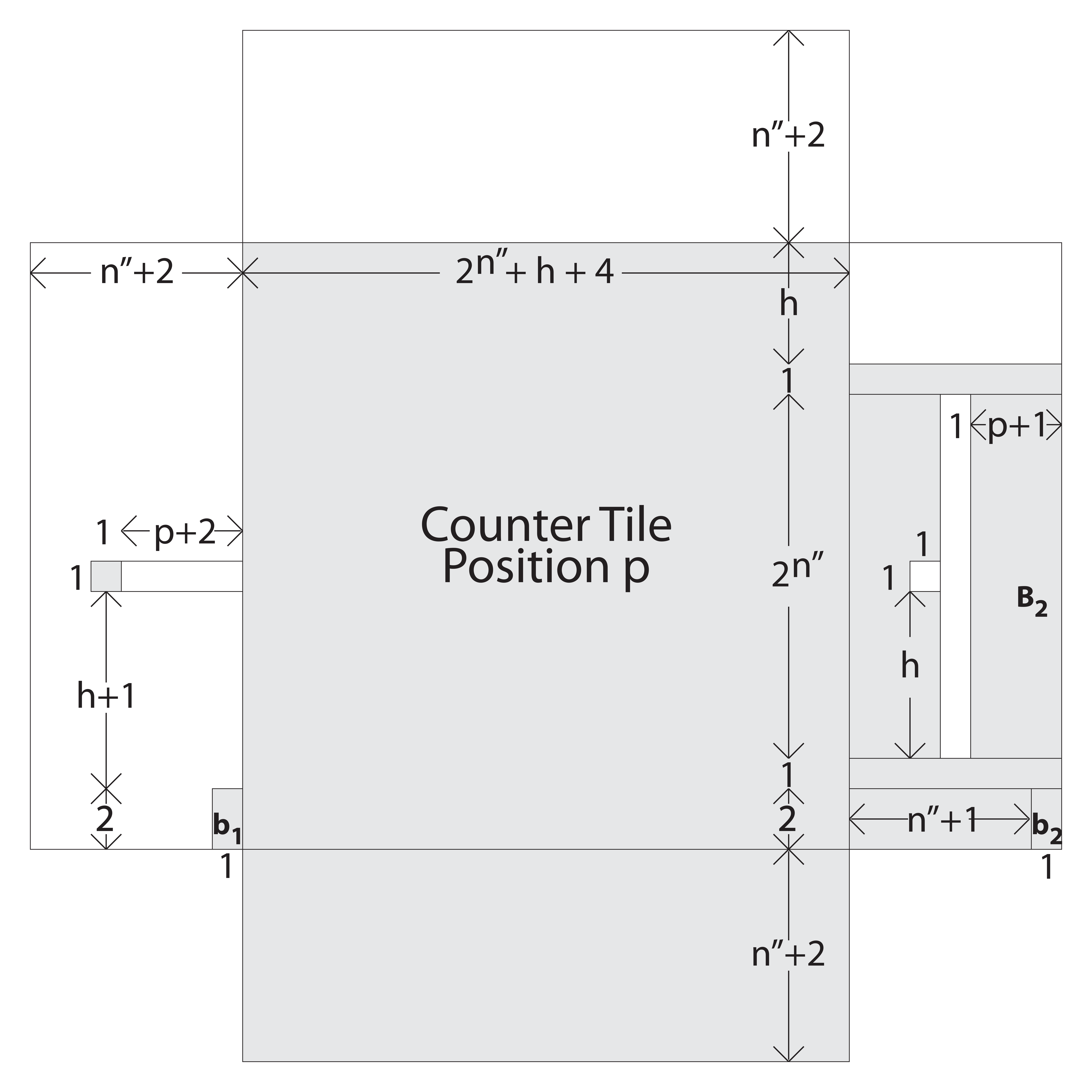}}}
    \quad
    {\subfloat[{\scriptsize Patterns used to represent bit values $b_1$ (left) and $b_2$ (right) on the southern corners of west and east geometries, respectively.}]
    {\label{fig:Soul-patch-pattern}\includegraphics[width=1.0in]{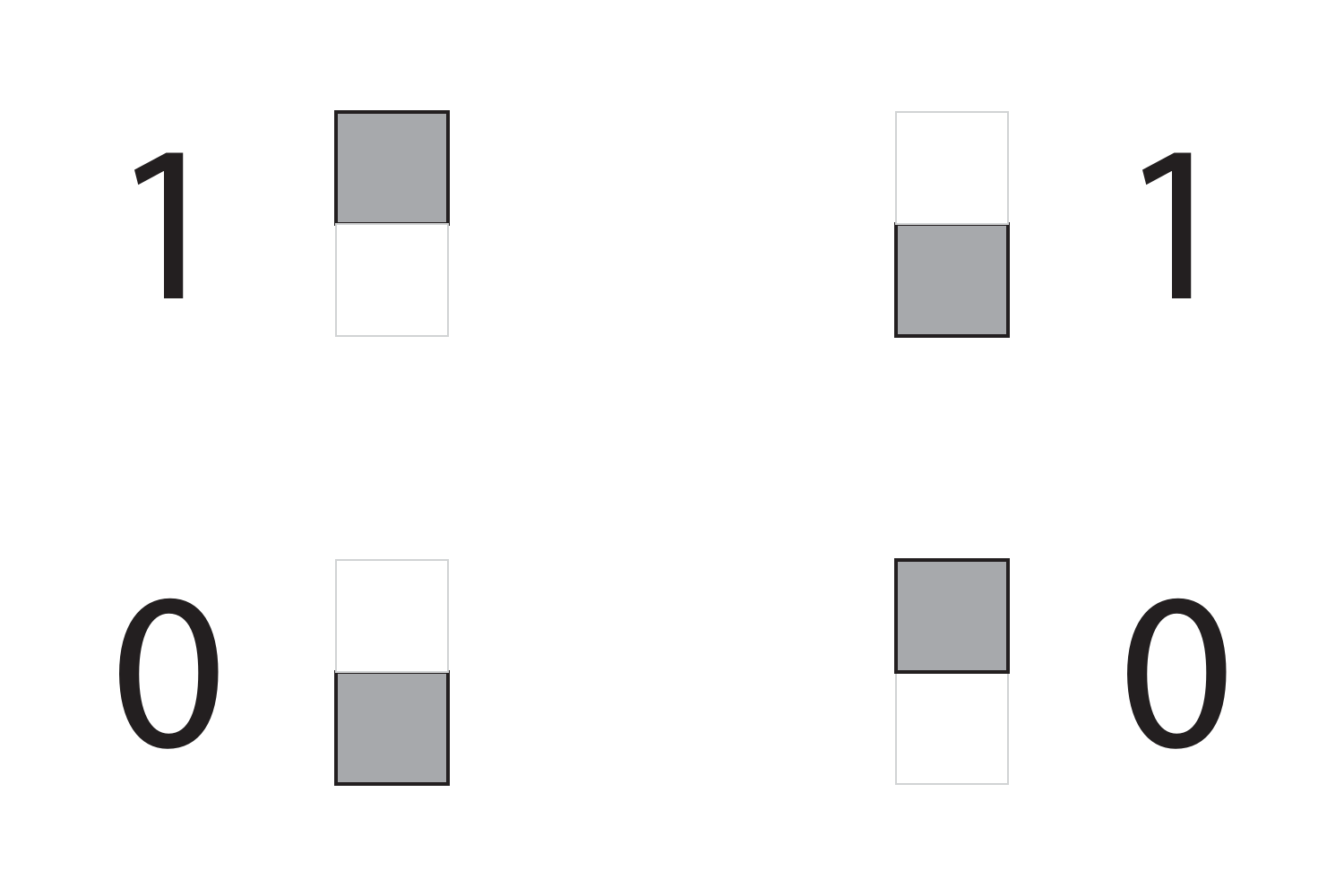}}}
    \quad
    \vspace{-5pt}
    \caption{\footnotesize Details of the geometries of the counter tiles.}
    \label{fig:counter-tiles}
\end{center}
\end{figure}

The purpose of the counter module $C_1$ is to count from $0$ through $2^{n''}-1$, for a total of $2^{n''}$ values, each represented by exactly one column.  The northern glue on the northernmost tile of each of column is used to represent one of the bits of the value $s$ (from left to right, the most significant bit to the least significant).  Each column is composed of $n'' + 4$ tiles arranged vertically.  The tiles in each column perform two tasks: 1. represent an $n''$-bit binary number $m$ for $0 \leq m < 2^{n''}-1$, and 2. present a northern glue which represents the correct value ($0$ or $1$) for the $2^{n''} - (m + 1)$th bit of $s$.  The southern $n''$ tiles, which we call the \emph{counter} tiles, each represent one bit of $m$, from top to bottom the most significant bit to the least significant.  The northernmost $4$ tiles are gadgets called \emph{caps} which serve to represent the bit values on the north of the columns, while ensuring their correct positioning relative to the bits of $s$.  There are two possible caps which can form, a $0$-cap and a $1$-cap.  Either cap can nondeterministically attach to any column of counter tiles to provide the $4$ northernmost tiles of any column (except for the leftmost and rightmost columns, which are special cases discussed later).  This allows for the formation of two versions of the columns which represent each counter value $m$: one that represents a $0$ on the north and one that represents a $1$. It is the purpose of the cap to ensure that only the version of the column with the correct bit of $s$ (i.e. for the column representing the value $m$ and thus the $2^{n''} - (m + 1)$th bit of $s$, the one representing the bit value matching the $2^{n''} - (m + 1)$th bit of $s$) can attach to neighboring columns.

Each column can completely form independently of every other column, and in fact must be fully formed before it can combine with any other column.  For each bit position of the counter, there are two tile types, one representing a $0$ and one a $1$.  Only those for the least significant bit position have strength-$1$ glues on their east and west, while those for every other bit position have $null$ glues on their east and west sides.  However, all counter tiles have geometries on their east and west sides which will be explained shortly.  The north and south sides of these tiles have flat geometries and strength-$2$ glues which allow them to nondeterministically combine with the bit values above and below (but only in the correct order of significance).  Thus, the counter portions of each column can form by nondeterministic combinations of exactly one tile type for each bit position, and therefore columns can form which represent each of the values $0$ through $2^{n''}-1$.  For each counter value, two columns can form: one with a $0$-cap and one with a $1$-cap. The east and west sides of the cap tiles have geometries but $null$ glues, except for the northernmost which has strength-$1$ glues on the east and west.  The north and south sides of each cap tile are similar to those of the counter tiles in that they have only flat geometries and strength-$2$ glues, except for the northernmost cap tile which has a strength-$1$ glue that represents the bit of $s$ provided by that column.  (The rightmost, and therefore least significant, bit of $s$ will be represented by a strength-$2$ glue, and both the most significant and least significant bits will be represented by special glues which convey both the bit values and the most and least significance of those positions.)  The positioning of the two strength-$1$ east and west glues in each column, on only the top and bottom tiles, ensures that the only way that two columns can possibly combine with each other (keeping in mind that this construction operates at temperature $2$) is if they are fully formed (i.e. they contain all $4$ cap tiles and exactly one counter tile for each bit position).

\begin{figure}
  \begin{center}
  {\subfloat[{\scriptsize  Template for the tile set which allows the formation of columns representing binary numbers in the range $1$ through $2^{n''}-2$ (preventing the formation of numbers $0$ and $2^{n''}-1$), and which increment the binary numbers represented by each column from west to east.  The bit significance positions for each group of tiles are shown.  Note that a unique group of tile types is created for every bit position, although those for positions $2$ through $n''-2$ are shown in a single group.  Each tile has strength-$2$ glues on its north and south edges (except for the north edge of the northernmost group and the south edge of the southernmost group).  The east and west bit values are represented in the actual tile set by the east and west geometries.}]{\label{fig:increment-and-edge-detect-tiles}\includegraphics[height=2.6in]{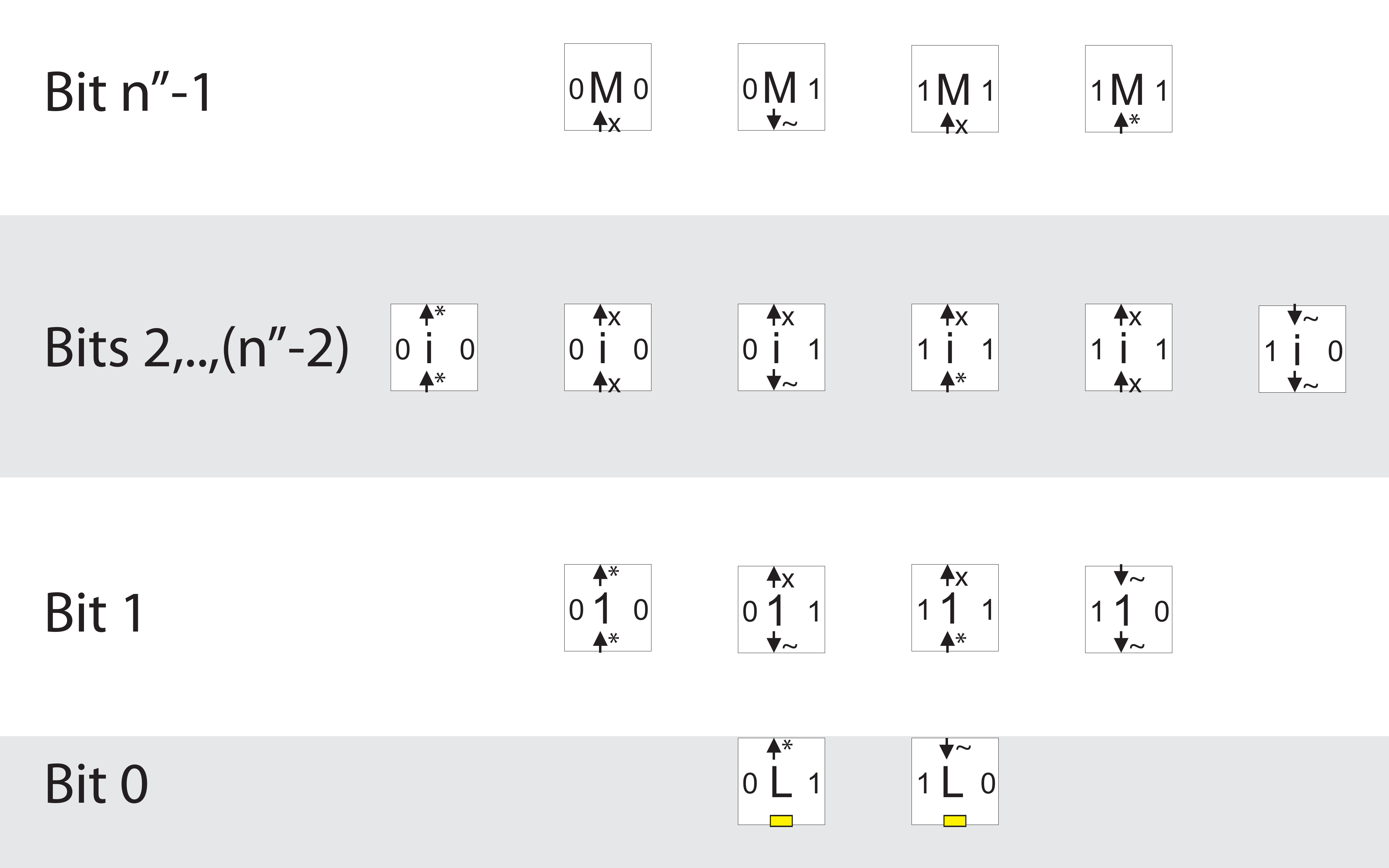}}}
  \quad
  {\subfloat[{\scriptsize Filler tiles which fill in the bulk of the square.}]{\label{fig:filler-tiles}\includegraphics[height=0.9in]{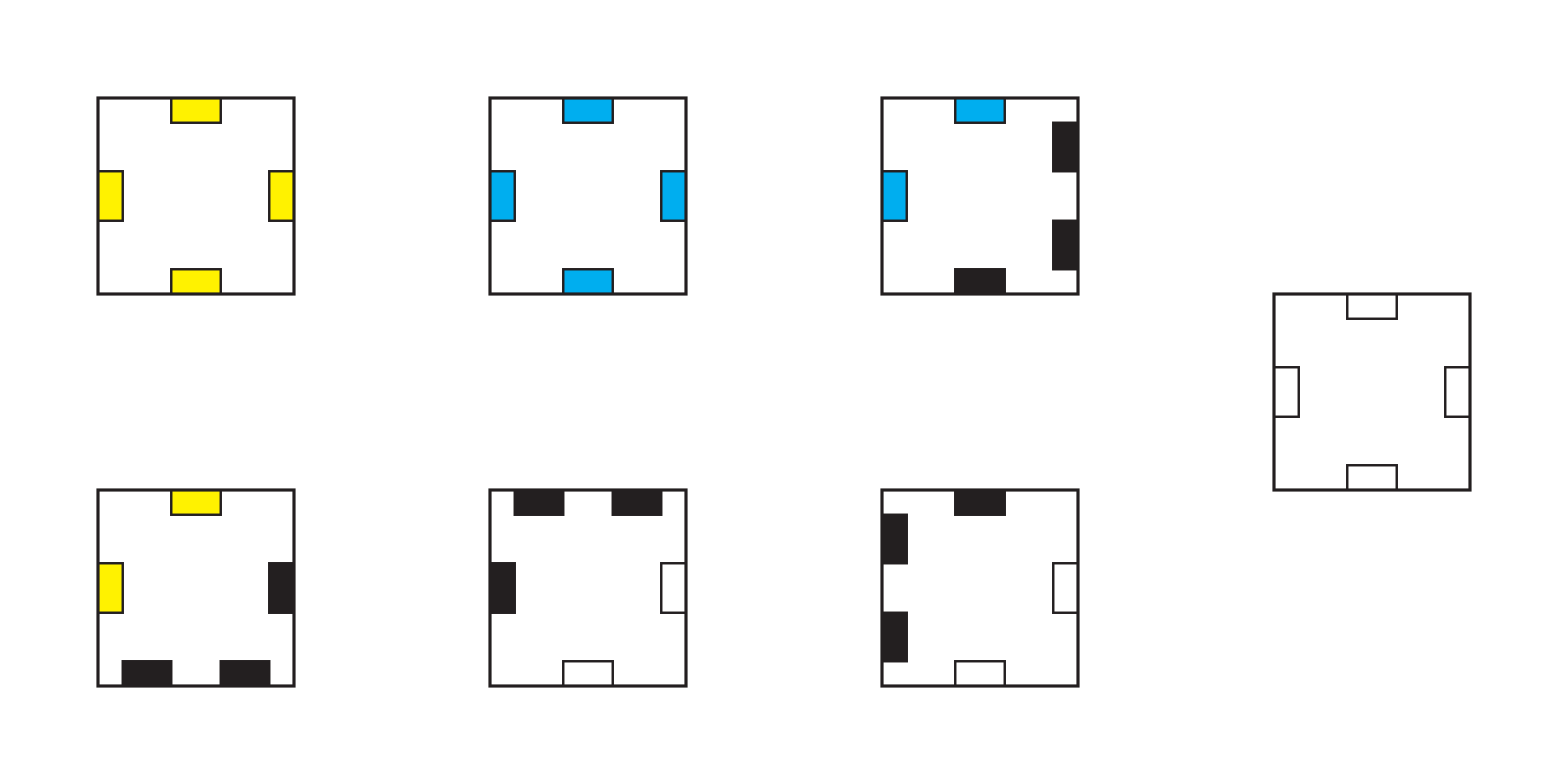}}}
  \quad
  \vspace{-5pt}
  \label{fig:counter-tiles}
  \end{center}
\end{figure}

\subsubsection{Counter tiles}

Figure~\ref{fig:increment-and-edge-detect-tiles} shows a tile set which performs portions of the functionality of the counter tile types.  This tile set is designed to form $1 \times n''$ tile columns representing $n''$-bit binary numbers, with specific groups of tile types designed for various groups of bit positions.  All tiles have strength-$2$ glues on their north and south sides (except for the north glue of those in the most significant bit position and the south glue of those in the least significant bit position) which match only those of the adjacent edges of tiles for adjacent bit positions.  The columns formed represent $n''$-bit binary numbers on their west and east sides such that for each number $x$ represented on the west, the value $x+1$ is represented on the east.  Additionally, every possible $n''$-bit binary value can form \emph{except} for the values $0$ and $2^{n''}-1$ (i.e. all $0$'s or all $1$'s).  Those values are represented by the leftmost and rightmost columns, respectively, of the counter $C_1$ and are formed by a separate set of hard coded tile types since the leftmost column must expose a special glue on the north which marks that bit of $s$ as the most significant bit, and the rightmost column must expose a special strength-$2$ glue on the north which marks that bit of $s$ as the least significant bit, and it must also have flat geometries and strength-$1$ glues along its eastern side which can attach to the filler tiles.

For each bit position $p$ (where $0 \leq p < n''$), in order to represent bit values $b_1$ on the west and $b_2$ on the east, the geometries are designed as shown in Figure~\ref{fig:complex-geometry-tile-dimensions}.  The $1 \times 2$ rectangles labeled $b_1$ and $b_2$ each have exactly one position filled in and one left empty, depending on the bit values being represented and following the patterns shown in Figure~\ref{fig:Soul-patch-pattern}.  It is these regions which enforce the restriction that only if two columns represent the same bits in all positions can they line up and come together completely to combine.  The remaining area of the west geometry is completely empty except for one unit located $h + 3$ units from the southern edge and $p + 2$ units from the east.  This single point is referred to as the prong, which will slide into a complementary socket in the east geometry of a tile in the same bit position of a column to the left.

The portion of the east geometry to the left of region $b_2$ is a $2 \times n''+1$ rectangle which is completely filled in.  The geometry immediately to the north of that creates the socket, which has one empty column $p+1$ units from the east and one additional empty spot $h$ units up that column and one unit to the west.  The area labeled $B_2$ consists of $p+1$ columns which are divided horizontally into a number of blocks dependent upon $p$, the bit position being represented.  It is divided into $2^{n'' - p}$ blocks of width $p+1$ and height $2^p$.  Those blocks can be thought of as being numbered from the southernmost as $0$, through the northernmost as $2^{n'' - p}$.  If $b_2 = 0$, then each even numbered block is empty and each odd numbered block is filled in.  If $b_2 = 1$, then the even numbered blocks are empty and the odd numbered blocks are filled in.  Examples can be seen in Figure~\ref{fig:column-combination-example1}.

Intuitively, the eastern socket is split into two groups of regions - those representing a $0$ and those representing a $1$.  If the socket is part of a $0$-tile (i.e. a tile which represents a $0$ for the particular bit of the counter value), the regions representing $0$ are left open and those representing $1$ are filled in.  For a $1$-tile, the opposite is true.  The number of regions is determined by the bit position.  For the most significant bit position, there are only two regions, with the bottom representing $0$ and the top $1$.  The tile for the next bit down has four regions, which are essentially each of the regions of the previous tile split into $2$ sections, with the bottom of each new pair representing $0$ and the top $1$.  This splitting and doubling continues downward until the tile for the least significant bit which has $2^{n''}$ regions consisting of individual squares.  In a $0$-tile, the bottom of each pair starting from the bottom is left open and the top of each is filled in, and vice versa for a $1$-tile.

The portion of the east geometry above the socket, which is an empty $n'' + 2 \times h$ rectangle, is referred to as the \emph{padding} region.  A padding region exists on the north of each west and east geometry, and its purpose is to guarantee that the geometries of tiles in one bit position of a column will not collide with the geometries of those in adjacent bit positions of a column correctly combining with it.  This padding region is sufficient because the maximum relative translation that two columns which can combine with each other will ever need to experience is a vertical offset of $2^n{''}/2$, and the padding region plus the separation provided by the north-south geometries of combined tiles is at least that large.

\subsubsection{Cap tiles}

The remainder of each column of $C_1$ consists of four tiles called the \emph{cap} which combine to form a vertical column one tile wide and $4$ tiles high.  There are two possible types of cap, a $0$-cap and a $1$-cap, which can nondeterministically attach to any partial column representing a counter value to form the top four tiles of a complete column.  It is the job of a $b$-cap (for $b = 0$ or $1$) to ensure that only if it has attached to a counter value $y$ (which will be positioned at the location of the $(2^{n''} - (y + 1))$th bit of $s$) and $b$ matches that bit of $s$, only then can that column attach to other columns within $C_1$ (specifically, those columns representing $y-1$ and $y+1$) and thus become a portion of that counter.  Otherwise, an unmatching column (i.e. one to which the ``wrong'' cap has nondeterministically attached) is relegated to use within $C_3$, which will be described later.

Please refer to Figure~\ref{fig:cap-tile-geometry} for a graphical depiction of the geometries of the cap tiles. The north and south geometries of all cap tiles are the same as those of the counter tiles:  a completely filled in $(2^{n''} + h + 4) \times (n'' + 2)$ south geometry and a completely empty $(2^{n''} + h + 4) \times (n'' + 2)$ north geometry.  The prong in the west geometry of cap tile $0$ is located $h + 3$ units from the souther and the prong in east geometry of cap tile $1$ is located $2h + 3$ units from the north.  Each consists of a single filled-in location $3$ units away from the bodies.  The remainder of those geometries are empty.  Additionally, the west geometry of cap tile $2$ and the east geometry of cap tile $3$ are completely empty.  Thus, it is the $4$ sockets, one on each cap tile, which provide the bulk of the functionality of the caps.

The east geometries of cap tiles $0$ and $2$ each have two empty rows on the south and a buffer region on the north.  In between are sockets which each have one completely filled-in row above and below them.  The west geometries of cap tiles $1$ and $3$ each have one filled-in row on the south, a buffer region on the north, two empty rows south of the buffer, one filled-in row south of that, and sockets in between the two filled-in rows.  The socket patterns of the cap tiles occupy a single column and are referred to as $S_1$ and $S_2$, plus their complements $\overline{S_1}$ and $\overline{S_2}$.  (Note that the sockets $\overline{S_1}$ and $\overline{S_2}$ are not ``utilized'' during the combination of columns which form into $C_1$, but only for those which become part of $C_3$.)  These columns are $2^{n''}$ units tall, and by associating each of these units with one of the $2^{n''}$ possible values of a counter column to which a cap can attach, it is possible to utilize the pattern of gaps and blocks in a socket to determine which type of cap ($0$ or $1$) can be attached to a particular counter value while allowing that column to become part of $C_1$.  By virtue of the fact that the prongs are relatively short compared with those of the counter tiles, and due to the patterns of sockets and prongs on the counter tiles, compatible columns (which can only combine if they are completely formed, so we only consider completely formed columns in this discussion) are forced to align in such a way that, during their combination, the prong of a cap which is incident upon the socket of a neighboring cap will be at a position which corresponds to the value represented by the counter value.  In this way, we can allow the prong to slide into an ``accepting'' gap if and only if the cap type is correct for that value, namely if that counter value should be associated with a $0$ or $1$ cap and thus a $0$ or $1$ bit of the number $s$.

To provide this functionality, the socket patterns for a cap of type $b$ are the following. $S_1$ treats the $2^{n''}$ units as representing each value $k$ for $0 \leq k < 2^{n''}-1$ counting downward from the top to bottom.  For the $k$th unit, if $k$ coincides with a numbered bit position of $s$ which has the bit value $b$, there is a gap, otherwise it is filled in.  $\overline{S_1}$ is the exact complement (which is equivalent to designing $S_1$ for bit value $1 - b$).  $S_2$ treats the $2^{n''}$ units as representing one ``invalid'' value plus each value $k$ for $0 \leq k < 2^{n''}-2$, counting upward from the bottom to top with the first position being the invalid, and thus always filled-in, position.  Again, for the $k$th unit, if $k$ coincides with a numbered bit position of $s$ which has the bit value $b$, there is a gap, otherwise it is filled in, and again, $\overline{S_2}$ is the exact complement (with the exception that for both the invalid location is filled in).

Intuitively, the reason for the north-south reversal and offset of $1$ for socket values between the pairs of sockets $S_1$ and $S_2$ is that the prongs that will be validating the combination of counter-to-cap values will be positioned (during the combination of a pair of columns) relative to a single counter value, and thus the sockets need to use that one value to validate the two columns which represent (in the correct case) two consecutive values.  Additionally, the motion of the prongs relative to the respective sockets will be vertically reversed due to the fact that they are on opposite sides of the combining columns.  (Also, recall that whenever two columns attempt to combine which do not encode consecutive values, the patterns encoded in regions $b_1$ and $b_2$ will prevent their combination.)  Essentially, $S_2$ and $\overline{S_2}$ are designed so that if they are included in a column which encodes the counter value $y$, the positioning of the prong from a complementary column to the right, which would represent the counter value $y + 1$, would align it with holes if and only if the bit value of the $((y + 1) - 1)$th, or $y$th, bit of $s$ equals $b$.

\subsubsection{Buffer columns and $C_3$}

As previously mentioned, every counter column can form in two versions, one with a $0$-cap and one with a $1$-cap. The design of the glues and geometries is sufficient to guarantee that a fully formed $C_1$ can and will include exactly one version of each column, namely that which will present the correct value for the corresponding bit of $s$.  In order to ensure that this construction is directed, and thus has exactly one terminal assembly, the versions of columns which are not included in the formation of $C_1$ must be included somewhere within the eventual $n \times n$ square.  That is the purpose of $C_3$.

Intuitively, $C_3$ is a counter of approximately double the length of $C_1$ which incorporates all of the ``bad'' columns (meaning those with the wrong caps to be included in $C_1$), but since they cannot directly combine with each other, each pair is separated by a \emph{buffer} column that not only allows them to combine but ensures that they are in fact bad columns (and hence the approximate doubling of the length as compared to $C_1$).

Each buffer column is formed of two portions, similar to the other columns.  However, the bottom portion of buffer columns simply represent a single binary value from the range $1$ through $2^{n''} - 1$, inclusive.  This value is represented on both the west and east sides rather than an incremented value appearing on the east side as in the counter columns, and the geometries which represent each bit are the same as those for the counter tiles.  The logic provided by the north-south glues must simply ensure that columns representing every binary number from $1$ through $2^{n''}-1$ can form, thus simply excluding a column of all $0$'s.  (This is because the counter column representing $0$ encodes the number $1$ on its east, and thus the buffer column need never combine with a column representing the number $0$.)  Thus, there will be a buffer column which can be sandwiched by every consecutive pair of bad columns.

There is only one version of each of the buffer cap tiles, which can be seen in Figure~\ref{fig:cap-tile-geometry}.  Similar to the tiles of the counter columns, the only non $\null$ glues on the west and east edges are strength-$1$ glues on the east and west of the buffer counter tiles in the least significant bit position (the southernmost in each column) and the buffer cap tile $3$ (the northernmost in each column).  As shown, the only ``interesting'' geometry is that of the prongs on buffer cap tile $2$ and $3$.  These prongs are positioned so that they must fit into the sockets of complementary counter columns which allow combinations of counter values with cap values that are the opposite of those required for inclusion in $C_1$.  In this way, full buffer columns can combine on both sides with exactly those full counter columns which received the ``wrong'' cap and in the correct order to count from $0$ through $2^{n''}-1$ and therefore construct $C_3$.

\begin{figure}[htp]
\begin{center}
    \includegraphics[height=8.5in]{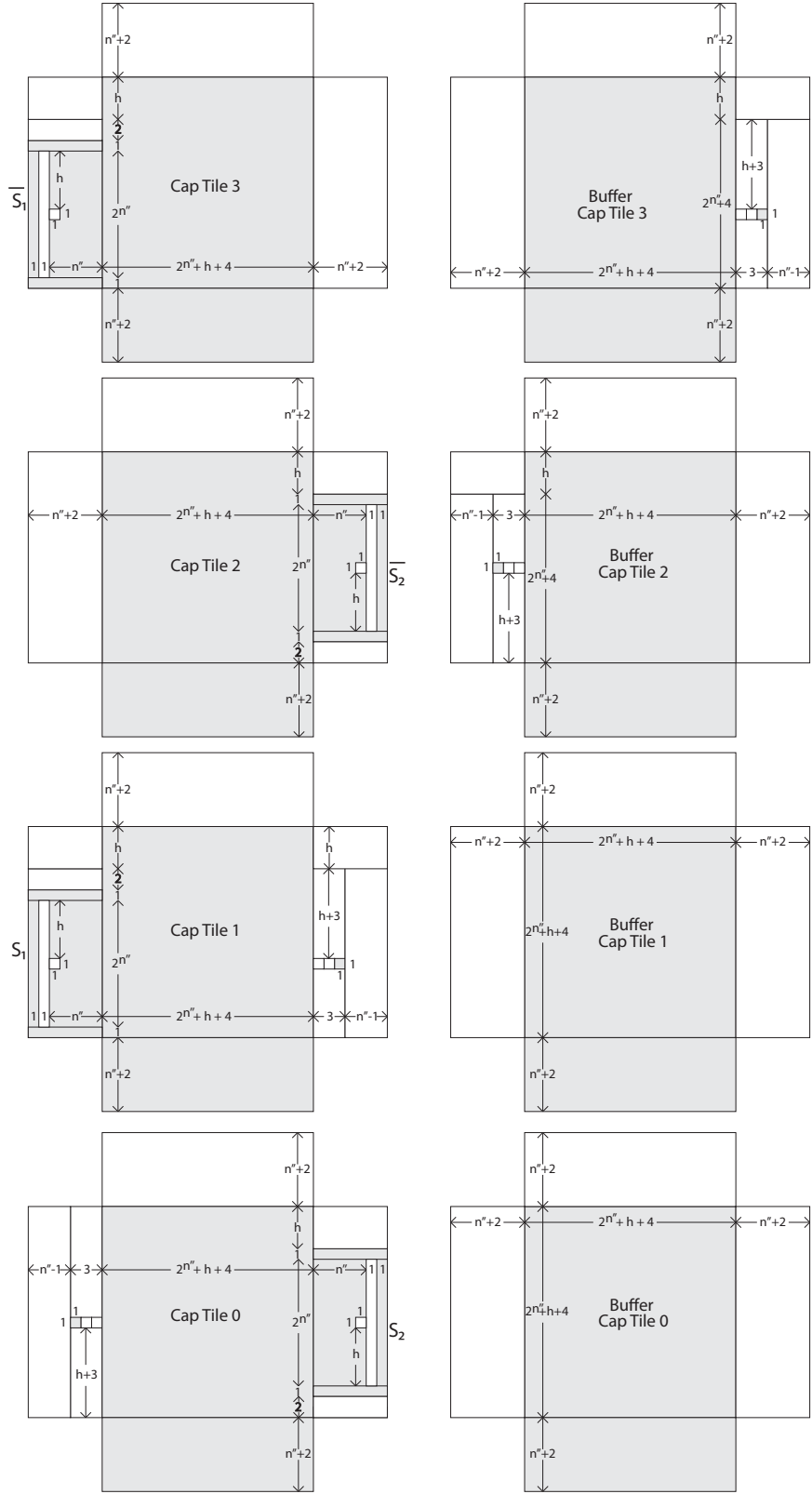} \caption{\label{fig:cap-tile-geometry} \footnotesize Dimensions of the features in the geometries of the counter column cap tiles as well as the buffer column cap tiles. White portions indicate areas which are not filled in, while grey areas are filled in (although socket areas $S_1$, $S_2$, $\overline{S_1}$, and $\overline{S_2}$ contain a mixture}.
\end{center}
\end{figure}

\subsection{Details of the counter $C_2$}

Once $C_1$ has completely formed, the binary number $s$ is encoded on its northern face as a series of strength-$1$ glues, along with a strength-$2$ glue on the north of the east-most tile, which is the least significant bit of $s$.  To these glues, tiles of the types shown in Figure~\ref{fig:c2-counter-tiles} attach to perform the counting of $s$ through $2^{n'}-1$.  This is a basic counter which can assemble in both the standard aTAM as well as the 2HAM, resulting in the same produced assembly since it is polyomino-safe (see \cite{Winfree06}).  It is notable that as each row assembles, it ``checks'' whether or not it is the maximal value of the counter and, if so, causes a tile to be placed in the west-most position which has a special north-facing glue that will eventually allow $C_3$ to attach.

\begin{figure}[htp]
\begin{center}
    \includegraphics[width=4.0in]{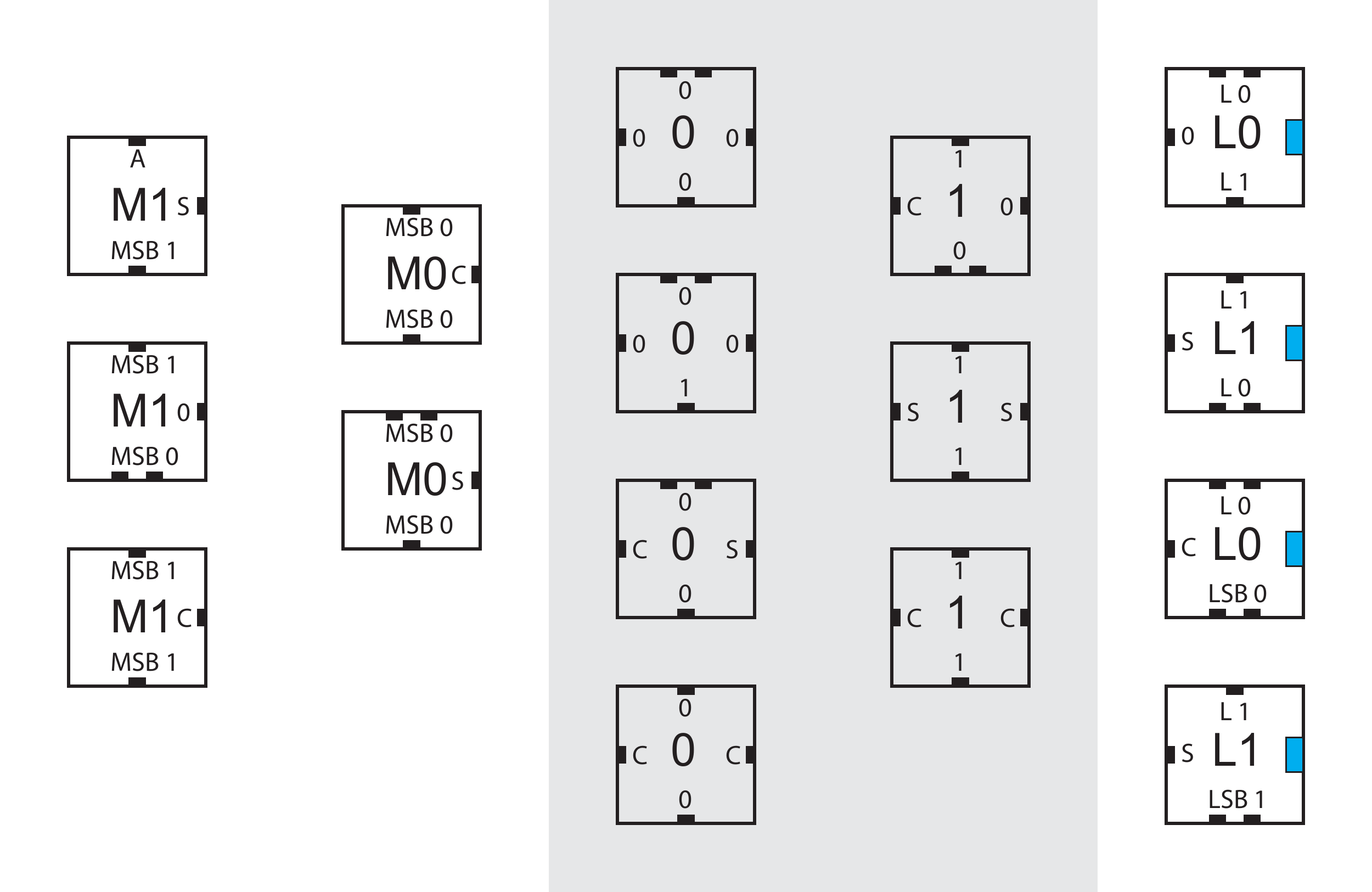} \caption{\label{fig:c2-counter-tiles} \footnotesize The tile types which form the counter $C_2$, grouped by significance of bit position (i.e. tile types for the most significant bit position are in the leftmost group, those for the ``interior'' bit positions are in the middle group, and those for the least significant bit position are in the rightmost group.}
\end{center}
\end{figure}

\subsection{Completion of the square}

Figure~\ref{fig:full-square-assembly-sequence} shows the sequence in which the full assembly can progress.  Once $C_1$ has completely formed and $C_2$ has completely grown northward from it, filler tiles are able to complete a sufficient portion of the rest of the square that a fully formed (and only a fully formed) $C_3$ can attach.  The filler tiles are shown in Figure~\ref{fig:filler-tiles}.  Once $C_3$ has attached, the remaining filler tiles are able to complete the formation of the full $n \times n$ square.

\subsection{An example: constructing a $250 \times 250$ square}
\label{sec:2HAM-square-example}

Here we provide an example of portions of this construction as applied to a $250 \times 250$ square.

\begin{itemize}
  \item $n$: 250
  \item $n'$: $\lceil \log n \rceil = \lceil \log 250 \rceil = 8$
  \item $n''$: $\lceil \log n' \rceil = \lceil \log 8 \rceil = 3$
  \item $s$: $2^{n'} + 2^{n''} + 2n'' + 8 - n = 2^8 + 2^3 + 2(3) + 8 - 250 = 28 = 00011100_2$
  \item $h$: $2^{n''-1} - 1 = 2^2 - 1 = 3$
  \item $C_1$: $2$-handed counter which counts from $0$ through $2^{n''}-1 = 2^3 - 1 = 7$ for a total of $2^3 = 8$ columns
  \item $C_2$: standard counter which counts from $s = 28$ through $2^{n'}-1 = 2^8 - 1 = 255$ for a total of $2^{n'} - s = 256 - 28 = 228$ columns
  \item $C_3$: $2$-handed counter with ``buffer'' columns which counts from $0$ through $2^3-1 = 7$ for a total of $2^{n''+1}-1 = 2^4 - 1 = 15$ columns
\end{itemize}

Figure~\ref{fig:example-dimensions} shows the dimensions for the components of the square for this example, as well as the dimensions of the tile bodies and geometries of counter tiles.

\begin{figure}
  \begin{center}
  {\subfloat[{\scriptsize Dimensions of the modules of the square.}]{\label{fig:square-dimensions-example}\includegraphics[height=3.0in]{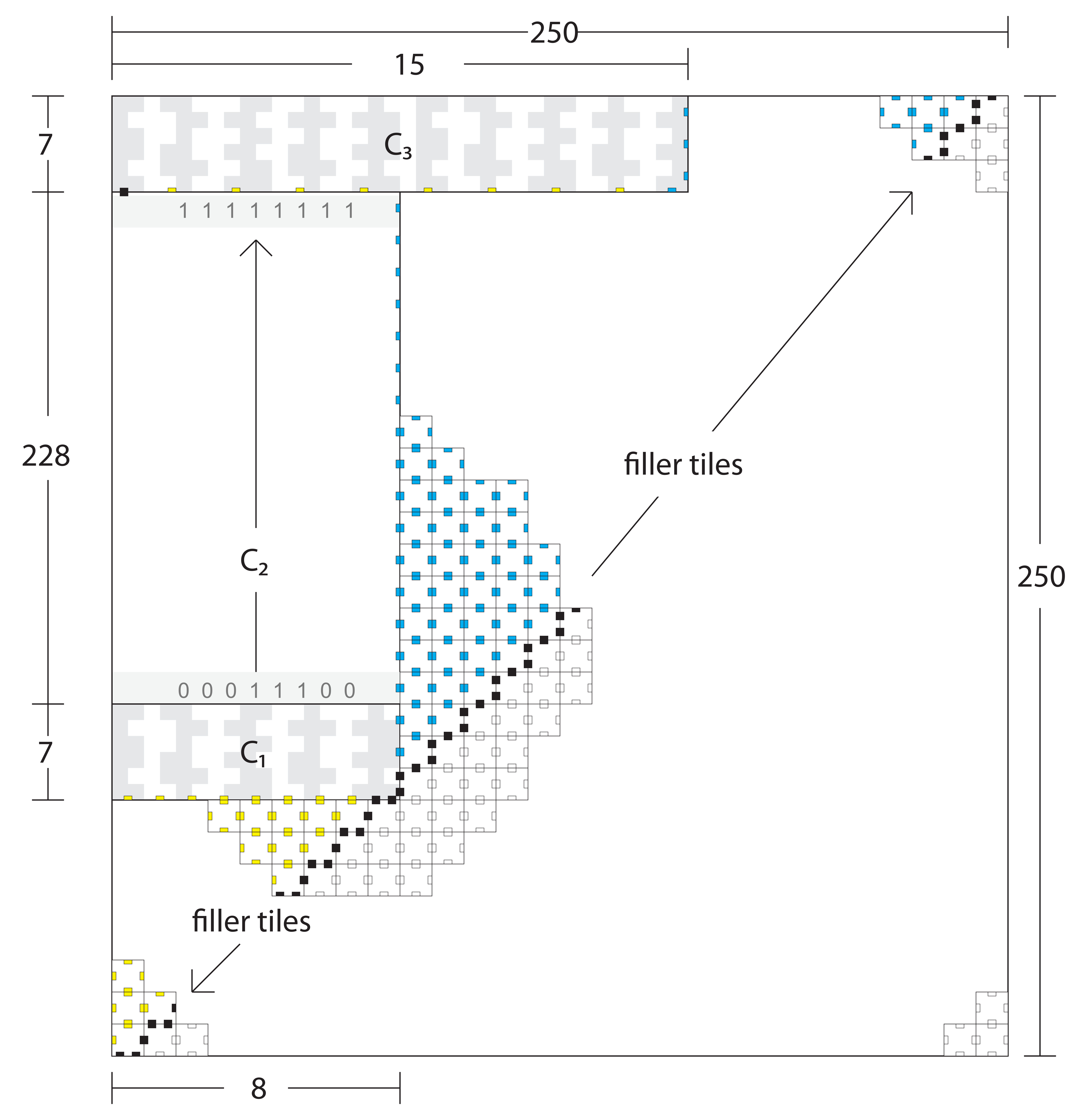}}}
  \quad
  {\subfloat[{\scriptsize Dimensions of the geometries of the counter tiles.}]{\label{fig:tile-dimensions-example}\includegraphics[height=3.0in]{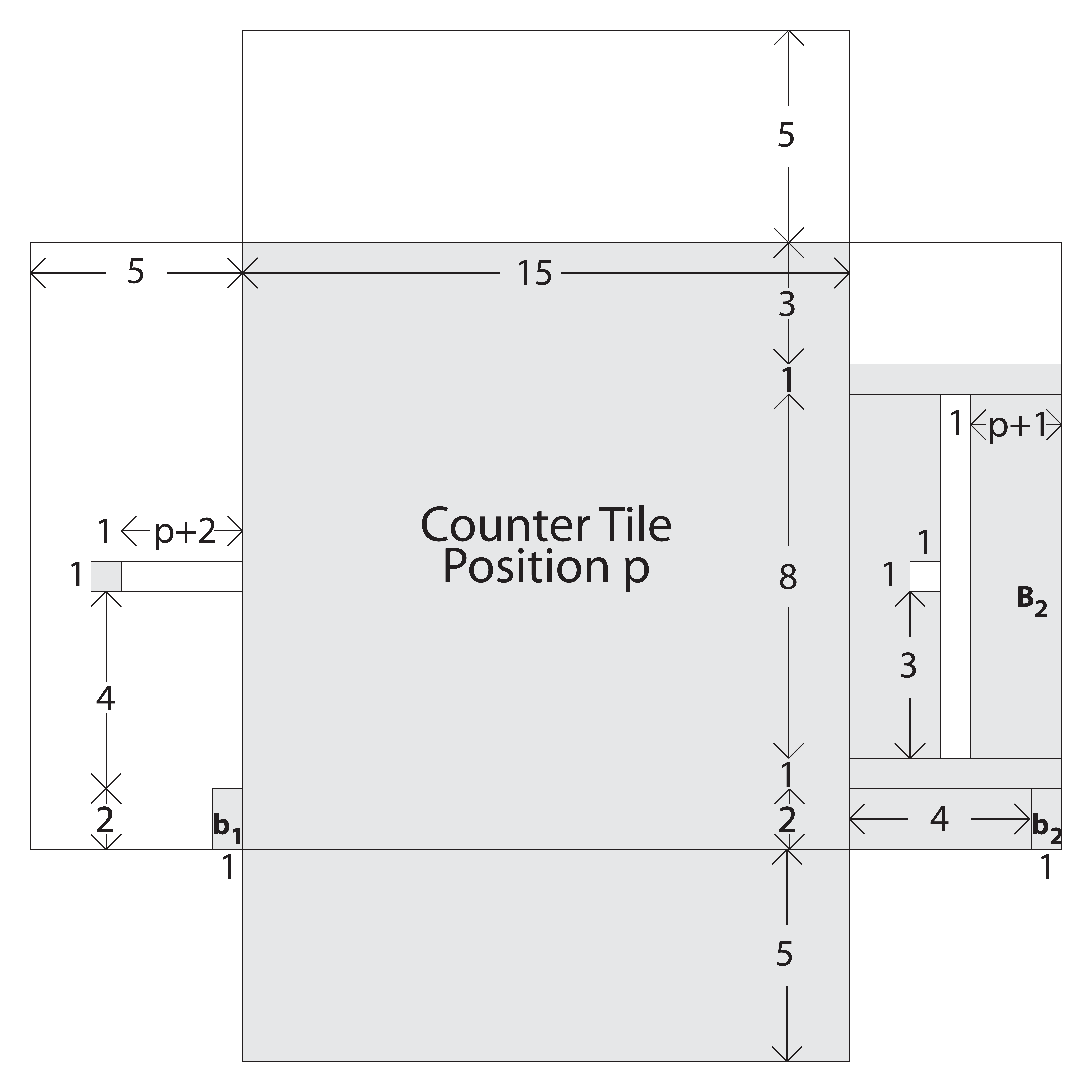}}}
  \quad
  \vspace{-5pt}
  \caption{\footnotesize Dimensions of the components of a $250 \times 250$ square in the $2$GAM construction.}
  \label{fig:example-dimensions}
  \end{center}
\end{figure}

The geometries of the counter tiles are a straightforward application of the generic definitions for the counter tile dimensions, which are based on the value of $n''$ and $h$, to form versions which represent $0$ and $1$ for each bit position.  However, the geometries of the cap tiles are also dependent upon the binary value of $s$ (padded to $8$ bits), in this case ``$00011100$'', since they perform a mapping of counter values to bit values of $s$.  In this case, the geometries of socket $S_1$ in the $0$-cap, which has positions for all $8$ possible counter values, must have openings in exactly the positions $0,1,5,6,$ and $7$ (and hence filled-in locations at positions $2,3,$ and $4$).  Recall that the positions are numbered from north to south for the $S_1$ sockets.  The $1$-cap should have the exact opposite pattern of openings and filled-in locations for its socket $S_1$, as should the socket $\overline{S_1}$ of the $0$-cap.  Clearly, the socket $\overline{S_1}$ of the $1$-cap should match the $0$-cap's $S_1$ socket.

 \begin{table*}\small
\centering
\tabcolsep=0.5\tabcolsep
\begin{tabular}{l|c|c|c|c|c|c|c|c|c|}

\textbf{bit position} & 0 & 1 & 2 & 3 & 4 & 5 & 6 & 7
\\ \hline
\textbf{bit of s} & 0 & 0 & 0 & 1 & 1 & 1 & 0 & 0
\\ \hline
\rule{0cm}{0.4cm} \textbf{$S_1$ for a $0$-cap} & & & X & X & X & & &
\\ \hline
\rule{0cm}{0.4cm} \textbf{$\overline{S_1}$ for a $0$-cap} & X & X & & & & X & X & X
\\ \hline
\rule{0cm}{0.4cm} \textbf{$S_2$ for a $0$-cap} & X & & & & X & X & X &
\\ \hline
\rule{0cm}{0.4cm} \textbf{$\overline{S_2}$ for a $0$-cap} & X & X & X & X & & & & X
\\ \hline
\rule{0cm}{0.4cm} \textbf{$S_1$ for a $1$-cap} & X & X & & & & X & X & X
\\ \hline
\rule{0cm}{0.4cm} \textbf{$\overline{S_1}$ for a $1$-cap} & & & X & X & X & & &
\\ \hline
\rule{0cm}{0.4cm} \textbf{$S_2$ for a $1$-cap} & X & X & X & X & & & & X
\\ \hline
\rule{0cm}{0.4cm} \textbf{$\overline{S_2}$ for a $1$-cap} & X & & & & X & X & X &
\\ \hline

\end{tabular}
\caption{Patterns for each combination of socket and cap type.  Each ''X'' represents a location which is filled-in.  Note that the numbering for the $S_1$ sockets starts at the north and increases southward, while it is the opposite for the $S_2$ sockets.}
\label{table:socket-patterns}
\end{table*}

The geometries of the $S_2$ sockets perform a mapping of counter values to bit values of $s$ as well, but they perform a ``one-off'' mapping by matching bit values of $s$ with the counter values which are $1$ greater than the bit positions.  Additionally, the direction in which the positions of the socket are labeled is reversed, meaning they are counted from the south to north.  Therefore, the sockets $S_2$ for the $0$-cap and $\overline{S_2}$ for the $1$-cap should have openings at exactly positions $1,2,3,$ and $7$ counting from top to bottom, while the sockets $\overline{S_2}$ for the $0$-cap and $S_2$ for the $1$-cap should have openings at exactly positions $4,5,$ and $6$.  Note that the $0$th position of $S_2$ sockets, due to the one-off mapping, is an invalid position which we always fill in (intuitively, because the counter value used to align with an $S_1$ socket is always $1$ greater than the number represented by the column containing the socket).  Please see Table~\ref{table:socket-patterns} for a full listing of the socket patterns and Figures \ref{fig:column-combination-example1}, \ref{fig:column-combination-example2}, and \ref{fig:column-combination-example-buffer} to view examples of counter and buffer columns.

\begin{figure}[htp]
\begin{center}
    \includegraphics[width=6.5in]{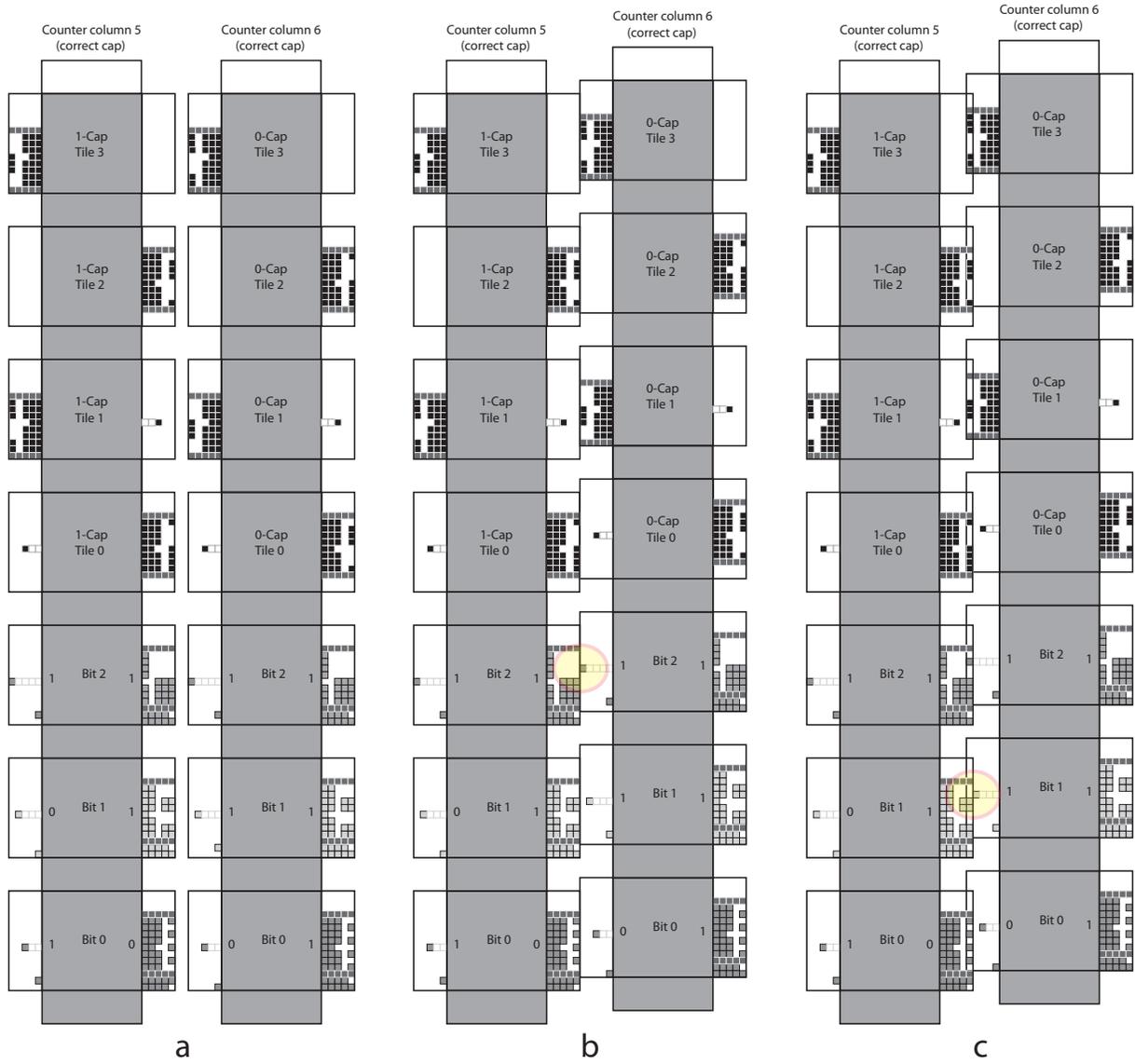} \caption{\label{fig:column-combination-example1} \footnotesize Two counter columns combining during the self-assembly of a $250 \times 250$ square (part 1).  a) Two fully formed and correct counter columns representing the values $5$ and $6$, b) Translation of the right column to allow prong-to-socket alignment for the most significant bit (highlighted), and c) Translation for alignment of the second bit (highlighted).}
\end{center}
\end{figure}

\begin{figure}[htp]
\begin{center}
    \includegraphics[width=6.5in]{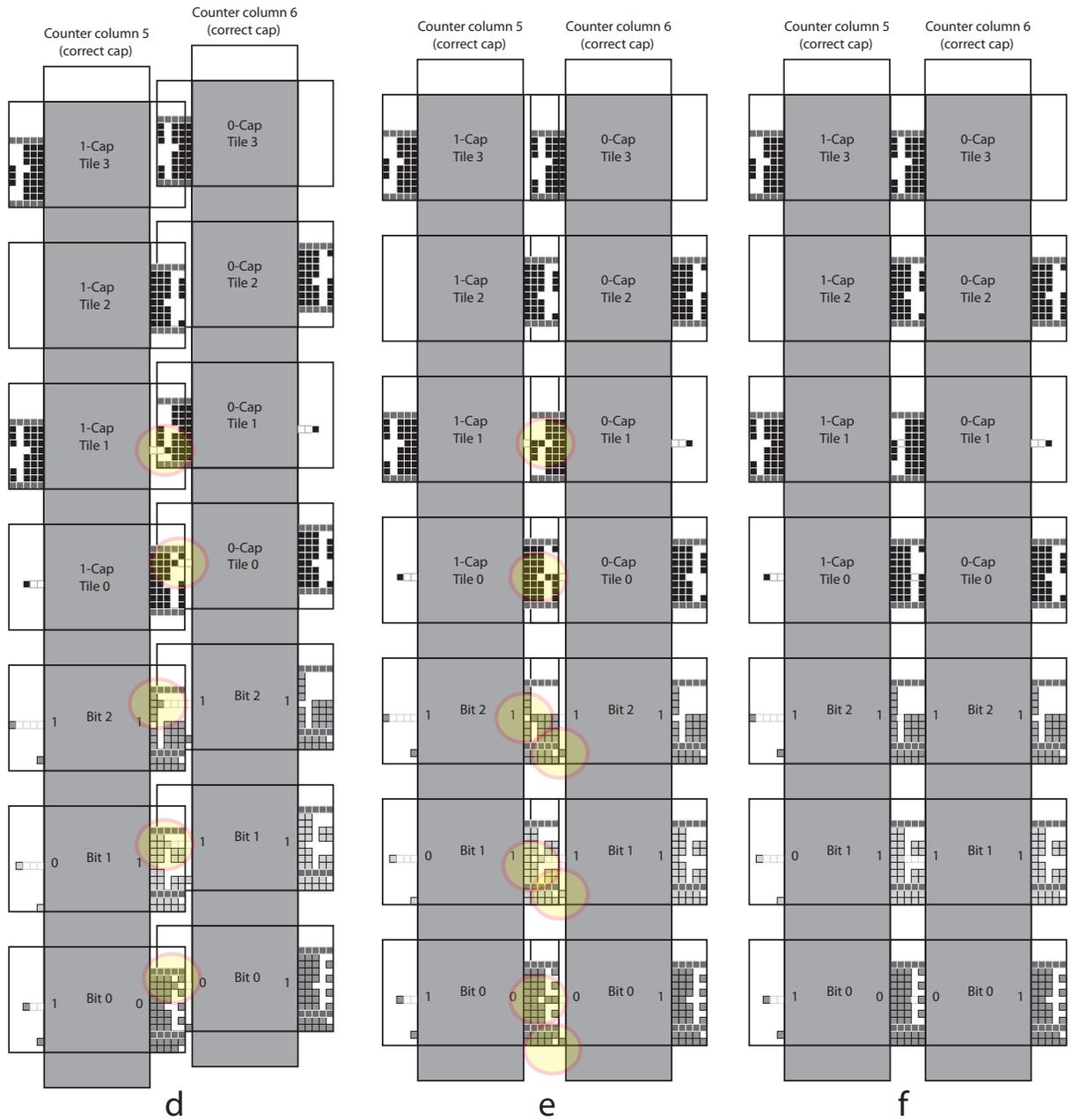} \caption{\label{fig:column-combination-example2} \footnotesize Two counter columns combining during the self-assembly of a $250 \times 250$ square (part 2).  a) After alignment of the final bit, all $4$ prongs (highlighted) are able to slide into the back column of the sockets in the left column, b) Final north-south translation aligns each prong with the west-most gaps in each socket (highlighted) and also aligns the bit patterns of regions $b_1$ and $b_2$ of each pair of counter tiles (also highlighted), c) After the final east-west translation, the counter columns are fully combined, allowing the glues to interact and stably bind them.}
\end{center}
\end{figure}

\begin{figure}[htp]
\begin{center}
    \includegraphics[width=5.5in]{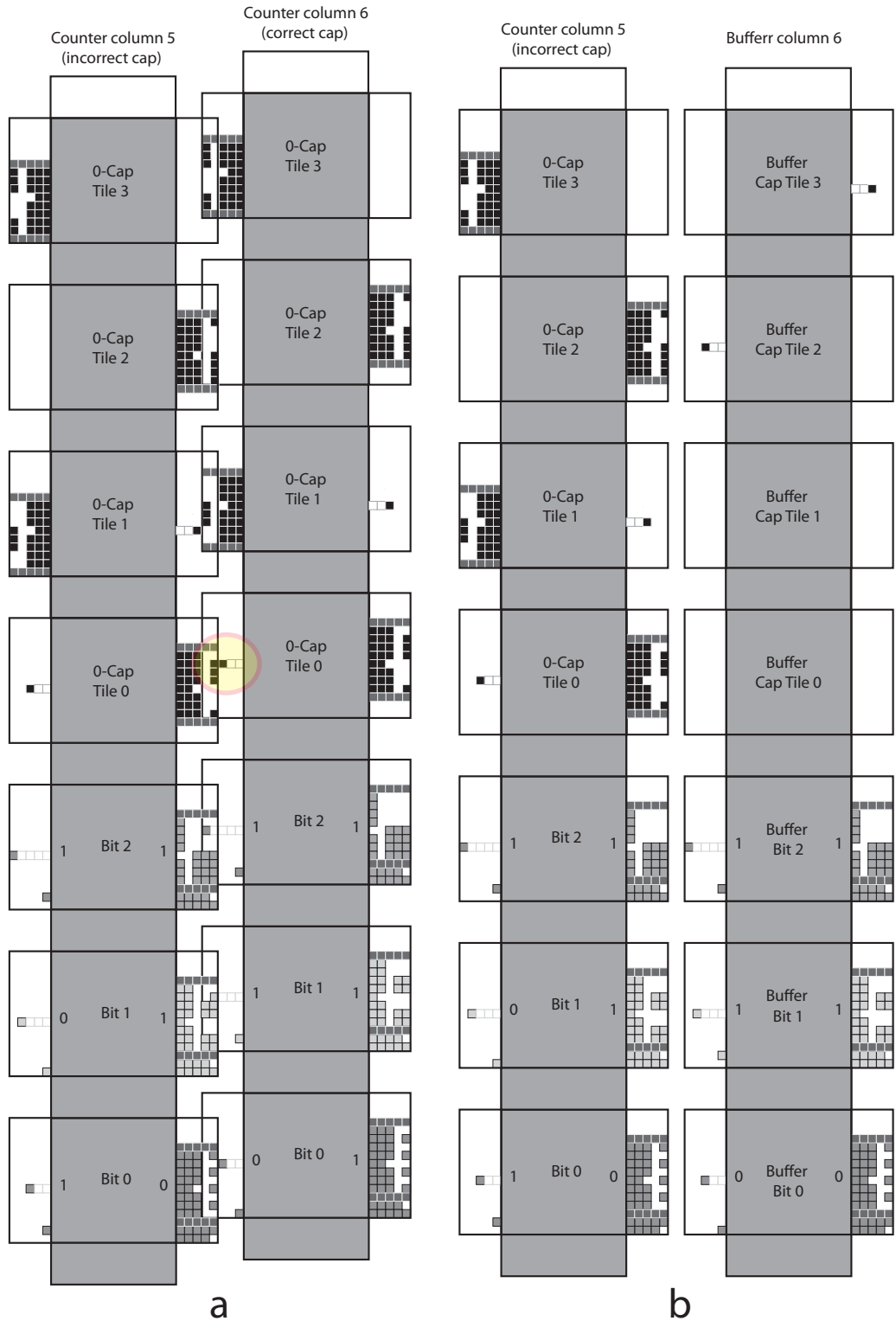} \caption{\label{fig:column-combination-example-buffer} \footnotesize Columns formed during the self-assembly of a $250 \times 250$ square.  a) A counter column with the wrong cap, which prevents its combination with the next consecutive counter column,  b. The incorrect counter column can combine with a buffer column encoding the correct counter value.}
\end{center}
\end{figure}

\subsection{Equivalence of disconnected $2$-dimensional tiles and $3$-dimensional tiles}
\label{sec:2d-conversion-to-3d}

\begin{wrapfigure}{l}{2.3in}
\begin{center}
    \includegraphics[width=2.0in]{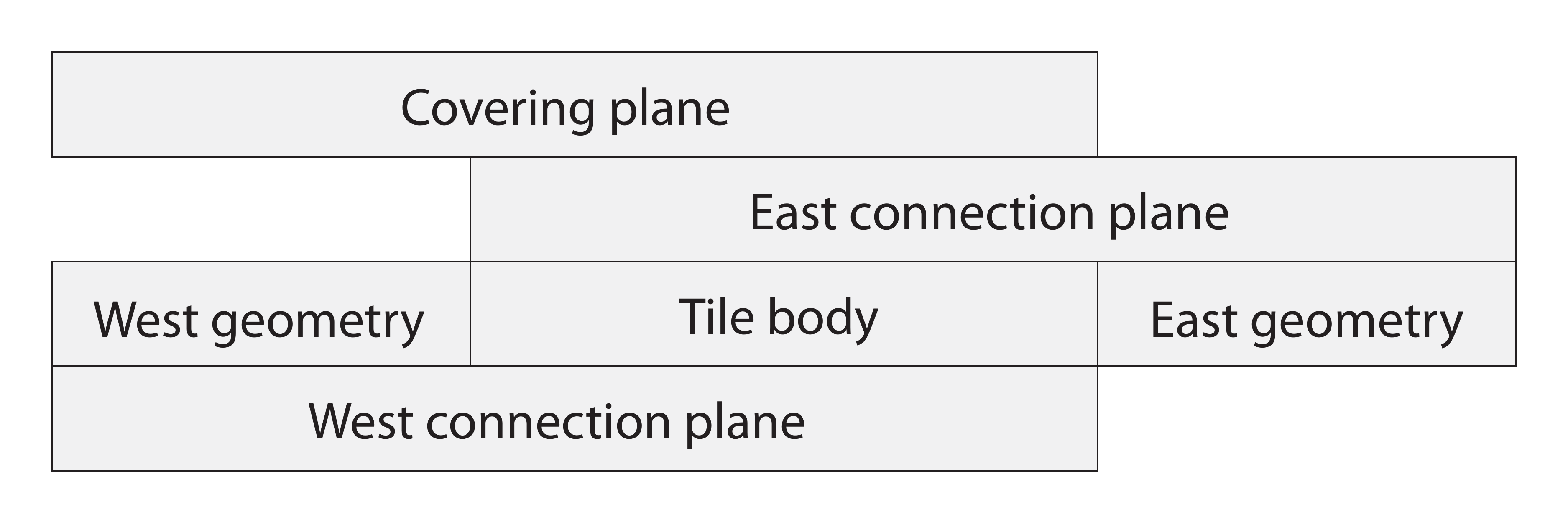} \caption{\label{fig:3-D-tile-side} \footnotesize Side view of equivalent but connected $3$-dimensional geometric tiles.}
\end{center}
\end{wrapfigure}

The construction in this section has been presented as a $2$-dimensional construction meaning that the tiles are $2$-dimensional objects, and due to the constraint of planarity they are never allowed to travel out of the plane during assembly.  However, the geometries defined for the tiles often contain disconnected portions, namely filled-in units which are not connected via an unbroken path of additional filled-in units to their tile bodies.  It is notable that by merely relaxing the restriction of forcing tiles to remain $2$-dimensional, and instead allowing them to extend into three additional planes, this construction can achieve connectivity while removing the explicit restriction that translations of tiles during assembly must only occur within the $x-y$ plane since that restriction instead becomes implicit due to the design.  See Figure~\ref{fig:3-D-tile-side} for a depiction of how the tiles are extended along the $z$-axis.  For each of the east and west geometries a solid filled-in square of units is added in the plane above and below, respectively.  These attach to every unit of the geometries as well as to additional filled-in squares above and below the tile body to ensure full connectivity.  Finally, an additional layer is added which covers the tile body and lies above the west geometry, with exactly one open plane between it and the west geometry.  In this way, it is ensured that all portions of the tiles and their geometries are connected, and exactly the tiles that could connect with each other in the $2$-dimensional planar version can now connect with each other.

\section{Two Function Problem}

In this section, we study the problem of bar-to-bump
reduction for the design of the tile face geometries, namely converting bar geometries to bump geometries (see Section~\ref{sec:geometry-classes}). The goal is to attempt to simplify the geometry patterns as much as possible. Called
the \emph{two function problem}, it takes an input integer $l$, and requires the design of
two functions $f: \{0, 1,\ldots,l\}\rightarrow \{0,1\}^n$ and $g :
\{0, 1,\ldots,l\}\rightarrow \{0,1\}^n$ to satisfy $f(x)\ AND\ g(y)
= 0^n$ if and only if $x + y \le l$, where $AND$ takes as arguments
two binary strings, $x$ and $y$, of matching length and returns the
string of binary values representing the logical AND operation
between each corresponding pair of bits in $x$ and $y$. The target
is to make $n$ as small as possible. It is trivial to design the two
functions with $n=l$. We show a lower bound that exactly matches the
upper bound.

\begin{definition}
Let $l, n\in N$. Define $f : \{0, 1,\ldots,l\}\rightarrow \{0,1\}^n$
and $g : \{0, 1,\ldots,l\}\rightarrow \{0,1\}^n$ where $f$ and $g$
take a number between $0$ and $l$, inclusive, and return a binary
string of length n such that $f(x)\ AND\ g(y) = 0^n$ if and only if
$x + y \le l$. (Note that we define $AND$ as the function which
takes as arguments two binary strings, $x$ and $y$, of matching
length and returns the string of binary values representing the
logical AND operation between each corresponding pair of bits in $x$
and $y$. In this case, there must not be a 1 in the same position of
both bit strings.) Goal: For a given $l$, find $f$ and $g$ such that
$n$ is minimal.
\end{definition}

We have Theorem~\ref{thm:twoFuncLower} for a lower bound for the two
function problem. This lower bound exactly matches the upper bound
shown in Theorem~\ref{thm:twoFuncUpper}.

\begin{theorem}\label{thm:twoFuncLower}
For each $l$, every solution for the two function problem needs
$n\ge l$.
\end{theorem}

\begin{proof}
Assume that we have a solution for the two function problem. Let $f
: \{0, 1,\ldots,l\}\rightarrow \{0,1\}^n$ and $g : \{0,
1,\ldots,l\}\rightarrow \{0,1\}^n$ such that $f(x)\ AND\ g(y) = 0^n$
if and only if $x + y \le l$.

For a string $s$ in $\{0,1\}^n$, define $C(s)$ to be the number of
$1$s in $s$.  For example, $C(010011)=3$. Define $D(f)=\sum_{i=0}^l
C(f(i))$. For each string $s=a_1a_2\ldots a_n$ in $\{0,1\}^n$,
define $G(s)=\{i: a_i=1\ and\ 1\le i\le n\}$. For each string
$s=a_1a_2\ldots a_n$ in $\{0,1\}^n$, define $s[i]=a_i$.

Let the solution $f$ and $g$ satisfy that $D(f)$ is the largest for
the least $n$. We claim that $G(f(0))\subseteq G(f(1))\subseteq
G(f(2))\subseteq \ldots \subseteq G(f(l))$.

Assume that there is a $m\in \{0,1,\ldots, l\}$ with
$G(f(m))\not\subseteq G(f(m+1))$. Let $i$ be the index such that
$f(m)[i]=1$ and $f(m+1)[i]=0$. Let $f(m)=a_1a_2\ldots a_n$ and
$f(m+1)=b_1b_2\ldots b_n$. Define function $f'$ as follows:
$f'(k)=f(k)$ for each $k\not= m+1$, and $f'(m+1)=b_1\ldots
b_{i-1}a_ib_{i+1}\ldots b_n=b_1\ldots b_{i-1}1b_{i+1}\ldots b_n$. We
claim that $f'$ and $g$ form a new solution for the two function
problem.

For each $y$ with $(m+1)+y\le l$, we have $m+y<l$. Therefore, $f(m)\
AND\ g(y)=0^n$. This implies $a_i \cdot g(y)[i]=0$. Furthermore, we
also have $f(m+1 )\ AND\ g(y)=0^n$. Therefore,
\begin{eqnarray*}
C(f'(m+1 )\ AND\ g(y))&\le& C(f(m+1 )\ AND\ g(y))+a_i \cdot
g(y)[i]\\
&\le& 0+0\\
&=&0.
\end{eqnarray*}
 Thus, $C(f'(m+1 )\ AND\ g(y))=0$. Therefore,
\begin{eqnarray*}
f'(m+1 )\ AND\ g(y)=0^n.
\end{eqnarray*}

For each $y$ with $(m+1)+y> l$, we have $f(m+1 )\ AND\
g(y)\not=0^n$. Since $a_i=1$, we have
\begin{eqnarray*}
C(f'(m+1 )\ AND\ g(y))&\ge& C(f(m+1 )\ AND\ g(y))\\
&\ge& 1.
\end{eqnarray*}
 Therefore,
\begin{eqnarray*}
f'(m+1 )\ AND\ g(y)\not =0^n.
\end{eqnarray*}
 Combining the last two cases, we have
that $f'$ and $g$ form a new solution for the two function problem.
We have $D(f)<D(f')$. This contradicts the fact that $D(f)$ is the
largest for the same $n$.

Therefore, the solution $f$ and $g$ has that $G(f(0))\subseteq
G(f(1))\subseteq G(f(2))\subseteq \ldots \subseteq G(f(l))$. It is
easy to see $f(i)\not= f(j)$ for $0\le i<j\le l$ because $f(i)\ AND\
g(l-i)=0^n$ and $f(j)\ AND\ g(l-i)\not=0^n$. Thus, we have that
$G(f(i))$ is a proper subset of $G(f(i+1))$. Therefore, $n\ge l$.
\end{proof}

\begin{theorem}\label{thm:twoFuncUpper}
The two function problem always has a solution with $n=l$.
\end{theorem}

\begin{proof}
Let $f(i)=1^i0^{l-i}$ and $g(l-i)=0^i1^{l-i}$. It is easy to verify
that $f(x)\ AND\ g(y) = 0^l$ if and only if $x + y \le l$.
\end{proof}

\section{Matrix Problem}\label{sec:matrix}

In this section, we study the complexity of a matrix problem that is related to tile
geometry design for compatibility specifications. An efficient solution
for this problem can achieve reductions in geometry sizes.

Called the {\it matrix problem}, the goal is generating a
binary matrix via vector inner products. Given a $0-1$ matrix $M$
of size $m\times n$, find a list of vectors $E_1,\ldots, E_m$ and a
list vectors $W_1,\ldots, W_n$ of length $l$ for each vector such
that $M(i,j)=0$ if and only if $E_i\cdot W_j=0$, where $\cdot$ is
the inner product of two vectors. The target is to minimize the
length $l$, which is denoted by $L(M)=l$ for the least $l$. We show
that for almost all $n\times n$ binary matrices $M$,
$L(M)=\Theta(n)$. A submatrix $M_i$ of $M$ is characterized by
$(R_i,C_i)$ where $R_i$ is a subset of row indices of $M$ and $C_i$
is a subset of column indices of $M$. In the case that all $1$
entries are in submatrices $M_1=(R_1, C_1),\cdots, M_t=(R_t, C_t)$
that satisfy $R_i\cap R_j=\emptyset$ and $C_i\cap C_j=\emptyset$ for
$i\not=j$, we show the equality $L(M)=L(M_1)+\cdots+L(M_t)$. This
relationship gives a tool for solving a class of concrete matrix
problems. For another case that all $0$ entries of $M$ are in
submatrices $M_1=(R_1, C_1),\cdots, M_t=(R_t, C_t)$ that satisfy
$R_i\cap R_j=\emptyset$ and $C_i\cap C_j=\emptyset$ for $i\not=j$,
we derive the lower bound and upper bound $\max(L(M_1),\cdots,
L(M_t))\le L(M)\le (1+\epsilon)\max(L(M_1),\cdots, L(M_t))+O({\log
n\over \epsilon})$ for an arbitrary constant $\epsilon>0$ using a
randomized algorithm. This randomized algorithm can be used to find
a matrix design with a small length $l$. The dimension of each
submatrix $M_i$ is $m_i\times n_i$ in table~\ref{table:summary}.

For two lists $E_1,\ldots, E_m$ and $W_1,\ldots, W_n$ such that each
$E_i$ or $W_i$ are of length $l$, they are called a {\it $l$-list
pair}.

{\it Matrix Problem:} Given a $m\times n$ $0-1$ matrix $M$, find a
list vectors $E_1,\ldots, E_m$ and a list vectors $W_1,\ldots, W_n$
of length $l$ for each vector such that $M(i,j)=0$ if and only if
$E_i\cdot W_j=0$, where $\cdot$ is the inner product of two vectors
such that $E_i\cdot W_j=\sum_{k=1}^l e_kw_k$ for $E_i=e_1e_2\cdots
e_l$ and $W_j=w_1w_2\cdots w_l$. Minimize the length $l$.

For an integer  matrix $M_{m\times n}$, its rank is the number of
linear independent rows. A matrix $M_{n\times n}$ is {\it diagonal
one} matrix if all the diagonal elements are one and all other
elements are zero. A matrix $M_{n\times n}$ is {\it diagonal zero}
matrix is all diagonal elements are zero and all other elements are
one.

\subsection{Some Lower Bounds}

In this section, we show some results for the lower bound of the
matrix problem. Most of the lower bounds results are for concrete
matrix problems.

\begin{definition}
For a binary matrix $M$, define $L(M)$ to be the least length $l$
for a solution of the matrix problem $M$.
\end{definition}

\begin{definition}\label{matrix-basic-def}
 Assume that $M=(a_{i,j})_{m\times n}$ is a binary matrix.
\begin{itemize}
\item
Define $S(M)=\min(m, n)$.
\item
A {\it submatrix} $M'=(R, C)$ of $M$ is characterized by a set $R$
of rows and a set $C$ of columns in $M$. The elements of $M'$ are
all elements $a_{i,j}$ with $i\in R$ and $j\in C$. We also define
$S(M')=\min(|R|, |C|)$.
\item
Two submatrices $M_1=(R_1, C_1)$ and $M_2=(R_2, C_2)$ of $M$ are
{\it independent} if $R_1\cap R_2=\emptyset$ and $C_1\cap
C_2=\emptyset$.
\item
A set $H$ of submatrices of $M$ is {\it independent} if every two of
them are independent.
\item
Let $H=\{M_1,\ldots, M_t\}$ be a set of independent submatrices,
define $T(H)=\max(S(M_1),\ldots, S(M_t))$.
\item
For a sequence $s=a_1a_2\ldots a_n$, and a subset $P=\{i_1,\ldots,
i_m \}$ with $i_1<i_2<\ldots<i_m$ of integers of $\{1,2,\ldots,
n\}$, define $s[P]=a_{i_1}a_{i_2}\ldots a_{i_m}$. For an integer
$i\le n$, define $s[i]=a_i$.
\end{itemize}
\end{definition}

\begin{theorem}\label{thm:indSubMatrices}
Let $M_1,\ldots, M_t$ be independent submatrices of $M_{m\times n}$
and all $1$s of $M$ are in $M_1,\ldots, M_t$. Then
$L(M)=L(M_1)+\ldots+L(M_t)$.
\end{theorem}

\begin{proof}
Assume that $E_1,\ldots, E_m$ and $W_1,\ldots, W_n$  form a solution
of least length $l$ for the matrix $M=(a_{i,j})_{m\times n}$. Let
$M_i=(R_i, C_i)$ for $i=1,\cdots, t$. Define $P_i=\{j:$ for some
$(u,v)\in R_i\times C_i$ with $E_u[j]=W_v[j]=1\}$.

{\bf Claim 1.} $E_{r_1}[P_i],\ldots, E_{r_y}[P_i]$ and
$W_{c_1}[P_i],\ldots, W_{c_z}[P_i]$ form a solution for $M_i$, where
$R_i=\{r_1,\ldots, r_y\}$ and $C_i=\{c_1,\ldots, c_z\}$.

\begin{proof}
By the definition of $P_i$, for each $(u,v)\in R_i\times C_i$, we
have $E_u[x]=0$ or $W_v[x]=0$ for any $v\not\in P_i$ (otherwise, $x$
is in $P_i$). Thus, $E_u\ AND\ W_v =0^l$ if and only if $E_u[P_i]\
AND\ W_v[P_i] =0^{l_i}$, where $l_i=|P_i|$.
\end{proof}

{\bf Claim 2.} $P_i\cap P_j=\emptyset$ for $i\not=j$.

\begin{proof}
This can be proved by contradiction. Assume $P_i\cap
P_j\not=\emptyset$. Let $x\in P_i\cap P_j$. There is a $(u_1,
v_1)\in R_i\times C_i$ such that $E_{u_1}[x]=W_{v_1}[x]=1$. There is
a $(u_2, v_2)\in R_j\times C_j$ such that $E_{u_2}[x]=W_{v_2}[x]=1$.
Therefore, $(E_{u_1}\ AND\ W_{v_2})$ contains a bit $E_{u_1}[x]
\cdot W_{v_2}[x]= 1$. Since $M_1,\ldots, M_t$ are independent,
$u_1\in R_i$ and $v_2\in C_j$, we have $a_{u_1, v_2}=0$. This
contradicts that $(E_{u_1}\ AND\ W_{v_2})=0^l$ if and only if
$a_{u_1,v_2}=0$.
\end{proof}

{\bf Claim 3.} $L(M)\ge L(M_1)+\ldots+L(M_t)$.

\begin{proof}
By Claim 1, since $E_1[P_i],\ldots, E_m[P_i]$ and $W_1[P_i],\ldots,
W_n[P_i]$ form a solution for the matrix problem $M_i$ for
$i=1,\cdots, t$, we have $L(M_i)\le |P_i|$. By Claim 2, we have
$|P_1|+\ldots +|P_t|\le l=L(M)$. Thus, $L(M)=|P_1|+\ldots +|P_t|\ge
L(M_1)+\ldots+L(M_t)$.
\end{proof}

Assume that $E_{i,r_1},\ldots, E_{i,r_y}$ and $W_{i,c_1},\ldots,
W_{i,c_z}$ form a solution for $M_i=(R_i,C_i)$ with
$R_i=\{r_1,\ldots, r_y\}$ and $C_i=\{c_1,\ldots, c_z\}$ for
$i=1,\ldots, t$. The length of each vector for $M_i$ is
$l_i=L(M_i)$.

For $j=1,\ldots, m$, define $E'_j=A_1\ldots A_t$ such that
$A_i=0^{l_i}$ if $i\not\in R_1\cup \ldots R_t$, and $A_i=E_{u,r_h}$
if $i=r_h\in R_u$. For $j=1,\ldots, m$, define $W'_j=B_1\ldots B_t$
such that $B_i=0^{l_i}$ if $i\not\in C_1\cup \ldots C_t$, and
$B_i=W_{u,c_h}$ if $i=c_h\in C_u$.

{\bf Claim 4.} $E'_1,\ldots, E'_m$ and $W'_1,\ldots, W'_n$ form a
solution for $M$.

\begin{proof}
Let $l=l_1+\ldots+l_t$.
Let's consider a position $(i,j)$ in matrix $M$. Let $E'_j=A_1\ldots
A_t$ and $W'_j=B_1\ldots B_t$. We discuss the following two cases.

\begin{itemize}
\item
Case 1. $(i,j)\not\in R_u\times C_u$ for all $u=1,\ldots, t$.  For
each $g\le t$, we have either $A_g=0^{l_g}$ or $B_g=0^{l_g}$ since
the submatrices $M_1,\ldots, M_t$ are independent. Therefore, $E'_i\
AND\ W'_j=0^l$.

\item
Case 2. $(i,j)\in R_u\times C_u$.  Since the submatrices
$M_1,\ldots, M_t$ are independent, we have $(i,j)\not\in R_v\times
C_v$ for $v\not= u$. We have $A_v\ AND\ B_v=0^{l_v}$ for each
$v\not= u$.Thus, $E'_i\ AND\ W'_j\not=0^l$ if and only if $A_u \ AND
\ B_u\not=0^{l_u}$.
\end{itemize}

\end{proof}

By Claim 4, we have $L(M)\le L(M_1)+\ldots+L(M_t)$. By Claim 3, we
have $L(M)= L(M_1)+\ldots+L(M_t)$.
\end{proof}

\begin{corollary}\label{cor:diagOne}
There minimum length for the $n\times n$ diagonal-one matrix problem
is exactly $n$.
\end{corollary}







\begin{theorem}\label{thm:mostMatrices}
For most of $n\times n$ matrices, the minimum solution is at least
${n\over 2}-n^{1-\epsilon}$ for any fixed $\epsilon>0$.
\end{theorem}

\begin{proof}
Let $\epsilon$ be an arbitrary positive constant. There are totally
$2^{n^2}$ many $n\times n$ binary matrices $M$. Assume $m\le {n\over
2}-n^{1-\epsilon}$. We consider how many matrices can be constructed
by using $l$-list pairs with $l\le m$.

 The total number of bits of a
$l$-list pair is $2nl$. The total number of $l$-list pairs is
$2^{2nl}$. The total number of $l$-list pairs with $l\le m$ is
$\sum_{l=1}^m 2^{2nl}<2\cdot 2^{2nm}=o(2^{n^2})$ since $m\le {n\over
2}-n^{1-\epsilon}$. Therefore, for most of $n\times n$ matrices,
there is no ${n\over 2}-n^{1-\epsilon}$ solution.

\end{proof}

\begin{theorem}\label{full-rank-theorem}
For the minimum length solution for a full rank matrix is greater
than $\log n$.
\end{theorem}

\begin{proof}
Assume that there is a $l$-list pair $E_1,\ldots, E_m$ and
$W_1,\ldots, W_n$ solution for an $n\times n$ matrix $M$ of rank
$n$. If there are two different $i\not=j$ with $E_i=E_j$, the $i$-th
row and $j$-row are the same. Thus, the matrix is not full rank. If
$l=\log n$, then there is $E_i$ to be all zeros. Thus the matrix $M$
has a row to be all zeros. This makes $M$ not to be full rank.
Therefore, $l>\log n$.

\end{proof}

\begin{theorem}\label{log-lower-theorem}
For the minimum length solution for the $n\times n$ diagonal zero
matrix is at greater than $\log n$.
\end{theorem}

\begin{proof}
The absolute value of the determinant of diagonal zero matrix is
$(n-1)$. Thus, it is full rank matrix. It follows from
Theorem~\ref{full-rank-theorem}.
\end{proof}

\begin{theorem}\label{upper-bound-theorem}
Assume that $k$ is an integer such that there is another integer
$1<l_1<l$ such that ${l \choose l_1}\ge n$. Then there is a $l$-list
pair solution for the diagonal zero matrix problem.
\end{theorem}

\begin{proof}
Let $E_1,\cdots, E_n$ be $n$ different binary string of length $l$
with exactly $l_1$ ones each. Let $W_i$ be the complementary binary
string of $E_i$. It is easy to see that $E_i\cdot W_i=0$ and
$E_i\cdot W_j>0$ for $i\not =j$.
\end{proof}

\begin{corollary}\label{cor:diagZero}
The minimum length for diagonal zero $n\times n$ matrix is between
$\log n+1$ and $\log n+\log\log n$.
\end{corollary}

\begin{proof}
Assume that $l$ is an even even number $\le \log n+\log\log n$. Let
$l_1=l/2$. By Stirling formula ${n!\over \sqrt{2\pi n}(({n\over
e})^n)}=1$, we have ${l\choose l_1}\sim {2^l\sqrt{2}\over\sqrt{\pi
l}}$. Thus, we can pick a $l\le \log n+\log\log n$ such that ${l
\choose l_1}\ge n$. The $n\times n$ diagonal zero matrix is of rank
$n$. The lower bound $\log n+1$ follows from
Theorem~\ref{log-lower-theorem}.

\end{proof}

\subsection{Upper Bounds and Algorithm for Matrix Problem}

In this section, we show a randomized algorithm to handle a class of
matrix problems. A matrix $M_{m\times n}$ has a list of independent
submatrices $M_1,\cdots, M_t$ that contain all zero entries of $M$.
We derive an upper bound of the solution for $L(M)$ to be close to
$\max(L(M_1),\cdots, L(M_t))$. This implies the interesting bound
$\max(L(M_1),\cdots, L(M_t))\le L(M)\le
(1+\epsilon)\max(L(M_1),\cdots, L(M_t))+O({\log (m+n)\over
\epsilon})$ for an arbitrary positive constant $\epsilon$.

We believe the following simple algorithm in
Lemma~\ref{rand-perm-lemma} for a random permutation is not new. For
completeness, we include it here.

\vskip 10pt

{\bf RandomPermutation($n$)}

Let $S=\{1,2,\cdots, n\}$;

For $i$ from $1$ to $n$

\qquad Select a random $a_i$ element from $S$

\qquad Let $S=S-\{a_i\}$;

Output $a_1a_2\cdots a_n$;

 {\bf End of RandomPermutation}

\vskip 10pt

\begin{lemma}\label{rand-perm-lemma}
Every permutation of  $1,2,\ldots, n$ has a equal probability to be
generated by RandomPermutation(.) that runs in $O(n^2)$ time.
\end{lemma}

\begin{proof}
Assume that $a_1a_2\cdots a_n$ be an arbitrary permutation. We just
need to prove that $P_1=a_1a_2\cdots a_ia_{i+1}\cdots a_n$ and
$P_2=a_1a_2\cdots a_{i+1}a_i\cdots a_n$ have the equal chance to be
generated. This is because a permutation can be converted into
another permutation via a finite number of swaps between two
neighbor items. Note that $P_1$ can be converted into $P_2$ by
swapping the two elements $a_i$ and $a_{i+1}$. Assume that the
partial permutation $a_1a_2\cdots a_{i-1}$ have been generated by
RandomPermutation(.). We have that $a_ia_{i+1}$ and $a_{i+1}a_i$
will be generated by a equal probability RandomPermutation(.) to
append to the last partial permutation $a_1a_2\cdots a_{i-1}$. The
computational time follows from that fact that it takes $O(n)$ time
to generate one element and update the set $S$ in the algorithm.
\end{proof}

\begin{lemma}\label{probability-lemma}
Assume that $S_1$ is a subset of  $k_1$ elements of $\{1, 2,\cdots,
n\}$. Let $S_2$ be a subset of $\{1,2,\cdots,n\}$ and of size $k_2$
with $k_2\ge \alpha n$ for some $\alpha\in (0,1)$. Then with
probability at most $p\le (1-\alpha)^{k_1}$, $P[S_2]\cap
S_1=\emptyset$ for a random permutation $P$ of $1,2,\cdots,n$, where
$P[S_2]$ is the set of the positions of elements $S_2$ in $P$.
\end{lemma}

\begin{proof} For two equal size subsets $S$ and $S'$ of $\{1, 2,\cdots, n\}$,
the number of permutations $P$ that make $P[S_2]\cap S=\emptyset$ is
the same as the number of permutations $P$ that make $P[S_2]\cap
S'=\emptyset$. By Lemma~\ref{rand-perm-lemma}, the probability for
$P[S_2]\cap S=\emptyset$ is the same as the probability for
$P[S_2]\cap S'=\emptyset$, where $P$ is a random permutation  of
$1,2,\cdots,n$. Thus, we can assume that $S_1$ is $\{1, 2,\cdots,
k_1\}$.

The probability that none of the elements of $S_2$ are selected in
the first $k_1$ positions is  $$p={n-k_2\over n}\cdot {n-k_2-1\over
n-1}\ldots {n-k_1-k_2-1\over n-k_1-1}.$$

Consider the function $f(x)={n-k_2-x\over n-x}=1-{k_2\over n-x}$. We
have its derivative function ${\partial f(x)\over
\partial x}=-{k_2\over (n-x)^2}<0$. Thus,  $f(x)$ is a decreasing function. We
have ${n-k_2-i\over n-i}\le {n-k_2\over n}$. Thus,
\begin{eqnarray}
{n-k_2\over n}\cdot {n-k_2-1\over n-1}\ldots {n-k_1-k_2-1\over
n-k_1-1}&\le& ({n-k_2\over n})^{k_1}\\
&\le& (1-\alpha)^{k_1}.
\end{eqnarray}

\end{proof}

\begin{lemma}\label{basic-design-lemma}
There is algorithm that given a binary matrix $M_{m\times n}$, it
gives a solution for the generalized matrix problem with $S(M)$
length in $O(mn)$ times, where $S(M)=\min(m,n)$ as in
Definition~\ref{matrix-basic-def}.
\end{lemma}

\begin{proof} Let $M=(a_{i,j})_{m\times n}$.
Without loss of generality, assume $n\le m$. For each column $k$,
let $W_k=0^{k-1}10^{n-k}$. For each row $j$, let $E_j=e_{j,1}\ldots
e_{j,n}$, where $e_{j,k}=1$ if and only if $a_{j,k}=1$.
\end{proof}








\begin{definition}
For a string $s$, define a padding function $pad_1(s, k_1, k_2,
k_3)=0^{k_1}1^{k_2}s1^{k_3-|s|}$ if $|s|\le k_3$, and $pad_1(s, k_1,
k_2, k_3)$ is the empty string otherwise. We also define another
padding function $pad_2(s, k_1, k_2,
k_3)=1^{k_1}0^{k_2}s1^{k_3-|s|}$ if $|s|\le k_3$, and $pad_2(s, k_1,
k_2, k_3)$ is the empty string otherwise. For a permutation
$p=i_1\ldots i_n$ of  $1,\ldots, n$, and a binary string
$s=a_1\ldots a_n$ of length $n$, define $P(s,p)=a_{i_1}\ldots
a_{i_n}$.
\end{definition}

\begin{lemma}\label{convert-lemma} Let $\epsilon$ be an arbitrary constant in
$(0,1)$.  Then there is an $O((m+n)^3)$ time randomized algorithm
such that given a binary matrix $M_{m\times n}$,  a set
$H=\{M_1,\cdots, M_t\}$ of independent submatrices that  all zeros
entries of  $M$ are in the submatrices of $H$, and also a solution
with length $l_i$ for each $M_i\in H$, it produces a solution for
$M$ with length $(1+\epsilon)U(H)+O({\log (n+m)\over \epsilon})$,
where $U(H)=\max\{l_i: M_i\in H\}$ is the largest length of the
solution among all submatrices in $H$.
\end{lemma}

\begin{proof}
Let $k_1=k_2={\epsilon\over 2} U(H)+{c\over \epsilon}\log (m+n)$ and
$k_3=U(H)$, where constant $c$ is selected such that
\begin{eqnarray}
(1-{\epsilon\over 4})^{{c\over \epsilon}\log(n+m)}\le {1\over
4mn}.\label{c-eqn}
\end{eqnarray}
The total length for designing the matrix solution is
$z=k_1+k_2+k_3$. We design $E_1,\cdots, E_m$ and $W_1,\cdots, W_n$
of  length $z$ each for a solution for the matrix problem $M$. By
Lemma~\ref{basic-design-lemma}, we assume $U(H)\le \min(m,n)$.

For each submatrix in $M_u=(R_u,C_u)$ in $H$ with $R_u=\{i_1,\ldots,
i_s\}$ and $C_u=\{j_1,\ldots, j_t\}$, let $M_u$ have a solution
$E_{i_1}^{(u)},\ldots, E_{i_s}^{(u)}$ and $W_{j_1}^{(u)},\ldots,
W_{j_t}^{(u)}$, which has length $l_u=|E_{i_1}^{(u)}|\le U(H)$. Let
$p_u$ be a random permutation $1,2,\ldots, z$. Now let
$E_{i_k}=P(pad_1(E_{i_k}^{(u)}),p_u)$ for $k=1,\ldots, s$ and let
$W_{j_k}=P(pad_2(W_{j_k}^{(u)}),p_u)$ for $k=1,\ldots, t$. Thus, if
row $i$ and column $j$ have $(i,j)\in R_u\times C_u$ for some
$M_u=(R_u,C_u)\in H$, we have  $E_i^{(u)}\ AND\ W_j^{(u)}=0^z$ if
and only if $P(pad_1(E_{i}^{(u)}),p_u)\ AND
P(pad_2(W_{j}^{(u)}),p_u)=0^z$  if and only if $E_i\ AND\ W_j=0^z$
if and only if $a_{i,j}=0$.

For each row $i$, if $i\not\in R_k$ for any $M_k=(R_k,C_k)\in H$,
then let  $E_i=1^z$. Similarly, for each column $j$, if $j\not\in
C_k$ for any $M_k=(R_k,C_k)\in H$, then let  $W_j=1^z$. Thus, for
row $i$ and column $j$, if $i\not\in R_k$ for any $M_k=(R_k,C_k)\in
H$, we always have $E_i\ AND\ W_j\not= 0^z$. Similarly, if $j\not\in
C_k$ for any $M_k=(R_k,C_k)\in H$, we always have $E_i\ AND\
W_j\not= 0^z$.

Let
\begin{eqnarray}
\alpha&=&{k_1\over z}\label{alpha-init-ineqn}\\
&=&{k_1\over 2k_1+U(H)}\\
 &=&{{\epsilon\over 2} U(H)+{c\over \epsilon}\log (m+n)\over
2({\epsilon\over 2} U(H)+{c\over \epsilon}\log
(m+n))+U(H)}\\
&\ge&{\epsilon U(H)+{2c\over \epsilon}\log (m+n)\over 2\epsilon
U(H)+2U(H)+{4c\over \epsilon}\log
(m+n))}\\
&\ge&{\epsilon U(H)+{2c\over \epsilon}\log (m+n)\over 4U(H)+{4c\over \epsilon}\log (m+n))}\\
&\ge&{\epsilon U(H)+\epsilon\cdot {c\over \epsilon}\log (m+n)\over
4U(H)+{4c\over \epsilon}\log
(m+n))}\\
&=&{\epsilon\over 4}.\label{alpha-end-ineqn}
\end{eqnarray}

For an entry $(i,j)$,  we consider the case that $i\in R_u$ in
$M_u=(R_u, C_u)$, and $j\in C_v$ in $M_v=(R_v, C_v)$ for some
$u\not=v$.
 For a sequence $s=a_1\ldots a_z$, let
$Q_1(s)$ be the the set of positions with bit $1$ in $s$
($Q_1(s)=\{i:a_i=1\ and \ 1\le i\le z\}$). Let
$E_{i_1}^{(u)},\cdots, E_{i_s}^{(u)}$ and $W_{j_1}^{(u)},\cdots,
W_{j_t}^{(u)}$ be a solution for submatrix $M_u$ with length
$l_u=|E_{i_1}^{(u)}|\le U(H)$. Since $i\in R_u$,
$E_i=P(pad_1(E_{i_k}^{(u)}),p_u)$, where $p_u$ is a random
permutation of $1,2,\cdots,z$, and $i_k=i$. Let
$E_{x_1}^{(v)},\cdots, E_{x_a}^{(v)}$ and $W_{y_1}^{(v)},\cdots,
W_{y_b}^{(v)}$ be the solution for submatrix $M_v$ with length
$l_v=|E_{i_1}^{(v)}|\le U(H)$. Since $j\in C_v$,
$W_j=P(pad_2(W_{y_k}^{(v)}),p_v)$, where $p_v$ is a random
permutation of $1,2,\cdots,z$, and $y_k=j$. Since $u\not=v$, $p_u$
and $p_v$ are independent permutations.  $Q_1(E_i)$ contains $k_1$
positions of $\{1,2,\ldots, z\}$, and $Q_1(W_j)$ another $k_2$
positions of $\{1,2,\ldots, z\}$. By Lemma~\ref{probability-lemma},
with probability at most $(1-\alpha)^{k_2}$, $Q_1(E_i)\cap
Q_1(W_j)=\emptyset$. Therefore, with probability at most
$(1-\alpha)^{k_2}$, $E_i\ AND\ W_j=0^n$.

Therefore, with probability at most $p=nm(1-\alpha)^{k_2}$, there is
a position  $(i,j)$ in $M$ with $E_i\ AND \ W_j=0^z$ not to be
equivalent to $a_{i,j}=0$. By the choice of $c$ at
equation~(\ref{c-eqn}), and inequalities~(\ref{alpha-init-ineqn})
to~(\ref{alpha-end-ineqn}), we have small probability $p$ defined
 above
 to be at most ${1\over 4}$.

 The computational time follows from the fact that for each $E_i$ or
 $W_j$, we need at most $O((m+n)^2)$ time, which is spent for generating a random permutation, from the solutions of matrices
 in $H$.
\end{proof}

\begin{theorem}\label{thm:indSubMatrices}
 Let $\epsilon$ be an arbitrary constant in $(0,1)$.
Assume that $H$ is a set of independent submatrices of $M$ that all
zero elements in $M$ are in the submatrices of $H$, then $V(H)\le
L(M)\le (1+\epsilon)V(H)+O({\log (n+m)\over \epsilon})$, where
$V(H)=\max\{L(M_i): M_i\in H\}$.
\end{theorem}

\begin{proof}
It is trivial to see the inequality $V(H)\le L(M)$ since a solution
for $M$ automatically implies a solution for its submatrix. The
inequality $L(M)\le (1+\epsilon)V(H)+O({\log (n+m)\over \epsilon})$
follows from Lemma~\ref{convert-lemma}.
\end{proof}

\begin{theorem}\label{thm:indSubMatrices}
 Let $\epsilon$ be an arbitrary constant in $(0,1)$.
There is a randomized algorithm such that given a binary matrices
$M_{m\times n}$ and a set $H$ of independent submatrices of $M$, if
all zero elements $M$ are in the submatrices of $H$, then the
algorithm returns a design with total length at most
$(1+\epsilon)T(H)+O({\log (n+m)\over \epsilon})$. Furthermore, the
time complexity of the algorithm is $O((n+m)^3)$.
\end{theorem}


\begin{proof}
It follows from Lemma~\ref{basic-design-lemma} and
Lemma~\ref{convert-lemma}.
\end{proof}

\begin{lemma}\label{brute-lemma}
There is a $O(2^{L(M)\min(m,n)}(m+n)^3)$ time algorithm to find an
optimal solution of length $L(M)$ for a binary matrix $M_{m\times
n}$.
\end{lemma}

\begin{proof}
By Lemma~\ref{basic-design-lemma}, it has a solution of length at
most $\min(m,n)$. Assume $m\le n$. We can find the least length by
trying all numbers at most $L(M)$. Each number takes at most
$O(2^{L(M)m})$ time. After vectors $E_1,\cdots, E_m$ are fixed, it
takes $O(mn)$ time to derive each $W_i$. This can be done to put the
least number of zeros in the $W_i$ to satisfy the zero entries of
column $i$ of $M_{m\times n}$. Thus, it takes $O(mn^2)$ time to
derive all the columns after fixing rows $E_1,\cdots, E_m$.
Therefore, the total time is $O(2^{L(M)}(m+n)^3)$.
\end{proof}

\begin{theorem}\label{thm:indSubMatrices}
 Let $\epsilon$ be an arbitrary constant in $(0,1)$.
Then there is a randomized algorithm such that given a binary
matrices $M_{m\times n}$ and a set $H$ of independent submatrices of
$M$, if all zero elements $M$ are in the submatrices of $H$, then
the algorithm returns a design with total length at most
$(1+\epsilon)V(H)+O({\log (n+m)\over \epsilon})$, where
$V(H)=\max\{L(M_i): M_i\in H\}$. Furthermore, the time complexity of
the algorithm is $O(2^{V(H)\min(m,n)}(n+m)^{3})$.
\end{theorem}

\begin{proof}
It follows from Lemma~\ref{brute-lemma} and
Lemma~\ref{convert-lemma}.
\end{proof}

\subsection{Revised Matrix Problem}

In this section, we revise the definition of the matrix problem, and
show that how the optimal length of solution depends on the rank of
the matrix.

{\it Revised Matrix Problem:} Given a $n\times n$ nonnegative
integer matrix $M$, find a list vectors $E_1,\ldots, E_n$ and a list
vectors $W_1,\ldots, W_n$ of length $l$ for each vector such that
$M(i,j)=E_i\cdot W_j$, where $\cdot$ is the inner product of two
vectors. Minimize the length $l$.

\begin{theorem}\label{revised-theorem}
Let $M$ be a $n\times n$ matrix with nonnegative elements. The
minimum length solution for the revised matrix problem $M$ is at
least $rank(M)$.
\end{theorem}

\begin{proof}
Assume that there is a $l$-list pair $E_1,\ldots, E_n$ and
$W_1,\ldots, W_n$ solution $M$ such that $M(i,j)=E_i\cdot M_j$.
Define $M_t$ to be the matrix such that $M_t(i,j)=E_i[t]W_j[j]$.
$M_t$ is generated by the $t$-th bit of the $l$-list pair. We have
$M=M_1+M_2+\ldots+M_l$. It is easy to see that rank of each $M_t$ is
at most one.

By the  well known fact in the linear algebra, we have $\rank(M)\le
\rank(M_1)+\rank(M_2)+\ldots+\rank(M_l)\le l$. Therefore,
$\rank(M)\le l$.
\end{proof}

\begin{corollary}
The minimum length solution for the full rank revised  matrix with
matrix size $n\times n$ problem is $n$.
\end{corollary}

\begin{corollary}
The minimum length solution for the revised matrix problem with
matrix size $n\times n$  for diagonal one matrix is $n$.
\end{corollary}

\begin{corollary}\label{cor:diagZero2}
The minimum length solution for the revised matrix problem for
diagonal zero $n\times n$ matrix is $n$.
\end{corollary}

Corollary~\ref{cor:diagZero2} is in contrast to
Corollary~\ref{cor:diagZero}.

\begin{proof}
The determinant of a diagonal zero matrix is $n-1$. It is full rank
matrix. It follows from Theorem~\ref{revised-theorem}.

\end{proof}

\end{document}